\documentclass[sigconf,screen,nonacm]{acmart}
\AtBeginDocument{%
  }

\usepackage{mathtools}
\usepackage{stmaryrd}
\mathtoolsset{showonlyrefs=true}
\usepackage{tikz,pgfplots}
\usetikzlibrary{cd}
\tikzset{%
	symbol/.style={%
		draw=none,
		every to/.append style={%
			edge node={node [sloped, allow upside down, auto=false]{$#1$}}}
	}
}
\usepackage{mathpartir}
\usepackage{ifthen}
\usepackage[bibliography=common]{apxproof}

\usepackage{macros.basic}
\renewcommand{\category}[1]{\mathbb{#1}}
\usepackage{macros.language}
\usepackage{macros.notation}
\usepackage{macros.aux}

\AtEndPreamble{%
\theoremstyle{acmdefinition}

\newtheorem{remark}[theorem]{Remark}}

\begin{document}

\title{On Complete Categorical Semantics for Effect Handlers}

\author{Satoshi Kura}
\email{satoshikura@acm.org}
\orcid{0000-0002-3954-8255}
\affiliation{%
  \institution{Waseda University}
  \city{Tokyo}
  \country{Japan}
}

\begin{abstract}
	Soundness and completeness with respect to equational theories for programming languages are fundamental properties in the study of categorical semantics.
	However, completeness results have not been established for programming languages with algebraic effects and handlers, which raises a question of whether the commonly used models in the literature, i.e., free model monads generated from algebraic theories, are the only valid semantic models for effect handlers.
	In this paper, we show that this is not the case.
	We identify the precise characterizations of categorical models of effect handlers that allow us to establish soundness and completeness results with respect to a certain equational theory for effect handling constructs.
	Notably, this allows us to capture not only free monad models but also the CPS semantics for effect handlers as models of the calculus.
\end{abstract}



\maketitle

\section{Introduction}
In categorical semantics of programming languages and lambda calculi, soundness and completeness are criteria of whether a given definition of models precisely captures the intended class of models for a calculus.
Soundness ensures that syntactically derivable equations $M = N$ between two terms of the calculus are semantically valid, meaning that the definition of models imposes necessary conditions on models.
This property is usually regarded as a minimal requirement for a valid definition of models.
Conversely, completeness ensures that all semantically valid equations between terms are syntactically derivable, meaning that the definition of models does not impose too strong restrictions on models.
Although completeness results are not necessarily as common as soundness results, they have been established for various calculi, such as the simply typed lambda calculus \cite{Lambek1986} and the computational lambda calculus \cite{MoggiLICS1989}, among many others.

This fundamental question of soundness and completeness has received little attention for programming languages with algebraic effects~\cite{PlotkinFoSSaCS2001,PlotkinApplCategStruct2003} and handlers~\cite{PlotkinESOP2009,PlotkinLMCS2013}.
Based on the fact that many computational effects can be represented by algebraic theories consisting of operations and equational axioms, handlers for algebraic effects provide a modular mechanism to implement those operations.
The de facto standard categorical semantics for effect handler calculi is given by free model monads generated from algebraic theories, or equivalently from Lawvere theories \cite{PlotkinESOP2009,PlotkinLMCS2013,AhmanPOPL2018}.
Free monads generated solely from operations are also used as models~\cite{BauerLMCS2014,ForsterJFunctProg2019,KammarICFP2013}, since enforcing such equational axioms on actual implementations of effect handlers is not very common.
Although these models admit soundness results with respect to a certain equational theory for effect handler calculi \cite[Section~7]{PlotkinESOP2009} \cite[Section~6]{BauerLMCS2014}, completeness results have not been established to the best of our knowledge.

The lack of completeness results implies that free model monads (or, in the absence of equational axioms, free monads) may not be the only valid semantic models for effect handlers.
Indeed, there seems to be an important example of models that are not free model monads, which is the CPS semantics for effect handlers \cite{HillerstromFSCD2017}.
For computational lambda calculus, it is known that the CPS transformation corresponds to the interpretation by the continuation monad \cite{FuhrmannInformationandComputation2004}.
By analogy, it is expected that the CPS semantics for effect handlers is modelled by a certain continuation-like monad.
To capture such models, we need to consider a more general notion of models for effect handlers.
Another motivation for generalizing the notion of models comes from the application of fibrational logical relations \cite{Hermida1993} to effect handlers.
The idea of fibrational logical relations is to construct logical relations as liftings of models in the base category $\category{B}$ to the total category $\category{E}$ along a fibration $p : \category{E} \to \category{B}$.
Several constructions of such liftings have been studied~\cite{KatsumataCSL2005,KatsumataPOPL2014,KammarMFPS2018} for lambda calculi without effect handlers.
To apply these techniques to effect handlers, imposing unnecessary restrictions on models may become an obstacle.

In this paper, we identify the precise characterizations of categorical models of effect handlers such that we can establish soundness and completeness results.
We consider two variants of effect handler calculi as targets of our study.
One is a basic calculus for deep effect handlers based on \cite{BauerLMCS2014}, which does not incorporate equational axioms from algebraic theories.
The other is an effect handler calculus based on \cite{LuksicJFunctProg2020}, which extends the former calculus by incorporating equational axioms.

We characterize models of these effect handler calculi by identifying the necessary and sufficient structure to interpret the effect handling construct $\handlewithto{M}{H}{x}{N}$.
For the effect handler calculus without equational axioms, the typing rule for this construct states that if (i) $M$ is a computation of type $A ! \Sigma$ where $A$ is a type of return values and $\Sigma$ is a signature for operations, (ii) $H$ is a handler for the operations in $\Sigma$ with return type $C$, and (iii) $N$ is a computation of type $C$ that takes the result of handling $M$ as $x : A$, then $\handlewithto{M}{H}{x}{N}$ is a computation of type $C$.
To give an interpretation of $\handlewithto{M}{H}{x}{N}$ from the interpretation of $M$, $H$, and $N$, we naturally need a morphism of the following form.
\[ \mathbf{handle}_{\Sigma, A, C} : T_{\Sigma} \interpret{A} \times \HBundle{\Sigma}{\interpret{C}} \times \interpret{A \to C} \to \interpret{C} \]
Here, we use a family of monads $\{ T_{\Sigma} \}_{\Sigma}$ indexed by signatures to interpret computation types $A ! \Sigma$ as $T_{\Sigma} \interpret{A}$, and $\HBundle{\Sigma}{\interpret{C}}$ denotes a suitable interpretation of handlers of type $\Sigma \Rightarrow C$.
Instead of using this morphism directly, we simplify it by eliminating $A$ and use the following family of morphisms indexed by $\Sigma$ and $C$ as the structure to interpret effect handlers:
\begin{equation}
	\mathbf{handle}_{\Sigma, C} : \HBundle{\Sigma}{\interpret{C}} \times T_{\Sigma} \interpret{C} \to \interpret{C} \label{eq:handle-structure}
\end{equation}
Then, we impose appropriate conditions on $\mathbf{handle}_{\Sigma, C}$ so that the interpretation of $\handlewithto{M}{H}{x}{N}$ satisfies desirable equations.
The resulting models of the effect handler calculus admit soundness and completeness results.
Notably, this notion of models allows us to capture not only free monad models but also the CPS semantics for effect handlers as models of the calculus.

We adopt a similar strategy for the effect handler calculus with equational axioms.
In this setting, computation types are extended to $A ! \Sigma / \mathcal{E}$ where $\mathcal{E}$ is a set of equational axioms for the operations in $\Sigma$, as in \cite{LuksicJFunctProg2020}.
Moreover, a computation of type $A ! \Sigma / \mathcal{E}$ can be handled only by handlers that respect the equations in $\mathcal{E}$.
To capture this situation, we need to modify the structure in~\eqref{eq:handle-structure} so that $\mathbf{handle}_{\Sigma, C}$ is defined only for handlers that respect the equations in $\mathcal{E}$.
This is technically more involved than the former calculus, but we can prove soundness and completeness results in this setting as well.

Although fibrational logical relations for effect handlers are currently left for future work, we expect that our results are a first step toward that direction.
Logical relations for effect handlers have been studied in \cite{BiernackiPOPL2018} using biorthogonal closure.
This suggests that fibrational logical relations can be constructed by adapting the $\top\top$-lifting \cite{KatsumataCSL2005} to effect handlers, which would require a flexible notion of models for effect handlers as we provide in this paper.

Our contributions are summarized as follows.
\begin{itemize}
	\item We give the first sound and complete categorical semantics for a simply typed calculus of deep effect handlers.
	\item We show that our semantics subsumes not only free monads but also CPS semantics for effect handlers.
	\item We extend the framework to effect handler calculus with equational axioms \cite{LuksicJFunctProg2020} and prove soundness and completeness for it.
\end{itemize}

The outline of the paper:
In Section~\ref{sec:preliminaries}, we review some preliminary notions about categorical semantics of lambda calculi with computational effects.
In Section~\ref{sec:syntax}, we introduce $\EffectHandlerCalculus$, a simply typed calculus of deep effect handlers without equational axioms, and present the equational theory for $\EffectHandlerCalculus$.
In Section~\ref{sec:semantics}, we define categorical models for $\EffectHandlerCalculus$ and prove soundness and completeness results.
In Section~\ref{sec:examples}, we show that free monad models and CPS semantics for effect handlers are instances of our models.
In Section~\ref{sec:with-effect-theories}, we extend the framework to effect handler calculus with equational axioms.
After discussing related work in Section~\ref{sec:related-work}, we conclude the paper in Section~\ref{sec:conclusion}.
Omitted proofs and details are provided in the appendix.

\section{Preliminaries}
\label{sec:preliminaries}
We introduce basic notions and notations for the semantics of the fine-grain call-by-value \cite{LevyInformationandComputation2003}, which is the basis of the calculus for effect handlers presented in Section~\ref{sec:syntax}.

\paragraph{Strong monads and strong Kleisli triples}
Let $\category{C}$ be a cartesian category.
A \emph{strong monad} on $\category{C}$ is a quadruple $(T, \eta^T, \mu^T, \strength^T)$ consisting of a functor $T : \category{C} \to \category{C}$ and natural transformations $\eta^T : \identity{\category{C}} \to T$ (unit), $\mu^T : T^2 \to T$ (multiplication), and $\strength^T_{X,Y} : X \times T Y \to T(X \times Y)$ (strength), satisfying the standard monad and strength laws.
We often say that $T$ is a strong monad, leaving the other components of the quadruple implicit.

It is known~\cite{McDermottMSFP2022} that a strong monad $(T, \eta^T, \mu^T, \strength^T)$ on a cartesian category $\category{C}$ is equivalent to a \emph{strong Kleisli triple} $(T, \eta^T, ({-})^{\sharp})$.
Since checking the laws for strong Kleisli triples is often easier than checking laws for $(T, \eta^T, \mu^T, \strength^T)$, we sometimes use strong Kleisli triples to define strong monads.

\begin{definition}\label{def:strong-kleisli-triple}
	A \emph{strong Kleisli triple} on a cartesian category $\category{C}$ is a triple $(T, \eta^T, ({-})^{\sharp})$ consisting of the following data:
	\begin{itemize}
		\item A mapping on objects $X \mapsto T X$.
		\item For each object $X$, a \emph{unit} morphism $\eta^T_X : X \to T X$.
		\item For each morphism $f : X \times Y \to T Z$, a morphism $f^{\sharp} : X \times T Y \to T Z$, called \emph{strong Kleisli extension} of $f$.
	\end{itemize}
	These data are required to satisfy the following equations:
	\begin{align}
		(\eta^T_X \comp \leftunitor)^{\sharp} &= \leftunitor &&: 1 \times T X \to T X \\
		f^{\sharp} \comp (X \times \eta^T_Y) &= f &&: X \times Y \to T Z \\
		g^{\sharp} \comp (W \!\times\! f^{\sharp}) \comp \associator &= (g^{\sharp} \comp (W \!\times\! f) \comp \associator)^{\sharp} \mkern-18mu &&: (W \times X) \times T Y \to T Z
	\end{align}
	where $\leftunitor_X : 1 \times X \to X$ and $\associator_{X, Y, Z} : (X \times Y) \times Z \to X \times (Y \times Z)$ are the canonical isomorphisms.
\end{definition}

\paragraph{Kleisli categories and Kleisli exponentials}
We write $\category{C}_T$ for the \emph{Kleisli category} of a monad $T$ on a category $\category{C}$.
There is a canonical adjunction $J^T \dashv K^T : \category{C}_T \to \category{C}$ associated with the Kleisli category.
Concretely, these functors are defined by $J^T X = X$ and $K^T X = T X$ for objects and $J^T f = \eta^T \comp f$ and $K^T g = \mu^T \comp T g$ for morphisms.
We sometimes omit superscripts $T$ in $J^T$ and $K^T$ when the monad $T$ is clear from the context.
Function types in fine-grain call-by-value are interpreted by \emph{Kleisli exponentials}.
\begin{definition}[Kleisli exponentials]
	Let $T$ be a strong monad on a cartesian category $\category{C}$.
	A \emph{Kleisli exponential} $\KleisliExp{T}{Y}{Z}$ is defined by the natural isomorphism $\category{C}_T(J(X \times Y), Z) \cong \category{C}(X, \KleisliExp{T}{Y}{Z})$.
	We write the \emph{currying} as $\Lambda : \category{C}(X \times Y, K Z) \cong \category{C}_T(J(X \times Y), Z) \cong \category{C}(X, \KleisliExp{T}{Y}{Z})$ and the \emph{evaluation morphism} as $\eval = \Lambda^{-1}(\identity{}) : (\KleisliExp{T}{X}{Y}) \times X \to K Y$.
\end{definition}
Note that if $\category{C}$ is cartesian closed, then an exponential $\exponential{Y}{T Z}$ from $Y$ to $T Z$ in $\category{C}$ gives a Kleisli exponential $\KleisliExp{T}{Y}{Z}$.

\paragraph{Models of the fine-grain call-by-value}
A \emph{$\lambda_c$-model} consists of a cartesian category $\category{C}$ with a strong monad $T$ on $\category{C}$ and Kleisli exponentials $\KleisliExp{T}{Y}{Z}$ for each $Y$ and $Z$.
The terminology reflects the fact that $\lambda_c$-models provide sound and complete semantics for Moggi's computational $\lambda$-calculus (also known as $\lambda_c$-calculus) \cite{MoggiLICS1989}.
Moreover, \citet{LevyInformationandComputation2003} showed that $\lambda_c$-models are also sound and complete for the fine-grain call-by-value.
They further showed that $\lambda_c$-models are equivalent to \emph{closed Freyd categories}, but we do not use this equivalence in this paper.
We use $\lambda_c$-models as a basis of our denotational semantics for effect handlers.

\section{An Effect Handler Calculus}
\label{sec:syntax}

Following~\cite{BauerLMCS2014,KammarICFP2013,HillerstromFSCD2017,LuksicJFunctProg2020}, we define a calculus of effect handlers, which we call $\EffectHandlerCalculus$.
To focus on the core ideas, we first provide a minimal calculus in this section and then discuss several extensions in Section~\ref{sec:calculus-extensions} and Section~\ref{sec:with-effect-theories}.

\subsection{Types and Terms}
Following the syntax of fine-grain call-by-value \cite{LevyInformationandComputation2003}, values and computations are distinguished in our calculus.
We define \emph{value terms} $V, W$; \emph{computation terms} $M, N$; and \emph{handlers} $H$ as follows.
\begin{align}
	V, W \quad&\coloneqq\quad x \mid \lambda x. M \mid \langle V_1, \dots, V_n \rangle \mid \pi_i\ V \\
	M, N \quad&\coloneqq\quad V\ W \mid \return{V} \mid \letin{x}{M}{N} \\
	&\qquad \mid \mathtt{op}(V) \mid \handlewithto{M}{H}{x}{N} \\
	H \quad&\coloneqq\quad \emptyset \mid \{ \mathtt{op}(x, k) \mapsto M \} \cup H
\end{align}
Value terms include variables $x$, lambda abstraction $\lambda x. M$, tuple $\langle V_1, \dots, V_n \rangle$, and projection $\pi_i\ V$.
Computation terms include function application $V\ W$, returning a value $\return{V}$, let-expression $\letin{x}{M}{N}$, operation invocation $\mathtt{op}(V)$, and effect handling construct (or \emph{handle-with} expression) $\handlewithto{M}{H}{x}{N}$.
In handle-with, operations in $M$ are handled by $H$, and the result value is bound to $x$ in $N$.
An effect handler $H$ is defined as a finite set of operation clauses $\mathtt{op}(x, k) \mapsto M$.
We separate operation clauses in $H$ from return clauses $N$, as it simplifies the presentation of the equational theory for effect handlers in Section~\ref{sec:equational-theory-summary}.
Although this syntax is less common in the recent literature on effect handlers, the original syntax for effect handling~\cite{PlotkinLMCS2013} also separates operation clauses from return clauses.
We define substitution $W[V/x]$, $M[V/x]$, and $H[V/x]$ of value terms $V$ for variables $x$ as usual.

\begin{figure}
	\begin{mathpar}
		\inferrule{
			\Gamma \vdash V : A
		}{
			\Gamma \vdash \return{V} : A ! \Sigma
		}
		\and
		\inferrule{
			\Gamma \vdash M : A ! \Sigma \\
			\Gamma, x : A \vdash N : B ! \Sigma
		}{
			\Gamma \vdash \letin{x}{M}{N} : B ! \Sigma
		}
		\and
		\inferrule{
			(\mathtt{op} : A \rightarrowtriangle B) \in \Sigma \\
			\Gamma \vdash V : A
		}{
			\Gamma \vdash \mathtt{op}(V) : B ! \Sigma
		}
		\and
		\inferrule{
			\Gamma \vdash M : A ! \Sigma \\
			\Gamma \vdash H : \Sigma \Rightarrow C \\
			\Gamma, x : A \vdash N : C
		}{
			\Gamma \vdash \handlewithto{M}{H}{x : A}{N} : C
		}
		\and
		\inferrule{ }{
			\Gamma \vdash \emptyset : \emptyset \Rightarrow C
		}
		\and
		\inferrule{
			\Gamma, x : A_{\mathtt{op}}, k : B_{\mathtt{op}} \to C \vdash M_{\mathtt{op}} : C \\
			\Gamma \vdash H : \Sigma \Rightarrow C
		}{
			\Gamma \vdash \{ \mathtt{op}(x, k) \mapsto M_{\mathtt{op}} \} \cup H : \{ \mathtt{op} : A_{\mathtt{op}} \rightarrowtriangle B_{\mathtt{op}} \} \cup \Sigma \Rightarrow C
		}
	\end{mathpar}
	\caption{Selected typing rules. See Appendix~\ref{sec:full-typing-rules} for full typing rules.}
	\label{fig:selected-typing-rules}
\end{figure}

We define a simple type system for $\EffectHandlerCalculus$ as follows.
Types are defined as follows:
\begin{align}
	&\text{value types} &A, B \quad&\coloneqq\quad A \to C \mid \prod_{i = 1}^n A_i \\
	&\text{computation types} &C, D \quad&\coloneqq\quad A ! \Sigma \\
	&\text{signatures} &\Sigma \quad&\coloneqq\quad \emptyset \mid \{ \mathtt{op} : A \rightarrowtriangle B \} \cup \Sigma \\
	&\text{handler types} &H \quad&\coloneqq\quad \Sigma \Rightarrow C
\end{align}
Here, $n \ge 0$ is a natural number, and $\mathtt{op}$ is an operation symbol.
We sometimes write $\prod_{i = 1}^n A_i = A_1 \times \dots \times A_n$ and $\prod_{i = 1}^0 A_i = \UnitType$ for product types.
Signatures $\Sigma$ are finite sets of operation symbols $\mathtt{op}$ with their argument types and return types.
The order of operation symbols in $\Sigma$ does not matter.
The computation type $A ! \Sigma$ is the type of computations that return a value of type $A$ and may perform operations in signature $\Sigma$.
The handler type $\Sigma \Rightarrow C$ is the type of handlers that handle operations in signature $\Sigma$ and transform computations of type $A ! \Sigma$ to those of type $C$.
As usual, contexts $\Gamma = x_1 : A_1, \dots, x_n : A_n$ are defined as a list of pairs of variables and value types.
Figure~\ref{fig:selected-typing-rules} shows typing rules involving signature and handlers.

\subsection{Equational Theory}\label{sec:equational-theory-summary}

We consider judgements for typed equations of the following forms:
\[ \Gamma \vdash V = W : A, \qquad \Gamma \vdash M = N : C, \qquad \Gamma \vdash H = H' : \Sigma \Rightarrow C \]
These judgements are defined as the least congruence relation containing the equations in Figures~\ref{fig:beta-eta-monad-equations} and~\ref{fig:effect-handler-equations}.
We omit obvious rules for reflexivity, symmetry, transitivity, and congruence for each term constructor; a complete description of these rules can be found in Appendix~\ref{sec:equational-theory}.
The rules in Figure~\ref{fig:beta-eta-monad-equations} correspond to the standard $\beta$-/$\eta$-laws for lambda abstraction and product types, and to the monad laws for let-expressions, respectively.

The most important part of our equational theory is the equations for handle-with in Figure~\ref{fig:effect-handler-equations}.
Here, we consider three equations.
\eqref{eq:handle-return} states that if the handled computation returns a value, the return clause is executed.
\eqref{eq:handle-let} states that handling a let-expression is equivalent to handling the first computation and then handling the second computation with the result of the first computation.
\eqref{eq:handle-op} states that when an operation is invoked in a handled computation, the corresponding operation clause in the handler is executed.

Note that the equations in Figure~\ref{fig:effect-handler-equations} differ from those considered in \cite{BauerLMCS2014}.
Our equational theory corresponds to the \emph{fine-grained operational semantics} for effect handlers studied in \cite{SieczkowskiICFP2023,KarachaliasOOPSLA2021}, while the equational theory in~\cite{BauerLMCS2014} corresponds to the traditional \emph{context-capturing operational semantics}.
The equations corresponding to context-capturing operational semantics are as follows.
\begin{align}
	&\handlewithto{\return{V}}{H}{x}{N} = N[V/x] \\
	&\handlewithto{\mathcal{C}[\mathtt{op}(V)]}{H}{x}{N} \\
	&\quad= M_{\mathtt{op}}\![V/x, \lambda y. \mathcal{C}[\return{y}]/k]
	\label{eq:context-capturing}
\end{align}
Here, $\mathcal{C}$ is an evaluation context and $(\mathtt{op}(x, k) \mapsto M_{\mathtt{op}}) \in H$ is the operation clause for $\mathtt{op}$.
We adopt the fine-grained version because it makes the interaction between handle-with and the monad structure (i.e., $\mathtt{return}$ and $\mathtt{let}$) clearer than the context-capturing version.
The relationship between these two styles of operational semantics is studied in \cite{SieczkowskiICFP2023}.
We will see in Section~\ref{sec:operational-semantics} that our equational theory is consistent with the context-capturing operational semantics.

\begin{figure}
	\begin{align}
		(\lambda x. M)\ V\ &=\ M[V/x]
		&V\ &=\ \lambda x. (V\ x) \\
		\pi_i\ \langle V_1, \dots, V_n \rangle \ &=\ V_i
		&V\ &=\ \langle \pi_1\ V, \dots, \pi_n\ V \rangle
	\end{align}
	\begin{align}
		&\letin{x}{\return{V}}{M} \quad=\quad M[V/x] \\
		&\letin{x}{M}{\return{x}} \quad=\quad M \\
		&\letin{y}{(\letin{x}{M}{N})}{L} \\
		&\quad=\quad \letin{x}{M}{\letin{y}{N}{L}}
	\end{align}
	\caption{The $\beta$ and $\eta$ laws for lambda abstraction and product types, and the monad laws for let-expressions.}
	\label{fig:beta-eta-monad-equations}
\end{figure}

\begin{figure}
	\begin{align}
		&\handlewithto{\return{V}}{H}{x}{M} \quad=\quad M[V / x] \\
		\tag{\textsc{Handle-Ret}} \label{eq:handle-return} \\
		&\handlewithto{(\letin{x}{L}{M})}{H}{y}{N} \\
		&\quad=\quad \handlewithto{L}{H}{x}{\handlewithto{M}{H}{y}{N}} \\ \tag{\textsc{Handle-Let}} \label{eq:handle-let} \\
		&\handlewithto{\mathtt{op}(V)}{H}{x}{M} \quad=\quad M_{\mathtt{op}}[V / x, \lambda x. M / k] \\
		&\qquad\qquad \text{where $(\mathtt{op}(x, k) \mapsto M_{\mathrm{op}}) \in H$} \tag{\textsc{Handle-Op}} \label{eq:handle-op}
	\end{align}
	\caption{Equations for effect handlers.}
	\label{fig:effect-handler-equations}
\end{figure}

\subsection{Operational Semantics}
\label{sec:operational-semantics}

As a sanity check of the equational theory, we show that it is consistent with a context-capturing small-step operational semantics.
\begin{definition}\label{def:operational-semantics}
	We define a relation $\leadsto$ on closed computation terms as the least relation satisfying the rules below.
	\begin{align}
		\mathcal{C}[M] &\leadsto \mathcal{C}[N] \qquad \text{if $M \leadsto N$} \\
		\letin{x}{\return{V}}{M} &\leadsto M[V/x] \\
		(\lambda x. M)\ V &\leadsto M[V/x] \\
		\handlewithto{\return{V}}{H}{x}{M} &\leadsto M[V / x] \\
		\handlewithto{\mathcal{C}[\mathtt{op}(V)]}{H}{x}{M} &\leadsto \\
		&\mkern-180mu M_{\mathtt{op}}[V/x,\lambda y. \handlewithto{\mathcal{C}[\return{y}]}{H}{x}{M}/k] \\
		&\mkern-160mu \text{if $\mathcal{C}$ does not contain handle-with} \label{eq:operational-handle-operation}
	\end{align}
	Here, evaluation contexts $\mathcal{C}$ are defined as follows:
	\[ \mathcal{C} \quad\coloneqq\quad [] \mid \letin{x}{\mathcal{C}}{N} \mid \handlewithto{\mathcal{C}}{H}{x}{M} \]
\end{definition}

As expected, the operational semantics satisfies type soundness.
\begin{proposition}[type soundness]\label{prop:type-soundness}
	If $\vdash M : C$ is a closed computation term, then either $M$ is a normal form, or there exists a computation term $N$ such that $M \leadsto N$ and $\vdash N : C$.
	Here, normal forms are computation terms of the form (a) $\return{V}$, (b) $\mathcal{C}[\mathtt{op}(V)]$ where $\mathcal{C}$ does not contain handle-with, or (c) $\mathcal{C}[(\pi_i\ V)\ W]$ where $\mathcal{C}$ may contain handle-with.
	\qed
\end{proposition}
\begin{appendixproof}[Proof of Proposition~\ref{prop:type-soundness}]
	We show that computation term $\vdash M : C$ is either a normal form or reducible by induction on the structure of $M$.
	\begin{itemize}
		\item If $M = V\ W$, then $V : A \to C$ is of the form $\lambda x. N$, $\pi_i\ V$, or a variable by typing rules.
		Since $M$ is closed, $V$ cannot be a variable.
		If $M = (\lambda x. N)\ W$, then it is reducible.
		If $M = (\pi_i\ V)\ W$, then it is a normal form.
		\item If $M = \letin{x}{N_1}{N_2}$, by the induction hypothesis, $N_1$ is either a normal form or reducible.
		If $N_1$ is reducible, then so is $M$ because $\mathcal{C} = \letin{x}{[]}{N_2}$ is an evaluation context.
		If $N_1$ is a normal form of the form $\return{V}$, then $M$ is reducible.
		If $N_1$ is a normal form of the form $\mathcal{C}[\mathtt{op}(V)]$ or $\mathcal{C}[(\pi_i\ V)\ W]$, then $M$ is also a normal form.
		\item If $M = \handlewithto{N_1}{H}{x}{N_2}$, then by the induction hypothesis, $N_1$ is either a normal form or reducible.
		If $N_1$ is reducible, then so is $M$ because $\mathcal{C} = \handlewithto{[]}{H}{x}{N_2}$ is an evaluation context.
		If $N_1$ is a normal form of the form $\return{V}$ or $\mathcal{C}[\mathtt{op}(V)]$, then $M$ is reducible.
		If $N_1$ is a normal form of the form $\mathcal{C}[(\pi_i\ V)\ W]$, then $M$ is also a normal form.
		\item Other cases are normal forms.
	\end{itemize}
	Next, we show type preservation: if $\Gamma \vdash M : C$ and $M \leadsto N$, then $\Gamma \vdash N : C$.
	Most cases are straightforward.
	When $M = \handlewithto{\mathcal{C}[\mathtt{op}(V)]}{H}{x}{N}$ and $\mathcal{C}$ does not contain handle-with, the signature for $\mathcal{C}[\mathtt{op}(V)]$ and that for $\mathtt{op}(V)$ are the same, so the corresponding operation clause exists in $H$.
\end{appendixproof}

\begin{remark}\label{rem:pattern-matching-product}
	In the definition of $\leadsto$, we have not included rules for decomposing tuples (e.g., $\pi_1\ \langle V_1, V_2 \rangle \leadsto V_1$), which is why we have normal forms of the form $\mathcal{C}[(\pi_i\ V)\ W]$ in Proposition~\ref{prop:type-soundness}.
	There are two possible approaches to avoid such normal forms; (a) adding reduction rules for value terms of the form $\pi_i\ V$ or (b) removing projection from the definition of value terms and adding pattern matching for tuples as computation terms, as in \cite{HillerstromFSCD2017}.
	We discuss both approaches in the appendix.
	Although the latter approach (b) is more natural in terms of operational semantics, we have chosen to include projection in value terms, as it is more standard when discussing completeness of categorical semantics \cite{LevyInformationandComputation2003}.
	This minor discrepancy between common practice of categorical semantics and operational semantics for dealing with tuples is orthogonal to our main target, i.e., effect handlers, and hence we do not pursue this issue further.
\end{remark}

\begin{toappendix}
Below, we consider two approaches for avoid normal forms of the form $(\pi_i\ V)\ W$.
The first one is to the operational semantics with reduction rules for value terms.
\begin{definition}\label{def:operational-semantics-with-value-reduction}
	We extend Definition~\ref{def:operational-semantics} by adding the following rule:
	\[ V\ W \leadsto V'\ W \qquad \text{if $V \leadsto^v V'$} \]
	where value term reduction $\leadsto^v$ is defined as follows:
	\begin{align}
		\pi_i\ \langle V_1, \dots, V_n \rangle &\leadsto^v V_i \\
		\mathcal{D}[V] &\leadsto^v \mathcal{D}[W] \quad \text{if $V \leadsto^v W$ where $\mathcal{D} \coloneqq [] \mid \pi_i\ \mathcal{D}$}
	\end{align}
\end{definition}

\begin{lemma}\label{lem:value-term-type-sound}
	For any closed value term $\vdash V : A$, either $V$ is a normal form or there exists a value term $W$ such that $V \leadsto^v W$.
	In the latter case, we have $\vdash W : A$.
	Here, normal forms for value terms are variables $x$, lambda abstractions $\lambda x. M$, and tuples $\langle V_1, \dots, V_n \rangle$.
\end{lemma}
\begin{proof}
	By induction on the structure of $V$.
	It suffices to show that value terms of the form $\pi_i\ V$ is reducible.
	If $V$ is a normal form, then it is of the form $\langle V_1, \dots, V_n \rangle$ by typing rules, so $\pi_i\ V$ is reducible.
	If $V$ is reducible, then so is $\pi_i\ V$ by the definition of $\leadsto^v$.
\end{proof}

\begin{proposition}
	If we use the extended operational semantics in Definition~\ref{def:operational-semantics-with-value-reduction}, then type soundness (Proposition~\ref{prop:type-soundness}) holds with normal forms being computation terms of the form $\return{V}$ or $\mathcal{C}[\mathtt{op}(V)]$ where $\mathcal{C}$ does not contain handle-with.
\end{proposition}
\begin{proof}
	By Lemma~\ref{lem:value-term-type-sound}, computation terms of the form $\mathcal{C}[(\pi_i\ V)\ W]$ are now reducible.
\end{proof}

The second approach is to modify the syntax of $\EffectHandlerCalculus$ (Section~\ref{sec:syntax}) so that decomposition of tuples occurs only in computation terms.
This approach is more natural in terms of operational semantics and taken in \cite{HillerstromFSCD2017}.
\begin{definition}\label{def:syntax-with-pattern-matching}
	We modify the syntax of value and computation terms of $\EffectHandlerCalculus$ as follows.
	\begin{align}
		V, W \quad&\coloneqq\quad x \mid \lambda x. M \mid \langle V_1, \dots, V_n \rangle \\
		M, N \quad&\coloneqq\quad V\ W \mid \return{V} \mid \letin{x}{M}{N} \mid \mathtt{op}(V) \mid \handlewithto{M}{H}{x}{N} \mid \patternmatch{V}{x_1, \dots, x_n}{M}
	\end{align}
	Typing rules for pattern matching are defined in the obvious way.
	We also define the operational semantics $\leadsto$ as follows.
	\begin{align}
		\mathcal{C}[M] &\leadsto \mathcal{C}[N] \quad \text{if $M \leadsto N$} \\
		\letin{x}{\return{V}}{M} &\leadsto M[V/x] \\
		(\lambda x. M)\ V &\leadsto M[V/x] \\
		\handlewithto{\return{V}}{H}{x}{M} &\leadsto M[V / x] \\
		\handlewithto{\mathcal{C}[\mathtt{op}(V)]}{H}{x}{M} &\leadsto M_{\mathtt{op}}[V/x,\lambda y. \handlewithto{\mathcal{C}[\return{y}]}{H}{x}{M}/k] \\
		&\qquad\text{if $\mathcal{C}$ does not contain handle-with} \\
		\patternmatch{\langle V_1, \dots, V_n \rangle}{x_1, \dots, x_n}{M} &\leadsto M[V_1/x_1, \dots, V_n/x_n]
	\end{align}
	Here, evaluation contexts $\mathcal{C}$ are defined as follows:
	\[ \mathcal{C} \quad\coloneqq\quad [] \mid \letin{x}{\mathcal{C}}{N} \mid \handlewithto{\mathcal{C}}{H}{x}{M} \]
\end{definition}
The main difference from the original syntax is that projection $\pi_i\ V$ is removed from value terms and added as pattern matching for tuples $\patternmatch{V}{x_1, \dots, x_n}{M}$ in computation terms.
Using this modified definition, we can show the following type soundness result, which is more natural than Proposition~\ref{prop:type-soundness}.

\begin{proposition}
	Consider the modified syntax and operational semantics in Definition~\ref{def:syntax-with-pattern-matching}.
	Then, type soundness (Proposition~\ref{prop:type-soundness}) holds with normal forms being computation terms of the form $\return{V}$ or $\mathcal{C}[\mathtt{op}(V)]$ where $\mathcal{C}$ does not contain handle-with.
\end{proposition}
\begin{proof}
	The proof is similar to that of Proposition~\ref{prop:type-soundness}.
	The main differences are that we can exclude normal forms of the form $\mathcal{C}[(\pi_i\ V)\ W]$ and that we need to consider the case for pattern matching as follows.
	Suppose $M = \patternmatch{V}{x_1, \dots, x_n}{N}$.
	Since $V$ is a closed value term of type $\prod_{i = 1}^n A_i$, it is of the form $\langle V_1, \dots, V_n \rangle$ by typing rules.
	Thus, $M$ is reducible.
\end{proof}
\end{toappendix}

The reduction relation $\leadsto$ is sound with respect to the equational theory we consider in Section~\ref{sec:equational-theory-summary}.
This gives a justification that our equational theory in Section~\ref{sec:equational-theory-summary} is reasonable.
\begin{proposition}\label{prop:operational-equational-soundness}
	If $M \leadsto N$, then $\Gamma \vdash M = N : C$.
	\qed
\end{proposition}
\begin{toappendix}
We prove Proposition~\ref{prop:operational-equational-soundness}.
Before the proof, we show the following lemma.
\begin{lemma}\label{lem:operational-handle-operation-sound}
	If $\mathcal{C}$ does not contain handle-with, then the following equation is derivable.
	\[ \handlewithto{\mathcal{C}[\mathtt{op}(V)]}{H}{x}{M} \quad=\quad M_{\mathtt{op}}[V/x,\lambda y. \handlewithto{\mathcal{C}[\return{y}]}{H}{x}{M}/k] \]
\end{lemma}
\begin{proof}
	By induction on the structure of evaluation contexts $\mathcal{C}$.
	If $\mathcal{C} = []$, it follows from \eqref{eq:handle-op} and \eqref{eq:handle-return}.
	\begin{align}
		&\handlewithto{\mathtt{op}(V)}{H}{x}{M} \\
		&= M_{\mathtt{op}}[V/x,\lambda x. M/k] \\
		&= M_{\mathtt{op}}[V/x,\lambda y. \handlewithto{\return{y}}{H}{x}{M}/k]
	\end{align}
	If $\mathcal{C} = \letin{y}{\mathcal{C}}{N}$, it follows from the induction hypothesis and \eqref{eq:handle-let}.
	\begin{align}
		&\handlewithto{\letin{y}{\mathcal{C}[\mathtt{op}(V)]}{N}}{H}{x}{M} \\
		&= \handlewithto{\mathcal{C}[\mathtt{op}(V)]}{H}{y}{\handlewithto{N}{H}{x}{M}} \\
		&= M_{\mathtt{op}}[V/x,\lambda z. \handlewithto{\mathcal{C}[\return{z}]}{H}{y}{\handlewithto{N}{H}{x}{M}}/k] \\
		&= M_{\mathtt{op}}[V/x,\lambda z. \handlewithto{\letin{y}{\mathcal{C}[\return{z}]}{N}}{H}{x}{M}/k]
		\qedhere
	\end{align}
\end{proof}
\end{toappendix}
\begin{appendixproof}[Proof of Proposition~\ref{prop:operational-equational-soundness}]
	The case for $\mathcal{C}[M] \leadsto \mathcal{C}[N]$ follows from congruence rules and the induction hypothesis.
	The case for \eqref{eq:operational-handle-operation} follows from Lemma~\ref{lem:operational-handle-operation-sound}.
	The other cases follow from the corresponding equations in the equational theory.
\end{appendixproof}

\begin{toappendix}
\section{Typing Rules}
\label{sec:full-typing-rules}

Contexts are defined as follows.
\[ \Gamma \coloneqq \emptyctx \mid \Gamma, x : A \]

Value terms $\Gamma \vdash V : A$
\begin{mathpar}
	\inferrule{
		(x : A) \in \Gamma
	}{
		\Gamma \vdash x : A
	}
	\and
	\inferrule{
		\Gamma, x : A \vdash M : C
	}{
		\Gamma \vdash \lambda x. M : A \to C
	}
	\and
	\inferrule{
		\forall i, \Gamma \vdash V_i : A_i
	}{
		\Gamma \vdash \langle V_1, \dots, V_n \rangle : \prod_{i = 1}^n A_i
	}
	\and
	\inferrule{
		\Gamma \vdash V : \prod_{i = 1}^n A_i
	}{
		\Gamma \vdash \pi_i\ V : A_i
	}
\end{mathpar}

Computation terms $\Gamma \vdash M : C$
\begin{mathpar}
	\inferrule{
		\Gamma \vdash V : A \to C \\
		\Gamma \vdash W : A
	}{
		\Gamma \vdash V\ W : C
	}
	\and
	\inferrule{
		\Gamma \vdash V : A
	}{
		\Gamma \vdash \return{V} : A {!} E
	}
	\and
	\inferrule{
		\Gamma \vdash M : A ! E \\
		\Gamma, x : A \vdash N : B ! E
	}{
		\Gamma \vdash \letin{x}{M}{N} : B ! E
	}
	\and
	\inferrule{
		(\mathtt{op} : A \rightarrowtriangle B) \in E \\
		\Gamma \vdash V : A
	}{
		\Gamma \vdash \mathtt{op}(V) : B ! E
	}
	\and
	\inferrule{
		\Gamma \vdash M : A ! E \\
		\Gamma \vdash H : E \Rightarrow C \\
		\Gamma, x : A \vdash N : C
	}{
		\Gamma \vdash \handlewithto{M}{H}{x}{N} : C
	}
\end{mathpar}

Handlers $\Gamma \vdash H : E \Rightarrow C$
\begin{mathpar}
	\inferrule{ }{
		\Gamma \vdash \emptyset : \emptyset \Rightarrow C
	}
	\and
	\inferrule{
		\Gamma, x : A_{\mathtt{op}}, k : B_{\mathtt{op}} \to C \vdash M_{\mathtt{op}} : C \\
		\Gamma \vdash H : E \Rightarrow C
	}{
		\Gamma \vdash \{ \mathtt{op}(x, k) \mapsto M_{\mathtt{op}} \} \cup H : \{ \mathtt{op} : A_{\mathtt{op}} \rightarrowtriangle B_{\mathtt{op}} \} \cup E \Rightarrow C
	}
\end{mathpar}

\section{Equational Theory}
\label{sec:equational-theory}

We consider the following judgment form:
\[ \Gamma \vdash V = W : A \qquad \Gamma \vdash M = N : C \qquad \Gamma \vdash H = H' : \Sigma \Rightarrow C \]
We implicitly use $\alpha$-conversion of bound variables in derivation rules.
We have four groups of derivation rules.
\begin{itemize}
	\item reflexivity, symmetry, transitivity, and congruence rules
	\item $\beta$ and $\eta$ rules for lambda abstraction and product types
	\item monad laws for let-expressions
	\item rules for effect handlers
\end{itemize}

For value types, reflexivity, symmetry, and transitivity rules are given as follows.
We have similar rules for computation terms and handlers.
\begin{mathpar}
	\inferrule[Val-Refl]{
		\Gamma \vdash V : A
	}{
		\Gamma \vdash V = V : A
	}
	\and
	\inferrule[Val-Symm]{
		\Gamma \vdash V = W : A
	}{
		\Gamma \vdash W = V : A
	}
	\and
	\inferrule[Val-Trans]{
		\Gamma \vdash U = V : A \\
		\Gamma \vdash V = W : A
	}{
		\Gamma \vdash U = W : A
	}
\end{mathpar}

Congruence rules are given for each term constructor.
We do not list all congruence rules here for brevity.
For example, the congruence rule for lambda abstraction is given as follows.
\begin{mathpar}
	\inferrule[Abs-Congr]{
		\Gamma, x : A \vdash M = N : C
	}{
		\Gamma \vdash \lambda x. M = \lambda x. N : A \to C
	}
\end{mathpar}

Other derivation rules are given as follows.
\begin{mathpar}
	\inferrule[Abs-Beta]{
		\Gamma, x : A \vdash M : B ! \Sigma \\
		\Gamma \vdash V : A
	}{
		\Gamma \vdash (\lambda x. M)\ V = M[V/x] : B ! \Sigma
	}
	\and
	\inferrule[Abs-Eta]{
		\Gamma \vdash V : A \to C \\
		\text{$x$ fresh}
	}{
		\Gamma \vdash V = \lambda x. (V\ x) : A \to C
	}
	\and
	\inferrule[Prod-Beta]{
		\forall i, \quad \Gamma \vdash V_i : A_i
	}{
		\Gamma \vdash \pi_i\ \langle V_1, \dots, V_n \rangle = V_i
	}
	\and
	\inferrule[Prod-Eta]{
		\Gamma \vdash V : \prod_{i = 1}^n A_i
	}{
		\Gamma \vdash V = \langle \pi_1\ V, \dots, \pi_n\ V \rangle : \prod_{i = 1}^n A_i
	}
	\and
	\inferrule[Ret-Let]{
		\Gamma \vdash V : A \\
		\Gamma, x : A \vdash M : B ! \Sigma
	}{
		\Gamma \vdash \letin{x}{\return{V}}{M} = M[V/x] : B ! \Sigma
	}
	\and
	\inferrule[Let-Let]{
		\Gamma \vdash M : A_1 ! \Sigma \\
		\Gamma, x : A_1 \vdash N : A_2 ! \Sigma \\
		\Gamma, y : A_2 \vdash L : A_3 ! \Sigma \\
	}{
		\Gamma \vdash \letin{y}{(\letin{x}{M}{N})}{L} = \letin{x}{M}{\letin{y}{N}{L}} : A_3 ! \Sigma
	}
	\and
	\inferrule[Let-Ret]{
		\Gamma \vdash M : A ! \Sigma
	}{
		\Gamma \vdash \letin{x}{M}{\return{x}} = M : A ! \Sigma
	}
	\and
	\inferrule[Handle-Ret]{
		\Gamma \vdash V : A \\
		\Gamma, x : A \vdash M : C \\
		\Gamma \vdash H : \Sigma \Rightarrow C
	}{
		\Gamma \vdash \handlewithto{\return{V}}{H}{x}{M} = M[V / x] : C
	}
	\and
	\inferrule[Handle-Let]{
		\Gamma \vdash H : \Sigma \Rightarrow C \\
		\Gamma \vdash L : A ! \Sigma \\
		\Gamma, x : A \vdash M : B ! \Sigma \\
		\Gamma, y : B \vdash N : C
	}{
		{\begin{aligned}
			\Gamma &\vdash \handlewithto{(\letin{x}{L}{M})}{H}{y}{N} \\
			&= \handlewithto{L}{H}{x}{\handlewithto{M}{H}{y}{N}} : C
		\end{aligned}}
	}
	\and
	\inferrule[Handle-Op]{
		(\mathtt{op} : A_{\mathtt{op}} \rightarrowtriangle B_{\mathtt{op}}) \in \Sigma \\
		(\mathtt{op}(x, k) \mapsto M_{\mathtt{op}}) \in H \\
		\Gamma \vdash H : \Sigma \Rightarrow C \\
		\Gamma, x : B_{\mathtt{op}} \vdash M : C \\
		\Gamma \vdash V : A_{\mathtt{op}}
	}{
		\Gamma \vdash \handlewithto{\mathtt{op}(V)}{H}{x}{M} = M_{\mathtt{op}}[V / x, \lambda x. M / k] : C
	}
	\and
	\inferrule[Handler-Congr]{
		\Gamma \vdash H = H' : \Sigma \Rightarrow C \\
		\Gamma, x : A_{\mathtt{op}}, k : B_{\mathtt{op}} \to C \vdash M_{\mathrm{op}} = N_{\mathrm{op}} : C
	}{
		\Gamma \vdash \{ \mathtt{op}(x, k) \mapsto M_{\mathrm{op}} \} \cup H = \{ \mathtt{op}(x, k) \mapsto N_{\mathrm{op}} \} \cup H' : \{ \mathtt{op} : A_{\mathtt{op}} \rightarrowtriangle B_{\mathtt{op}} \} \cup \Sigma \Rightarrow C
	}
\end{mathpar}

\section{Basic Properties}
\label{sec:basic-properties}
\begin{lemma}\label{lem:weakening}
	For any value term $V$, the following rule is admissible:
	\begin{mathpar}
		\inferrule{
			\Gamma, \Delta \vdash V : A
		}{
			\Gamma, x : B, \Delta \vdash V : A
		}
	\end{mathpar}
	The same holds for computation terms and handlers.
\end{lemma}
\begin{proof}
	Straightforward by (simultaneous) induction on type derivations.
\end{proof}

\begin{definition}
	A \emph{substitution} is a finite mapping $\sigma$ from variables to value terms.
	The application of a substitution $\sigma$ to a term $M$, denoted by $M[\sigma]$, is defined by replacing each free occurrence of a variable $x$ in $M$ with $\sigma(x)$ if $x$ is in the domain of $\sigma$; otherwise, $x$ remains unchanged.
	Concrete definitions by induction on the structure of terms are standard and omitted here.
	If the domain of $\sigma$ is $\{ x_1, \dots, x_n \}$, and $\sigma(x_i) = V_i$ for each $i = 1, \dots, n$, then we often write $M[\sigma]$ as $M[V_1/x_1, \dots, V_n/x_n]$.
\end{definition}

\begin{definition}
	Let $\sigma$ be a substitution.
	We write $\sigma : \Gamma \to \Delta$ if the domain of $\sigma$ is the set of variables in the context $\Delta$ and for every $x : A$ in the context $\Delta$, we have $\Gamma \vdash \sigma(x) : A$.
\end{definition}

\begin{lemma}
	For any value term $V$, the following rule is admissible:
	\begin{mathpar}
		\inferrule{
			\Delta \vdash V : B \\
			\sigma : \Gamma \to \Delta
		}{
			\Gamma \vdash V[\sigma] : B
		}
	\end{mathpar}
	The same holds for computation terms and handlers.
\end{lemma}
\begin{proof}
	The proof is by (simultaneous) induction on type derivations of $V$.
\end{proof}

\begin{definition}
	We define the composition of substitutions as follows.
	Given two substitutions $\sigma_1$ and $\sigma_2$, their composition $\sigma_2 \comp \sigma_1$ is the substitution defined by $(\sigma_2 \comp \sigma_1)(x) \coloneqq \sigma_2(x)[\sigma_1]$.
\end{definition}

\begin{lemma}\label{lem:subst-assoc}
	Let $\sigma_1 : \Gamma_1 \to \Gamma_2$ and $\sigma_2 : \Gamma_2 \to \Gamma_3$ be substitutions, and $\Gamma_3 \vdash V : A$ be a value term.
	Then, the following definitional equality holds:
	\[ V[\sigma_2 \comp \sigma_1] = V[\sigma_2][\sigma_1]. \]
	The same holds for computation terms and handlers.
\end{lemma}
\begin{proof}
	We prove the statement for value terms, computation terms and handlers simultaneously.
	The proof is by induction on $V$.
	Most cases are straightforward.
	If a term constructor involves variable binding, we use weakening (Lemma~\ref{lem:weakening}) to extend $\sigma : \Gamma \to \Delta$ to $\sigma[z \mapsto z] : \Gamma, z : B \to \Delta, z : B$ where $z$ is a variable bound in $V$.
	As is standard, $\sigma[z \mapsto W]$ is the substitution defined as follows.
	\[ \sigma[z \mapsto W](x) \coloneqq \begin{cases}
		W & \text{if } x = z \\
		\sigma(x) & \text{otherwise}
	\end{cases} \]
	For example, if $V = \lambda z : B. M$, then by applying the induction hypothesis to $\sigma_1[z \mapsto z] : \Gamma_1, z : B \to \Gamma_2, z : B$, $\sigma_2[z \mapsto z] : \Gamma_2, z : B \to \Gamma_3, z : B$, and $\Gamma_3, z : B \vdash M : C$, we have
	\begin{align}
		M[\sigma_2 \comp \sigma_1] &= M[\sigma_2[z \mapsto z] \comp \sigma_1[z \mapsto z]] \\
		&= M[\sigma_2[z \mapsto z]][\sigma_1[z \mapsto z]] \\
		&= M[\sigma_2][\sigma_1]
	\end{align}
	By the congruence rule for lambda abstraction and by definition of substitution, we have
	\[ (\lambda z. M)[\sigma_2 \comp \sigma_1] = (\lambda z. M)[\sigma_2][\sigma_1]. \qedhere \]
\end{proof}

\end{toappendix}

\section{Denotational Semantics}
\label{sec:semantics}
We extend the categorical semantics of fine-grain call-by-value (i.e., $\lambda_c$-models in Section~\ref{sec:preliminaries}) to interpret effect handlers.

\subsection{Models and Interpretations}
We first informally discuss the structures needed to interpret signatures, computation types, effect handlers, and handle-with in $\EffectHandlerCalculus$.
Then, we formally define models of $\EffectHandlerCalculus$.

\paragraph{Semantics of signatures}
We define a \emph{semantic signature} in a category $\category{C}$ as a partial function $S$ from operation symbols to pairs of objects in $\category{C}$ such that the domain of $S$ is finite.
We often use set-like notation to denote a semantic signature:
\[ S \quad=\quad \{ \mathtt{op}_1 : A_1 \rightarrowtriangle B_1,\quad \mathtt{op}_2 : A_2 \rightarrowtriangle B_2, \dots \} \]
We also write $(\mathtt{op} : A \rightarrowtriangle B) \in S$ if $S(\mathtt{op}) = (A, B) \in \category{C}^2$.
Let $\SemanticEffectType{\category{C}}$ be the collection of all semantic signatures in a category $\category{C}$.
We will see in Definition~\ref{def:interpretation-lambda-eff} that a (syntactic) signature $\Sigma$ is interpreted as a semantic signature $\interpret{\Sigma} \in \SemanticEffectType{\category{C}}$.

\paragraph{Semantics of computation types}
Computation types $A ! \Sigma$ are interpreted using a strong monad: $\interpret{A ! \Sigma} = T_{\interpret{\Sigma}} \interpret{A}$.
Here, we need a \emph{family} of strong monads $\{ T_S \}_{S \in \SemanticEffectType{\category{C}}}$ indexed by semantic signatures rather than a single strong monad, since sequencing two computations is defined only when their signatures are the same (recall Figure~\ref{fig:selected-typing-rules}).
The interaction between different signatures only occurs when we handle effects, which we will discuss below.

\paragraph{Semantics of effect handlers}
By the typing rule for handlers, a handler of type $\Sigma \Rightarrow C$ consists of a finite family of terms $\lambda x. \lambda k. M_{\mathtt{op}} : A_{\mathtt{op}} \to (B_{\mathtt{op}} \to C) \to C$ for each $(\mathtt{op} : A_{\mathtt{op}} \rightarrowtriangle B_{\mathtt{op}}) \in \Sigma$.
Thus, we define its semantic counterpart as follows.
\begin{definition}
	Let $\category{C}$ be a cartesian category and $T$ be a strong monad on $\category{C}$ such that $\category{C}$ has Kleisli exponentials for $T$.
	For any semantic signature $S$ and $X \in \category{C}_{T}$, we define $\HBundle{S}{X} \in \category{C}$ by
	\[ \HBundle{S}{X} \quad\coloneqq\quad \prod_{(\mathtt{op} : A_{\mathtt{op}} \rightarrowtriangle B_{\mathtt{op}}) \in S} (\KleisliExp{T}{A_{\mathtt{op}} \times (\KleisliExp{T}{B_{\mathtt{op}}}{X})}{X}). \]
	Note that $\HBundle{S}{X}$ can be seen as the object of $S$-algebra structures on $X$ in the Kleisli category $\category{C}_T$.
\end{definition}

Then, an effect handler $\Gamma \vdash H : \Sigma \Rightarrow A ! \Sigma'$ is interpreted as $\interpret{H} : \interpret{\Gamma} \to \HBundle{\interpret{\Sigma}}{J \interpret{A}}$ where $J : \category{C} \to \category{C}_{T_{\interpret{\Sigma'}}}$.

\paragraph{Interpretation of handle-with}
The interpretation of a handle-with expression should be constructed from the interpretation of its sub-terms. By the typing rule (Figure~\ref{fig:selected-typing-rules}), this is achieved if we have a morphism of the form
\[ \interpret{A ! \Sigma} \times \interpret{\Sigma \Rightarrow B ! \Sigma'} \times \interpret{A \to B ! \Sigma'} \to \interpret{B ! \Sigma'} \]
for each value type $A, B$ and signature $\Sigma, \Sigma'$.
As we have discussed above, $\interpret{\Sigma \Rightarrow B ! \Sigma'}$ and $\interpret{A ! \Sigma}$ correspond to $\HBundle{\interpret{\Sigma}}{J \interpret{B}}$ and $T_{\interpret{\Sigma}} \interpret{A}$, respectively.
By factoring the above morphism through the evaluation morphism $T_{\interpret{\Sigma}} \interpret{A} \times (\KleisliExp{T_{\interpret{\Sigma'}}}{\interpret{A}}{J \interpret{B}}) \to T_{\interpret{\Sigma}} T_{\interpret{\Sigma'}} \interpret{B}$, we can eliminate $A$ and instead consider a morphism of the following form:
\[ \mathbf{handle} : \HBundle{\interpret{\Sigma}}{J \interpret{B}} \times T_{\interpret{\Sigma}} T_{\interpret{\Sigma'}} \interpret{B} \to T_{\interpret{\Sigma'}} \interpret{B} \]

The morphism $\mathbf{handle}$ is required to satisfy certain equations to capture the equations for effect handlers in Figure~\ref{fig:effect-handler-equations}, as we will formalize in Definition~\ref{def:effect-handler-model}.

\paragraph{Models and interpretations of $\EffectHandlerCalculus$}
By summarizing the discussions above, we define models of $\EffectHandlerCalculus$ as follows.

\begin{figure*}[tb]
	\centering
	\begin{tikzcd}
		\HBundle{S}{X} \times K^{T_{S'}} X \ar[r, "\identity{} \times \eta^{T_S}"] \ar[rd, swap, "\pi_2"] & \HBundle{S}{X} \times T_S K^{T_{S'}} X \ar[d, "\mathbf{handle}_{S, S', X}"] \\
		& K^{T_{S'}} X
	\end{tikzcd}
	\hspace{5em}
	\begin{tikzcd}
		\HBundle{S}{X} \times {T_S}^2 K^{T_{S'}} X \ar[r, "\identity{} \times \mu^{T_S}"] \ar[d, "{\langle \pi_1, \strength^{T_S} \rangle}"] & \HBundle{S}{X} \times T_S K^{T_{S'}} X \ar[dd, "\mathbf{handle}_{S, S', X}"] \\
		\HBundle{S}{X} \times T_S (\HBundle{S}{X} \times T_S K^{T_{S'}} X) \ar[d, "\identity{} \times T_S \mathbf{handle}_{S, S', X}"] \\
		\HBundle{S}{X} \times T_S K^{T_{S'}} X \ar[r, "\mathbf{handle}_{S, S', X}"] & K^{T_{S'}} X
	\end{tikzcd} \\[2em]
	\begin{tikzcd}
		\HBundle{S}{X} \times (A_{\mathtt{op}} \times (\KleisliExp{T_{S'}}{B_{\mathtt{op}}}{X})) \ar[r, "\identity{} \times \MixedAlgOp{S'}{X}{\interpret{\mathtt{op}}^T_S}"] \ar[d, "(\pi_{\mathtt{op}} \times \identity{})"] & \HBundle{S}{X} \times T_S K^{T_{S'}} X \ar[d, "\mathbf{handle}_{S, S', X}"] \\
		(\KleisliExp{T_{S'}}{(A_{\mathtt{op}} \times (\KleisliExp{T_{S'}}{B_{\mathtt{op}}}{X}))}{X}) \times (A_{\mathtt{op}} \times (\KleisliExp{T_{S'}}{B_{\mathtt{op}}}{X})) \ar[r, "\eval"] & K^{T_{S'}} X
	\end{tikzcd}
	\caption{Commutative diagrams for \eqref{eq:handle-unit}, \eqref{eq:handle-multiplication}, and \eqref{eq:handle-operation}.}
	\label{fig:effect-handling-strong-monad}
\end{figure*}

\begin{definition}[$\EffectHandlerCalculus$-model]\label{def:effect-handler-model}
	A \emph{$\EffectHandlerCalculus$-model} $(\category{C}, T, \interpret{{-}}, \mathbf{handle})$ is a tuple consisting of the following data:
	\begin{itemize}
		\item A cartesian category $\category{C}$.
		\item A family of strong monads $\{ T_S \}_{S \in \SemanticEffectType{\category{C}}}$ on $\category{C}$ indexed by semantic signatures.
		\item For each semantic signature $S \in \SemanticEffectType{\category{C}}$ and $(\mathtt{op} : A_{\mathtt{op}} \rightarrowtriangle B_{\mathtt{op}}) \in S$, a morphism $\interpret{\mathtt{op}}^T_S : A_{\mathtt{op}} \to T_S B_{\mathtt{op}}$ in $\category{C}$, which we often simply write as $\interpret{\mathtt{op}}$.
		\item For each semantic signature $S, S' \in \SemanticEffectType{\category{C}}$ and $X \in \category{C}_{T_{S'}}$, a morphism in $\category{C}$ of the following form.
		\[ \mathbf{handle}^T_{S, S', X} : \HBundle{S}{X} \times T_S K^{T_{S'}} X \to K^{T_{S'}} X \]
	\end{itemize}
	These data are required to satisfy the following conditions.
	\begin{itemize}
		\item The category $\category{C}$ has Kleisli exponentials for each $T_{S}$.
		\item The morphism $\mathbf{handle}^T_{S, S', X}$ satisfies the following three equations.
		\begin{align}
			&\mathbf{handle}_{S, S', X} \comp (\identity{} \times \eta) \quad=\quad \pi_2 \\
			&\qquad: \HBundle{S}{X} \times K^{T_{S'}} X \to K^{T_{S'}} X \tag{handle-unit} \label{eq:handle-unit} \\
			&\mathbf{handle}_{S, S', X} \comp (\identity{} \times \mu) \\
			&\quad=\quad \mathbf{handle}_{S, S', X} \comp (\identity{} \times T_S \mathbf{handle}_{S, S', X}) \comp \langle \pi_1, \strength \rangle \\
			&\qquad\qquad: \HBundle{S}{X} \times {T_S}^2 K^{T_{S'}} X \to K^{T_{S'}} X \tag{handle-mult} \label{eq:handle-multiplication} \\
			&\mathbf{handle}_{S, S', X} \comp (\identity{} \times \MixedAlgOp{S'}{X}{\interpret{\mathtt{op}}^T_S}) \quad=\quad \eval \comp (\pi_{\mathtt{op}} \times \identity{}) \\
			&\quad: \HBundle{S}{X} \times (A_{\mathtt{op}} \times (\KleisliExp{T_{S'}}{B_{\mathtt{op}}}{X})) \to K^{T_{S'}} X \tag{handle-op} \label{eq:handle-operation}
		\end{align}
	\end{itemize}
	Here, $\pi_{\mathtt{op}} : \HBundle{S}{X} \to \KleisliExp{T_{S'}}{(A_{\mathtt{op}} \times (\KleisliExp{T_{S'}}{B_{\mathtt{op}}}{X}))}{X}$ is the projection, and $\MixedAlgOp{S'}{X}{\interpret{\mathtt{op}}^T_S}$ is a morphism defined as follows.
	\begin{align}
		&\MixedAlgOp{S'}{X}{\interpret{\mathtt{op}}^T_S} \quad\coloneqq\quad T_S \eval \comp \strength^{T_S} \comp \braiding \comp (\interpret{\mathtt{op}}^T_S \times \identity{}) \\
		&\quad\qquad:\quad A_{\mathtt{op}} \times (\KleisliExp{T_{S'}}{B_{\mathtt{op}}}{X}) \to T_S K^{T_{S'}} X
	\end{align}
	Here, $\braiding_{X, Y} : X \times Y \to Y \times X$ is the canonical isomorphism in $\category{C}$.
	Figure~\ref{fig:effect-handling-strong-monad} shows the commutative diagrams corresponding to the above three equations.
	Note that $\MixedAlgOp{S'}{X}{\interpret{\mathtt{op}}^T_S}$ is defined similarly to the algebraic operation corresponding to a generic effect~\cite{PlotkinApplCategStruct2003}, but it mixes two different signatures $S$ and $S'$: if $S = S'$, then $\mu^{T_S} \comp \MixedAlgOp{S}{X}{\interpret{\mathtt{op}}^T_S} : A_{\mathtt{op}} \times (\KleisliExp{T_{S}}{B_{\mathtt{op}}}{X}) \to K^{T_{S}} X$ is the algebraic operation corresponding to the generic effect $\interpret{\mathtt{op}}^T_S : A_{\mathtt{op}} \to T_S B_{\mathtt{op}}$.
\end{definition}

One way to understand the three equations is to see that they correspond to the three equations for effect handlers in Figure~\ref{fig:effect-handler-equations}.
Another way is to see the morphism $\mathbf{handle}^T_{S, S', X}$ as a mapping from an effect handler to an Eilenberg--Moore $T_S$-algebra on $K^{T_{S'}} X$.
Then, two equations \eqref{eq:handle-unit} and \eqref{eq:handle-multiplication} correspond to the laws for Eilenberg--Moore algebras.
The last equation \eqref{eq:handle-operation} states that the interpretation of an operation $\mathtt{op}$ in the Eilenberg--Moore algebra coincides with the interpretation of the corresponding operation clause in the handler.
We will further discuss the relation to Eilenberg--Moore algebras in Section~\ref{sec:em-algebras}.

\begin{definition}[interpretation]\label{def:interpretation-lambda-eff}
	Let $(\category{C}, T, \interpret{{-}}, \mathbf{handle})$ be a $\EffectHandlerCalculus$-model.
	We define the interpretation of $\EffectHandlerCalculus$ as follows.
	The interpretation of value types $\interpret{A} \in \category{C}$ is defined similarly to that for fine-grain call-by-value: function types and product types are interpreted by Kleisli exponentials and products, respectively.
	\[ \interpret{A \to B ! \Sigma} \coloneqq \KleisliExp{T_{\interpret{\Sigma}}}{\interpret{A}}{J^{} \interpret{B}} \qquad
	\interpret{\prod_{i = 1}^n A_i} \coloneqq \prod_{i=1}^n \interpret{A_i} \]
	Computation types are interpreted as $\interpret{A ! \Sigma} \coloneqq J^{T_{\interpret{\Sigma}}} \interpret{A} \in \category{C}_{T_{\interpret{\Sigma}}}$.
	Signatures are interpreted as semantic signatures in $\category{C}$ in the obvious way.
	\[ \interpret{\Sigma} \coloneqq \{ \mathtt{op} : \interpret{A_{\mathtt{op}}} \rightarrowtriangle \interpret{B_{\mathtt{op}}} \mid (\mathtt{op} : A_{\mathtt{op}} \rightarrowtriangle B_{\mathtt{op}}) \in \Sigma \} \]
	The interpretation of value types naturally induces the interpretation of contexts $\interpret{x_1 : A_1, \dots, x_n : A_n} \coloneqq \interpret{A_1} \times \dots \times \interpret{A_n} \in \category{C}$.
	The interpretation of terms and handlers is defined by simultaneous induction on their typing derivations.
	\begin{align}
		&\interpret{\Gamma \vdash V : A} &&: \interpret{\Gamma} \to \interpret{A} \\
		&\interpret{\Gamma \vdash M : A ! \Sigma} &&: \interpret{\Gamma} \to T_{\interpret{\Sigma}} \interpret{A} \\
		&\interpret{\Gamma \vdash H : \Sigma \Rightarrow A ! \Sigma'} &&: \interpret{\Gamma} \to \HBundle{\interpret{\Sigma}}{J^{T_{\interpret{\Sigma'}}} \interpret{A}}
	\end{align}
	We omit most cases for brevity (see Appendix~\ref{sec:proof-of-soundness} for details) and only give the interpretation related to effect handlers here.
	The interpretation of $\Gamma \vdash \mathtt{op}(V) : A ! \Sigma$ is defined by the corresponding morphism in the $\EffectHandlerCalculus$-model: $\interpret{\mathtt{op}(V)} = \interpret{\mathtt{op}}^T_{\Sigma} \comp \interpret{V}$.
	The interpretation of a handler $\Gamma \vdash H : \Sigma \Rightarrow A ! \Sigma'$ is defined by tupling the interpretations of its operation clauses:
	\[ \Lambda \interpret{M_{\mathtt{op}}} : \interpret{\Gamma} \to \KleisliExp{T_{\interpret{\Sigma'}}}{(\interpret{A_{\mathtt{op}}} \times (\KleisliExp{T_{\interpret{\Sigma'}}}{\interpret{B_{\mathtt{op}}}}{J \interpret{A}}))}{J \interpret{A}} \]
	where $(\mathtt{op}(x, k) \mapsto M_{\mathtt{op}}) \in H$.
	Finally, the interpretation of handle-with $\Gamma \vdash \handlewithto{M}{H}{x}{N} : A ! \Sigma'$ is defined using the morphism $\mathbf{handle}_{\interpret{\Sigma}, \interpret{\Sigma'}, J \interpret{A}}$ of the effect-handling strong monad as follows.
	\begin{align}
		&\interpret{\handlewithto{M}{H}{x}{N}} \\
		&\coloneqq\quad \mathbf{handle}_{\interpret{\Sigma}, \interpret{\Sigma'}, J \interpret{A}} \comp \tupling{\interpret{H}}{T_{\interpret{\Sigma}} \interpret{N} \comp \strength^{T_{\interpret{\Sigma}}} \comp \tupling{\identity{}}{\interpret{M}}}
	\end{align}
\end{definition}

As expected, the denotational semantics is sound with respect to the equational theory of $\EffectHandlerCalculus$.
\begin{theorem}[soundness]\label{thm:soundness}
	If $\Gamma \vdash V = W : A$ is derivable, then $\interpret{V} = \interpret{W}$ for any $\EffectHandlerCalculus$-model.
	Similarly, if $\Gamma \vdash M = N : C$ or $\Gamma \vdash H = H' : \Sigma \Rightarrow C$ is derivable, then $\interpret{M} = \interpret{N}$ or $\interpret{H} = \interpret{H'}$, respectively.
	\qed
\end{theorem}

By Proposition~\ref{prop:operational-equational-soundness}, we also have the soundness of the operational semantics with respect to the denotational semantics.
\begin{corollary}
	The operational semantics is sound with respect to the denotational semantics: if $M \leadsto N$, then $\interpret{M} = \interpret{N}$.
	\qed
\end{corollary}

\subsection{Completeness}
The denotational semantics is also complete with respect to the equational theory of $\EffectHandlerCalculus$.
The proof is done by constructing a term model, as in the case of fine-grain call-by-value \cite{LevyInformationandComputation2003}.

\begin{definition}\label{def:term-model-category}
	We define a category $\TermModel{\EffectHandlerCalculus}$ as follows.
	\begin{itemize}
		\item Objects in $\TermModel{\EffectHandlerCalculus}$ are value types $A$.
		\item Morphisms $V : A \to B$ are value terms $x : A \vdash V : B$ modulo equations $x : A \vdash V = W : B$.
		For notational simplicity, we often simply write $V$ instead of its equivalence class $[V]$.
		\item Identity morphisms are given by $x : A \vdash x : A$.
		\item The composition of two morphisms $V : A \to B$ and $W : B \to C$ is defined by substitution $V \comp W = V[W/x]$.
	\end{itemize}
\end{definition}

It is straightforward to see that $\TermModel{\EffectHandlerCalculus}$ is a cartesian category where categorical products are given by product types.

\begin{definition}\label{def:term-model-monad}
	We define a $\EffectHandlerCalculus$-model structure on the cartesian category $\TermModel{\EffectHandlerCalculus}$ as follows.
	\begin{itemize}
		\item For each semantic signature $S \in \SemanticEffectType{\TermModel{\EffectHandlerCalculus}}$, we define a strong monad $T_{S}$ by the following strong Kleisli triples.
		\begin{itemize}
			\item The object map is given by $T_S A \coloneqq \UnitType \to A ! S$.
			Here, we naturally identify a semantic signature $S$ in $\TermModel{\EffectHandlerCalculus}$ with a syntactic signature.
			\item For each type $A$, the unit $\eta_A : A \to T_S A$ is given by
			$x : A \vdash\ \lambda y. \return{x}\ : \UnitType \to A ! S$.
			\item For $x : A_1 \times A_2 \vdash V : \UnitType \to B ! S$, the strong Kleisli extension $x' : A_1 \times (\UnitType \to A_2 ! S) \vdash V^{\sharp} : \UnitType \to B ! S$ is given by
			$\lambda y. \letin{x_2}{\pi_2\ x'\ \langle\rangle}{V[\langle \pi_1\ x', x_2 \rangle/x]\ \langle\rangle}$.
		\end{itemize}
		\item The Kleisli exponential $\KleisliExp{T_{S}}{A}{B}$ is given by $A \to B ! S$.
		\item The morphism $\mathbf{handle}_{S, S', A}$ is defined by
		\begin{align}
			&x : \HBundle{S}{A!S'} \times (\UnitType \to (\UnitType \to A ! S') ! S) \\
			&\vdash \lambda z. \handlewithto{\pi_2\ x\ \langle\rangle}{H}{r}{r\ \langle\rangle} \quad: \UnitType \to A ! S'
		\end{align}
		where $H$ is an effect handler induced from $\pi_1\ x : \HBundle{S}{A!S'}$.
		\[ H\ \coloneqq\ \{ \mathtt{op}(y, k) \mapsto \pi_{\mathtt{op}}\ (\pi_1\ x)\ \langle y, k \rangle \mid (\mathtt{op} : A_{\mathtt{op}} \rightarrowtriangle B_{\mathtt{op}}) \in S \} \]
		\item The interpretation of operations is given as
		$x : A_{\mathtt{op}} \vdash \lambda y. \mathtt{op}(x) : \UnitType \to B_{\mathtt{op}} ! S$.
	\end{itemize}
\end{definition}

\begin{theorem}\label{thm:term-model}
	Definition~\ref{def:term-model-monad} indeed gives a $\EffectHandlerCalculus$-model.
	\qed
\end{theorem}

By interpreting terms in the term model, we can show the completeness of the denotational semantics.
The detailed proof is given in Appendix~\ref{sec:proof-of-completeness}.
\begin{theorem}[completeness]\label{thm:completeness}
	Let $\Gamma \vdash V : A$ and $\Gamma \vdash W : A$ be value terms such that $\interpret{V} = \interpret{W}$ holds for any $\EffectHandlerCalculus$-model.
	Then, $\Gamma \vdash V = W : A$ is derivable.
	We have similar statements for computation terms and handlers.
	\qed
\end{theorem}

\subsection{Relation to Eilenberg--Moore Algebras}\label{sec:em-algebras}

Assume that $\category{C}$ is a cartesian closed category with equalisers.
Then, we can rephrase $\EffectHandlerCalculus$-models using Eilenberg--Moore algebras.

\begin{definition}
	Let $T$ be a strong monad on $\category{C}$ and $X \in \category{C}$.
	We define $e : \mathbf{EM}(T; X) \to \exponential{T X}{X}$ as a joint equaliser of the following two diagrams.
	\begin{center}
		\begin{tikzcd}[row sep=small]
			\exponential{T X}{X} \ar[rd, swap, "\exponential{\eta}{X}"] \ar[r, "!"] & 1 \ar[d, "i"{inner sep=1ex}] \\
			& \exponential{X}{X} \\
		\end{tikzcd}

		\begin{tikzcd}[row sep=small]
			\exponential{T X}{X} \ar[rd, swap, "\exponential{\mu}{X}"] \ar[r, "{\tupling{\identity{}}{\lfloor T \rfloor}}"] & (\exponential{T X}{X}) \times (\exponential{T^2 X}{T X}) \ar[d, "m"{inner sep=1ex}] \\
			& \exponential{T^2 X}{X}
		\end{tikzcd}
	\end{center}
	Here, $i_X : 1 \to \exponential{X}{X}$ and $m_{X, Y, Z} : (\exponential{Y}{Z}) \times (\exponential{X}{Y}) \to \exponential{X}{Z}$ are identity and composition morphisms for internal homs, respectively; and $\lfloor T \rfloor : \exponential{X}{Y} \to \exponential{T X}{T Y}$ is the application of the strong functor $T$ to internal homs.
\end{definition}
Intuitively, $\mathbf{EM}(T; X)$ is a subobject of the internal hom $\exponential{T X}{X}$ that collects Eilenberg--Moore $T$-algebras on $X$.
If $\category{C} = \Set$, then $\mathbf{EM}(T; X)$ is the set of Eilenberg--Moore $T$-algebras on $X$.

\begin{proposition}\label{prop:model-via-em-algebras}
	If $\category{C}$ is a cartesian closed category with equalisers, then there is a bijective correspondence between $\EffectHandlerCalculus$-models on $\category{C}$ and the following data:
	\begin{itemize}
		\item A family of strong monads $\{ T_{S} \}_{S \in \SemanticEffectType{\category{C}}}$ on $\category{C}$.
		\item For each semantic signature $S \in \SemanticEffectType{\category{C}}$ and $(\mathtt{op} : A_{\mathtt{op}} \rightarrowtriangle B_{\mathtt{op}}) \in S$, a morphism $\interpret{\mathtt{op}}^T_S : A_{\mathtt{op}} \to T_S B_{\mathtt{op}}$ in $\category{C}$.
		\item For each semantic signature $S, S' \in \SemanticEffectType{\category{C}}$ and $X \in \category{C}_{T_{S'}}$, a morphism $h_{S, S', X} : \HBundle{S}{X} \to \mathbf{EM}(T_S; K X)$ such that
		$\pi_{\mathtt{op}} = (\exponential{\MixedAlgOp{S'}{X}{\interpret{\mathtt{op}}^T_S}}{K X}) \comp e \comp h$
		for each $(\mathtt{op} : A_{\mathtt{op}} \rightarrowtriangle B_{\mathtt{op}}) \in S$.
		\qed
	\end{itemize}
\end{proposition}
\begin{appendixproof}[Proof of Proposition~\ref{prop:model-via-em-algebras}]
	Since $\category{C}$ is cartesian closed, $\category{C}$ has Kleisli exponentials for any strong monad on $\category{C}$.
	It remains to show the correspondence between $\mathbf{handle}_{E, E', X}$ and $h_{E, E', X}$.
	It is straightforward to show that the three equations in Definition~\ref{def:effect-handler-model} are equivalent to the following equations on $\Lambda \mathbf{handle}_{E, E', X} : \HBundle{E}{X} \to \exponential{T_E K X}{K X}$.
	\begin{align}
		&(\exponential{\eta}{K X}) \comp \Lambda \mathbf{handle}_{E, E', X} \quad=\quad i_{K X} \comp {!} \\
		&(\exponential{\mu}{K X}) \comp \Lambda \mathbf{handle}_{E, E', X}\ =\ m \comp \tupling{\identity{}}{\lfloor T \rfloor} \comp \Lambda \mathbf{handle}_{E, E', X} \\
		&(\exponential{\MixedAlgOp{E'}{X}{\interpret{\mathtt{op}}^T_E}}{K X}) \comp \Lambda \mathbf{handle}_{E, E', X} \quad=\quad \pi_{\mathtt{op}}
	\end{align}
	Thus, given $\mathbf{handle}_{E, E', X}$, we can factor $\Lambda \mathbf{handle}_{E, E', X}$ through the equaliser $e : \mathbf{EM}(T_E; K X) \to \exponential{T_E K X}{K X}$ to obtain $h_{E, E', X} : \HBundle{E}{X} \to \mathbf{EM}(T_E; K X)$.
	Conversely, given $h_{E, E', X}$, we can define $\mathbf{handle}_{E, E', X}$ by $\Lambda^{-1}(e \comp h_{E, E', X})$.
\end{appendixproof}

\subsection{Adding Base Types and Sum Types}\label{sec:calculus-extensions}

In Section~\ref{sec:syntax}, the language $\EffectHandlerCalculus$ is intentionally kept minimal to focus on the core concepts of effect handlers.
To make the language more practical, we consider a couple of extensions here.

\paragraph{Base types and primitive operations}
It is straightforward to extend $\EffectHandlerCalculus$ with base types such as integer and boolean, along with constants and primitive operations:
\[ A \coloneqq \cdots \mid b \qquad\qquad V \coloneqq \cdots \mid f(V_1, \dots, V_n) \]
Here, $b \in \mathbf{Base}$ ranges over base types and $f : \prod_{i = 1}^n A_i \to B$ ranges over constants and primitive operations.
Here, constants can be seen as primitive operations of type $\UnitType \to B$.
To interpret this extended language, we add interpretations of these base types and constants to the definition of $\EffectHandlerCalculus$-models.
A \emph{model} of $\EffectHandlerCalculus$ extended with base types and primitive operations is a $\EffectHandlerCalculus$-model $(\category{C}, T, \interpret{{-}}, \mathbf{handle})$ where $\interpret{-}$ is extended so that it assigns $\interpret{b} \in \category{C}$ for each base type $b \in \mathbf{Base}$ and a morphism $\interpret{f} : \prod_{i = 1}^n \interpret{A_i} \to \interpret{B}$ in $\category{C}$ for each primitive operation $f : \prod_{i = 1}^n A_i \to B$, in addition to interpretations of operations as in Definition~\ref{def:effect-handler-model}.
Then, it is straightforward to extend the soundness and completeness theorems to this setting.

\paragraph{Sum types}
Adding sum types is also straightforward.
We define $\EffectHandlerCalculus^{+}$ as the extension of $\EffectHandlerCalculus$ with sum types:
\begin{align}
	A &\coloneqq \cdots \mid A + B \qquad\qquad
	V \coloneqq \cdots \mid \iota_1\ V \mid \iota_2\ V \\
	M &\coloneqq \cdots \mid \caseof{V}{\casepattern{\iota_1\ x}{M};\ \casepattern{\iota_2\ y}{N}}
\end{align}
For simplicity, we only consider binary sum types here, but the extension to $n$-ary sum types $A_1 + \dots + A_n$ is straightforward.
We extend the equational theory of $\EffectHandlerCalculus$ by the following $\beta$/$\eta$-laws for sum types.
\begin{align}
	\forall i \in \{1, 2\}, \quad \caseof{\iota_i\ V}{\casepattern{\iota_1\ x}{M_1};\ \casepattern{\iota_2\ x}{M_2}} &= M_i[V/x] \\
	\caseof{V}{\casepattern{\iota_1\ y}{M[\iota_1\ y/x]};\ \casepattern{\iota_2\ z}{M[\iota_2\ z/x]}} &= M[V/x]
\end{align}
No additional equations are required for the interaction between sum types and handle-with expressions, since they are derivable from the existing equations.
\begin{proposition}
	The following equation is derivable in $\EffectHandlerCalculus^{+}$.
	\begin{align}
		&\handlewithto{\caseof{V}{\casepattern{\iota_1\ x_1}{M_1};\ \casepattern{\iota_2\ x_2}{M_2}}}{H}{y}{N} \\
		&= \caseof{V}{\casepattern{\iota_1\ x_1}{\handlewithto{M_1}{H}{y}{N}};\\
		&\qquad\qquad\qquad \casepattern{\iota_2\ x_2}{\handlewithto{M_2}{H}{y}{N}}}
	\end{align}
\end{proposition}
\begin{proof}
	Let $L \coloneqq \handlewithto{\caseof{z}{\casepattern{\iota_1\ x_1}{M_1};\ \casepattern{\iota_2\ x_2}{M_2}}}{H}{y}{N}$ where $z$ is a fresh variable.
	\begin{align}
		&\handlewithto{\caseof{V}{\casepattern{\iota_1\ x_1}{M_1};\ \casepattern{\iota_2\ x_2}{M_2}}}{H}{y}{N} \\
		&= L[V/z] \\
		&= \caseof{V}{\casepattern{\iota_1\ x_1}{L[\iota_1\ x_1/z]};\ \casepattern{\iota_2\ x_2}{L[\iota_2\ x_2/z]}} \\
		&= \caseof{V}{\casepattern{\iota_1\ x_1}{\handlewithto{M_1}{H}{y}{N}};\\
		&\qquad\qquad\qquad \casepattern{\iota_2\ x_2}{\handlewithto{M_2}{H}{y}{N}}}
		\qedhere
	\end{align}
\end{proof}

In $\EffectHandlerCalculus^{+}$, injections $\iota_i\ V$ are value constructors, whereas case analysis is a computation constructor.
Consequently, we interpret sum types using a structure that slightly differs from ordinary categorical coproducts, which we call \emph{Kleisli coproducts}.
This situation is analogous to that of Kleisli exponentials, which are used to interpret function types:
since lambda abstraction is a value constructor while application is a computation constructor, Kleisli exponentials are defined using the adjunction between $\category{C}$ and the Kleisli category $\category{C}_T$.
Similarly, we define Kleisli coproducts using the same adjunction.
\begin{definition}
	A \emph{Kleisli binary coproduct} in a cartesian category $\category{C}$ with a strong monad $T$ is a diagram
	$X_1 \xrightarrow{\iota_1} X_1 + X_2 \xleftarrow{\iota_2} X_2$
	in $\category{C}$ such that the image under $J^T : \category{C} \to \category{C}_T$ is a coproduct diagram in $\category{C}_T$, i.e., for any $Y$ and morphisms $f_1 : X_1 \to T Y$ and $f_2 : X_2 \to T Y$ in $\category{C}$, there exists a unique morphism $[f_1, f_2] : X_1 + X_2 \to T Y$ such that $[f_1, f_2] \comp \iota_1 = f_1$ and $[f_1, f_2] \comp \iota_2 = f_2$.
	A Kleisli binary coproduct $X_1 \xrightarrow{\iota_1} X_1 + X_2 \xleftarrow{\iota_2} X_2$ is \emph{distributive} if for any $Z \in \category{C}$, $Z \times ({-})$ maps the diagram to a Kleisli binary coproduct diagram $Z \times X_1 \xrightarrow{\identity{Z} \times \iota_1} Z \times (X_1 + X_2) \xleftarrow{\identity{Z} \times \iota_2} Z \times X_2$.
\end{definition}
If $\category{C}$ has binary coproducts, then it has Kleisli binary coproducts for any strong monad $T$ because the left adjoint $J^T$ preserves colimits.
If moreover, $\category{C}$ is distributive, then the Kleisli binary coproducts are distributive.
On the other hand, Kleisli binary coproducts are not necessarily binary coproducts in $\category{C}$: for example, if $T = 1$ is the monad defined by a terminal object $1 \in \category{C}$, then any diagram of the form $X_1 \to Z \leftarrow X_2$ is a Kleisli binary coproduct diagram since every hom-set $\category{C}_T(X, Y)$ is a singleton set.

In a $\EffectHandlerCalculus$-model, we have a family of strong monads, but the interpretation of sum types should not depend on the choice of a strong monad.
To capture this, we define the following notion.
\begin{definition}
	Let $\{T_i\}_{i \in I}$ be a family of strong monads on a cartesian category $\category{C}$.
	We say that a diagram $X_1 \xrightarrow{\iota_1} X_1 + X_2 \xleftarrow{\iota_2} X_2$ is a \emph{joint Kleisli binary coproduct} if it is a Kleisli binary coproduct for each strong monad $T_i$.
	We define \emph{joint distributive Kleisli binary coproduct} similarly.
\end{definition}

We extend $\EffectHandlerCalculus$-models as follows.
\begin{definition}
	A \emph{$\EffectHandlerCalculus^{+}$-model} is a $\EffectHandlerCalculus$-model $(\category{C}, T)$ such that $\category{C}$ has joint distributive Kleisli binary coproducts for $\{ T_S \}_{S \in \SemanticEffectType{\category{C}}}$.
\end{definition}

\begin{proposition}
	The calculus $\EffectHandlerCalculus^{+}$ is sound and complete with respect to $\EffectHandlerCalculus^{+}$-models.
	\qed
\end{proposition}

\begin{remark}
	We could alternatively consider slightly different calculus by adding case analysis as value constructors.
	\[ V \coloneqq \cdots \mid \iota_1\ V \mid \iota_2\ V \mid \caseof{V}{\casepattern{\iota_1\ x}{W_1};\ \casepattern{\iota_2\ y}{W_2}} \]
	In this case, models are defined simply as $\EffectHandlerCalculus$-models where $\category{C}$ is a distributive category.
	However, this makes the operational semantics less natural similarly to the case of pattern matching discussed in Remark~\ref{rem:pattern-matching-product}.	
\end{remark}

\begin{toappendix}
\section{Proof of Soundness}
\label{sec:proof-of-soundness}
\subsection{Details of the Definition of the Interpretation}

\begin{definition}
	The interpretation of value types $A$, computation types $C$, and signature $\Sigma$ is defined by simultaneous induction as follows.
	\begin{itemize}
		\item Value types $\interpret{A} \in \category{C}$:
		\[ \interpret{A \to C} \coloneqq \KleisliExp{T_{\interpret{\Sigma}}}{\interpret{A}}{\interpret{C}} \qquad
		\interpret{\prod_{i = 1}^n A_i} \coloneqq \prod_{i=1}^n \interpret{A_i} \]
		\item Computation types $\interpret{A ! \Sigma} \in \category{C}^{T_{\interpret{\Sigma}}}$:
		\[ \interpret{A ! \Sigma} \coloneqq J^{T_{\interpret{\Sigma}}} \interpret{A} \]
		\item Signatures $\interpret{\Sigma} \in \SemanticEffectType{\category{C}}$:
		\begin{gather}
			\interpret{\emptyset} \coloneqq \emptyset \qquad
			\interpret{\{ \mathtt{op} : A_{\mathtt{op}} \rightarrowtriangle B_{\mathtt{op}} \} \cup \Sigma} \coloneqq \{ \mathtt{op} : \interpret{A_{\mathtt{op}}} \rightarrowtriangle \interpret{B_{\mathtt{op}}} \} \cup \interpret{\Sigma}
		\end{gather}
	\end{itemize}
	Contexts $\interpret{\Gamma} \in \category{C}$ are interpreted as follows.
	\begin{gather}
		\interpret{\cdot} \coloneqq 1 \qquad
		\interpret{\Gamma, x : A} \coloneqq \interpret{\Gamma} \times \interpret{A}
	\end{gather}
\end{definition}

\begin{definition}
	The interpretation of terms and handlers is defined by simultaneous induction on typing derivations as follows.
	\begin{itemize}
		\item Value terms $\interpret{\Gamma \vdash V : A} : \interpret{\Gamma} \to \interpret{A}$ (often written as $\interpret{V}$ for brevity):
		\begin{align}
			\interpret{\Gamma, y : B \vdash x : A} \quad&\coloneqq\quad \begin{cases}
				\pi_2 & x = y \\
				\interpret{\Gamma \vdash x : A} \comp \pi_1 & x \neq y
			\end{cases} \\
			\interpret{\langle V_1, \dots, V_n \rangle} \quad&\coloneqq\quad \langle \interpret{V_1}, \dots, \interpret{V_n} \rangle \\
			\interpret{\pi_i\ V} \quad&\coloneqq\quad \pi_i \comp \interpret{V} \\
			\interpret{\lambda x : A. M} \quad&\coloneqq\quad \Lambda(\interpret{M})
		\end{align}
		\item Computation terms $\interpret{\Gamma \vdash M : A ! \Sigma} : \interpret{\Gamma} \to T_{\interpret{\Sigma}} \interpret{A}$:
		\begin{align}
			\interpret{V\ W} \quad&\coloneqq\quad \eval \comp \tupling{\interpret{V}}{\interpret{W}} \\
			\interpret{\mathtt{return}\ V} \quad&\coloneqq\quad \eta^{T_{\interpret{\Sigma}}} \comp \interpret{V} \\
			\interpret{\letin{x}{M}{N}} \quad&\coloneqq\quad \mu^{T_{\interpret{\Sigma}}} \comp T_{\interpret{\Sigma}} \interpret{N} \comp \strength^{T_{\interpret{\Sigma}}} \comp \tupling{\identity{}}{\interpret{M}}\ =\ \interpret{N}^{\sharp} \comp \tupling{\identity{}}{\interpret{M}} \\
			\interpret{\mathtt{op}(V)} \quad&\coloneqq\quad \interpret{\mathtt{op}}^T_{\interpret{\Sigma}} \comp \interpret{V} \\
			\interpret{\handlewithto{M}{H}{x}{N}} \quad&\coloneqq\quad \mathbf{handle}_{\interpret{\Sigma}, \interpret{\Sigma'}, J \interpret{B}} \comp \tupling{\interpret{H}}{T_{\interpret{\Sigma}} \interpret{N} \comp \strength^{T_{\interpret{\Sigma}}} \comp \tupling{\identity{}}{\interpret{M}}}
		\end{align}
		In the definition of $\interpret{\handlewithto{M}{H}{x}{N}}$, we assume $\Gamma \vdash \handlewithto{M}{H}{x}{N} : B ! \Sigma'$, and morphisms in the right-hand side are $\mathbf{handle}_{\interpret{\Sigma}, \interpret{\Sigma'}, J \interpret{B}} : \HBundle{\interpret{\Sigma}}{J \interpret{B}} \times T_{\interpret{\Sigma}} T_{\interpret{\Sigma'}} \interpret{B} \to T_{\interpret{\Sigma'}} \interpret{B}$ and
		\[ \interpret{\Gamma} \xrightarrow{\tupling{\identity{}}{\interpret{M}}} \interpret{\Gamma} \times T_{\interpret{\Sigma}} \interpret{A} \xrightarrow{\strength^{T_{\interpret{\Sigma}}}} T_{\interpret{\Sigma}} (\interpret{\Gamma} \times \interpret{A}) \xrightarrow{T_{\interpret{\Sigma}} \interpret{N}} T_{\interpret{\Sigma}} T_{\interpret{\Sigma'}} \interpret{B}. \]
		\item Handlers $\interpret{\Gamma \vdash H : \Sigma \Rightarrow C} : \interpret{\Gamma} \to \HBundle{\interpret{\Sigma}}{\interpret{C}}$:
		\[ \interpret{H} \quad\coloneqq\quad \langle \Lambda (\interpret{M_{\mathtt{op}}} \comp \associator^{-1}) \rangle_{(\mathtt{op} : A_{\mathtt{op}} \rightarrowtriangle B_{\mathtt{op}}) \in \Sigma} \]
		where $\associator_{X, Y, Z} : (X \times Y) \times Z \to X \times (Y \times Z)$ is the associator of the cartesian monoidal category $\category{C}$ and $\interpret{M_{\mathtt{op}}} : (\interpret{\Gamma} \times \interpret{A_{\mathtt{op}}}) \times (\KleisliExp{T_{\interpret{\Sigma'}}}{\interpret{B_{\mathtt{op}}}}{J \interpret{B}}) \to T_{\interpret{\Sigma'}} \interpret{B}$.
		That is, for each $\mathtt{op}$, we have a morphism
		\[ \Lambda (\interpret{M_{\mathtt{op}}} \comp \associator^{-1}) : \interpret{\Gamma} \to \KleisliExp{T_{\interpret{\Sigma'}}}{(\interpret{A_{\mathtt{op}}} \times (\KleisliExp{T_{\interpret{\Sigma'}}}{\interpret{B_{\mathtt{op}}}}{J \interpret{B}}))}{J \interpret{B}} \]
		and $\interpret{H}$ is the tupling of those morphisms.
	\end{itemize}
\end{definition}

\subsection{Proofs}
The non-trivial part is the three equations for effect handlers in Figure~\ref{fig:effect-handler-equations}.
Other equations (Figure~\ref{fig:beta-eta-monad-equations}) are standard ones for fine-grain call-by-value \cite{LevyInformationandComputation2003}.

\begin{lemma}[weakening]\label{lem:weakening-interpretation}
	Let $\Gamma, \Delta \vdash V : A$ be a well-typed term.
	The interpretation of $\Gamma, x : B, \Delta \vdash V : A$ is given by
	\[ \interpret{\Gamma, x : B, \Delta \vdash V : A} = \interpret{\Gamma, \Delta \vdash V : A} \comp \mathrm{proj}_{\Gamma; x : B; \Delta} \]
	where $\mathrm{proj}_{\Gamma; x : B; \Delta} : \interpret{\Gamma, x : B, \Delta} \to \interpret{\Gamma, \Delta}$ is the canonical projection defined by
	\[ \mathrm{proj}_{\Gamma; x : B; \cdot} = \pi_1 \qquad \mathrm{proj}_{\Gamma; x : B; \Delta, y : B'} = \mathrm{proj}_{\Gamma; x : B; \Delta} \times \identity{} \]
	We also have similar equations for computations and handlers.
\end{lemma}
\begin{proof}
	By straightforward induction.
\end{proof}

\begin{lemma}[substitution]\label{lem:subst-interpretation}
	For $\Gamma \vdash V_1 : A_1, \dots, \Gamma \vdash V_n : A_n$ and $x_1 : A_1, \dots, x_n : A_n \vdash W : B$, the interpretation of $\Gamma \vdash W[V_1/x_1, \dots, V_n/x_n] : B$ is given by
	\[ \interpret{W[V_1/x_1, \dots, V_n/x_n]} = \interpret{W} \comp \interpret{V_1, \dots, V_n} \]
	where $\interpret{V_1, \dots, V_n} : \interpret{\Gamma} \to \interpret{x_1 : A_1, \dots, x_n : A_n}$ is defined by
	\[ \interpret{\cdot} = {!} \qquad \interpret{V_1, \dots, V_n, V_{n+1}} = \tupling{\interpret{V_1, \dots, V_n}}{\interpret{V_{n+1}}} \]
	We also have similar equations for computation terms $x_1 : A_1, \dots, x_n : A_n \vdash M : C$ and handlers $x_1 : A_1, \dots, x_n : A_n \vdash H : E \Rightarrow C$.
\end{lemma}
\begin{proof}
	By straightforward induction.
	We use Lemma~\ref{lem:weakening-interpretation}.
\end{proof}

\begin{lemma}\label{lem:handle-return-sound}
	\eqref{eq:handle-return} is sound.
\end{lemma}
\begin{proof}
	We use Lemma~\ref{lem:subst-interpretation}.
	Let $\Sigma$ be the signature such that $\Gamma \vdash H : \Sigma \Rightarrow C$.
	For simplicity, we write $\interpret{\Sigma}$ as $\Sigma$.
	Then, we have the following.
	\begin{align}
		&\interpret{\handlewithto{\mathbf{return}\ V}{H}{x}{M}} \\
		&= \mathbf{handle} \comp \tupling{\interpret{H}}{T_\Sigma \interpret{M} \comp \strength^{T_\Sigma} \comp \tupling{\identity{}}{\eta^{T_\Sigma} \comp \interpret{V}}} \\
		&= \mathbf{handle} \comp \tupling{\interpret{H}}{\eta^{T_\Sigma} \comp \interpret{M} \comp \tupling{\identity{}}{\interpret{V}}} \\
		&= \pi_2 \comp \tupling{\interpret{H}}{\interpret{M} \comp \tupling{\identity{}}{\interpret{V}}} \\
		&= \interpret{M} \comp \tupling{\identity{}}{\interpret{V}} \\
		&= \interpret{M[V/x]}
		\qedhere
	\end{align}
\end{proof}

\begin{lemma}\label{lem:handle-let-sound}
	\eqref{eq:handle-let} is sound.
\end{lemma}
\begin{proof}
	We use Lemma~\ref{lem:weakening-interpretation}.
	For simplicity, we write $\interpret{\Sigma}$ as $\Sigma$.
	The left-hand side can be rewritten as follows.
	\allowdisplaybreaks
	\begin{align}
		&\interpret{\handlewithto{(\letin{x}{L}{M})}{H}{y}{N}} \\
		&= \mathbf{handle} \comp \tupling{\interpret{H}}{T_\Sigma \interpret{N} \comp \strength \comp \tupling{\identity{}}{\mu \comp T_\Sigma \interpret{M} \comp \strength \comp \tupling{\identity{}}{\interpret{L}}}} \\
		&= \mathbf{handle} \comp \tupling{\interpret{H}}{T_\Sigma \interpret{N} \comp \strength \comp (\identity{} \times \mu) \comp \tupling{\identity{}}{T_\Sigma \interpret{M} \comp \strength \comp \tupling{\identity{}}{\interpret{L}}}} \\
		&= \mathbf{handle} \comp \tupling{\interpret{H}}{T_\Sigma \interpret{N} \comp \mu \comp T_\Sigma \strength \comp \strength \comp \tupling{\identity{}}{T_\Sigma \interpret{M} \comp \strength \comp \tupling{\identity{}}{\interpret{L}}}} \\
		&= \mathbf{handle} \comp \tupling{\interpret{H}}{\mu \comp (T_\Sigma)^2 \interpret{N} \comp T_\Sigma \strength \comp \strength \comp \tupling{\identity{}}{T_\Sigma \interpret{M} \comp \strength \comp \tupling{\identity{}}{\interpret{L}}}} \\
		&= \mathbf{handle} \comp (\identity{} \times \mu) \comp \tupling{\interpret{H}}{(T_\Sigma)^2 \interpret{N} \comp T_\Sigma \strength \comp \strength \comp \tupling{\identity{}}{T_\Sigma \interpret{M} \comp \strength \comp \tupling{\identity{}}{\interpret{L}}}} \\
		&= \mathbf{handle} \comp (\identity{} \times T_\Sigma \mathbf{handle}) \comp \tupling{\pi_1}{\strength} \comp \tupling{\interpret{H}}{(T_\Sigma)^2 \interpret{N} \comp T_\Sigma \strength \comp \strength \comp \tupling{\identity{}}{T_\Sigma \interpret{M} \comp \strength \comp \tupling{\identity{}}{\interpret{L}}}} \\
		&= \mathbf{handle} \comp (\identity{} \times T_\Sigma \mathbf{handle}) \comp \tupling{\interpret{H}}{\strength \comp \tupling{\interpret{H}}{(T_\Sigma)^2 \interpret{N} \comp T_\Sigma \strength \comp \strength \comp \tupling{\identity{}}{T_\Sigma \interpret{M} \comp \strength \comp \tupling{\identity{}}{\interpret{L}}}}}
	\end{align}
	On the other hand, the right-hand side can be rewritten as follows.
	\begin{align}
		&\interpret{\handlewithto{L}{H}{x}{\handlewithto{M}{H}{y}{N}}} \\
		&= \mathbf{handle} \comp \tupling{\interpret{H}}{T_\Sigma \mathbf{handle} \comp T_\Sigma \tupling{\interpret{H} \comp \pi_1}{T_\Sigma (\interpret{N} \comp (\pi_1 \times \identity{})) \comp \strength \comp \tupling{\identity{}}{\interpret{M}}} \comp \strength \comp \tupling{\identity{}}{\interpret{L}}} \\
		&= \mathbf{handle} \comp (\identity{} \times T_\Sigma \mathbf{handle}) \comp \tupling{\interpret{H}}{T_\Sigma \tupling{\interpret{H} \comp \pi_1}{T_\Sigma (\interpret{N} \comp (\pi_1 \times \identity{})) \comp \strength \comp \tupling{\identity{}}{\interpret{M}}} \comp \strength \comp \tupling{\identity{}}{\interpret{L}}}
	\end{align}
	Thus, it suffices to show that the following equation holds.
	\begin{align}
		&\strength \comp \tupling{\interpret{H}}{(T_\Sigma)^2 \interpret{N} \comp T_\Sigma \strength \comp \strength \comp \tupling{\identity{}}{T_\Sigma \interpret{M} \comp \strength \comp \tupling{\identity{}}{\interpret{L}}}} \\
		&= T_\Sigma \tupling{\interpret{H} \comp \pi_1}{T_\Sigma \interpret{N} \comp T_\Sigma (\pi_1 \times \identity{}) \comp \strength \comp \tupling{\identity{}}{\interpret{M}}} \comp \strength \comp \tupling{\identity{}}{\interpret{L}}
	\end{align}

	This is shown as follows.
	\begin{align}
		&\strength \comp \tupling{\interpret{H}}{(T_\Sigma)^2 \interpret{N} \comp T_\Sigma \strength \comp \strength \comp \tupling{\identity{}}{T_\Sigma \interpret{M} \comp \strength \comp \tupling{\identity{}}{\interpret{L}}}} \\
		&= (T_\Sigma (\interpret{H} \times T_\Sigma \interpret{N})) \comp \strength \comp \tupling{\identity{}}{T_\Sigma \strength \comp \strength \comp \tupling{\identity{}}{T_\Sigma \interpret{M} \comp \strength \comp \tupling{\identity{}}{\interpret{L}}}} \\
		&= (T_\Sigma (\interpret{H} \times T_\Sigma \interpret{N})) \comp T_\Sigma (\identity{} \times \strength) \comp \strength \comp (\identity{} \times \strength) \comp \associator \comp \tupling{\tupling{\identity{}}{\identity{}}}{T_\Sigma \interpret{M} \comp \strength \comp \tupling{\identity{}}{\interpret{L}}} \\
		&= (T_\Sigma (\interpret{H} \times T_\Sigma \interpret{N})) \comp T_\Sigma (\identity{} \times \strength) \comp T_\Sigma \associator \comp \strength \comp (\tupling{\identity{}}{\identity{}} \times \identity{}) \comp \tupling{\identity{}}{T_\Sigma \interpret{M} \comp \strength \comp \tupling{\identity{}}{\interpret{L}}} \\
		&= (T_\Sigma (\interpret{H} \times T_\Sigma \interpret{N})) \comp T_\Sigma (\identity{} \times \strength) \comp T_\Sigma \tupling{\pi_1}{\identity{}} \comp \strength \comp \tupling{\identity{}}{T_\Sigma \interpret{M} \comp \strength \comp \tupling{\identity{}}{\interpret{L}}} \\
		&= (T_\Sigma (\interpret{H} \times T_\Sigma \interpret{N})) \comp T_\Sigma \tupling{\pi_1}{\strength} \comp \strength \comp \tupling{\identity{}}{T_\Sigma \interpret{M} \comp \strength \comp \tupling{\identity{}}{\interpret{L}}} \\
		&= (T_\Sigma (\interpret{H} \times T_\Sigma \interpret{N})) \comp T_\Sigma \tupling{\pi_1}{\strength} \comp T_\Sigma (\identity{} \times \interpret{M}) \comp \strength \comp \tupling{\identity{}}{\strength \comp \tupling{\identity{}}{\interpret{L}}} \\
		&= (T_\Sigma (\interpret{H} \times T_\Sigma \interpret{N})) \comp T_\Sigma \tupling{\pi_1}{\strength \comp (\identity{} \times \interpret{M})} \comp \strength \comp \tupling{\identity{}}{\strength \comp \tupling{\identity{}}{\interpret{L}}} \\
		&= (T_\Sigma (\interpret{H} \times T_\Sigma \interpret{N})) \comp T_\Sigma \tupling{\pi_1}{\strength \comp(\identity{} \times \interpret{M})} \comp \strength \comp (\identity{} \times \strength) \comp \associator \comp \tupling{\tupling{\identity{}}{\identity{}}}{\interpret{L}} \\
		&= (T_\Sigma (\interpret{H} \times T_\Sigma \interpret{N})) \comp T_\Sigma \tupling{\pi_1}{\strength \comp (\identity{} \times \interpret{M})} \comp T_\Sigma \associator \comp \strength \comp \tupling{\tupling{\identity{}}{\identity{}}}{\interpret{L}} \\
		&= (T_\Sigma (\interpret{H} \times T_\Sigma \interpret{N})) \comp T_\Sigma \tupling{\pi_1}{\strength \comp (\identity{} \times \interpret{M})} \comp T_\Sigma \associator \comp T_\Sigma (\tupling{\identity{}}{\identity{}} \times \identity{}) \comp \strength \comp \tupling{\identity{}}{\interpret{L}} \\
		&= (T_\Sigma (\interpret{H} \times T_\Sigma \interpret{N})) \comp T_\Sigma \tupling{\pi_1}{\strength \comp (\identity{} \times \interpret{M})} \comp T_\Sigma \tupling{\pi_1}{\identity{}} \comp \strength \comp \tupling{\identity{}}{\interpret{L}} \\
		&= (T_\Sigma (\interpret{H} \times T_\Sigma \interpret{N})) \comp T_\Sigma \tupling{\pi_1}{\strength \comp \tupling{\pi_1}{\interpret{M}}} \comp \strength \comp \tupling{\identity{}}{\interpret{L}} \\
		&= T_\Sigma \tupling{\interpret{H} \comp \pi_1}{T_\Sigma \interpret{N} \comp \strength \comp (\pi_1 \times \identity{}) \comp \tupling{\identity{}}{\interpret{M}}} \comp \strength \comp \tupling{\identity{}}{\interpret{L}} \\
		&= T_\Sigma \tupling{\interpret{H} \comp \pi_1}{T_\Sigma \interpret{N} \comp T_\Sigma (\pi_1 \times \identity{}) \comp \strength \comp \tupling{\identity{}}{\interpret{M}}} \comp \strength \comp \tupling{\identity{}}{\interpret{L}}
		\qedhere
	\end{align}
\end{proof}

\begin{lemma}\label{lem:handle-op-aux}
	Let $\Gamma \vdash V : A$ be a well-typed value term, $\Gamma, x : B \vdash M : C$ be a well-typed computation term, and $\mathtt{op} : A \rightarrowtriangle B$ be an operation symbol in $\Sigma$.
	Then, the following equation holds.
	\[ T_\Sigma \interpret{M} \comp \strength \comp \tupling{\identity{}}{\interpret{\mathtt{op}} \comp \interpret{V}} \quad=\quad \MixedAlgOp{\Sigma'}{X}{\interpret{\mathtt{op}}} \comp \tupling{\interpret{V}}{\Lambda \interpret{M}} \]
\end{lemma}
\begin{proof}
	\begin{align}
		&\MixedAlgOp{\Sigma'}{X}{\interpret{\mathtt{op}}} \comp \tupling{\interpret{V}}{\Lambda \interpret{M}} \\
		&= T \eval \comp \strength \comp \braiding \comp (\interpret{\mathtt{op}} \times \identity{}) \comp \tupling{\interpret{V}}{\Lambda \interpret{M}} \\
		&= T \eval \comp \strength \comp \tupling{\Lambda \interpret{M}}{\interpret{\mathtt{op}} \comp \interpret{V}} \\
		&= T \eval \comp T (\Lambda \interpret{M} \times \identity{}) \comp \strength \comp \tupling{\identity{}}{\interpret{\mathtt{op}} \comp \interpret{V}} \\
		&= T \interpret{M} \comp \strength \comp \tupling{\identity{}}{\interpret{\mathtt{op}} \comp \interpret{V}}
		\qedhere
	\end{align}
\end{proof}

\begin{lemma}\label{lem:handle-op-sound}
	\eqref{eq:handle-op} is sound.
\end{lemma}
\begin{proof}
	We use Lemma~\ref{lem:subst-interpretation}.
	\begin{align}
		&\interpret{\handlewithto{\mathtt{op}(V)}{H}{x}{M}} \\
		&= \mathbf{handle} \comp \tupling{\interpret{H}}{T_\Sigma \interpret{M} \comp \strength \comp \tupling{\identity{}}{\interpret{\mathtt{op}} \comp \interpret{V}}} \\
		&= \mathbf{handle} \comp (\identity{} \times \MixedAlgOp{\Sigma'}{X}{\interpret{\mathtt{op}}}) \comp \tupling{\interpret{H}}{\tupling{\interpret{V}}{\Lambda \interpret{M}}} \\
		&= \eval \comp (\pi_{\mathtt{op}} \times \identity{}) \comp \tupling{\interpret{H}}{\tupling{\interpret{V}}{\Lambda \interpret{M}}} \\
		&= \eval \comp (\Lambda(\interpret{M_{\mathrm{op}}} \comp \associator^{-1}) \times \identity{}) \comp \tupling{\identity{}}{\tupling{\interpret{V}}{\Lambda \interpret{M}}} \\
		&= \interpret{M_{\mathrm{op}}} \comp \associator^{-1} \comp \tupling{\identity{}}{\tupling{\interpret{V}}{\Lambda \interpret{M}}} \\
		&= \interpret{M_{\mathrm{op}}} \comp \tupling{\tupling{\identity{}}{\interpret{V}}}{\Lambda \interpret{M}} \\
		&= \interpret{M_{\mathrm{op}}[V/x, \lambda x. M/k]}
		\qedhere
	\end{align}
\end{proof}

\section{Proof of Completeness}
\label{sec:proof-of-completeness}
\subsection{Category}

\begin{lemma}\label{lem:term-model-category}
	Definition~\ref{def:term-model-category} defines a category.
\end{lemma}
\begin{proof}
	This is divided into several lemmas.
	\begin{itemize}
		\item The composition is well-defined by Lemma~\ref{lem:subst-value-eq-1},\ref{lem:subst-value-eq-2}.
		\item The left/right-unit law is immediate.
		\item Associativity follows from Lemma~\ref{lem:subst-assoc}.
		\[ U \comp (V \comp W) = U[V[W/x]/x] = U[V/x][W/x] = (U \comp V) \comp W \qedhere \]
	\end{itemize}
\end{proof}

\begin{lemma}[substitution for equations]\label{lem:subst-value-eq-1}
	For any value term $W$, the following rule is admissible:
	\begin{mathpar}
		\inferrule{
			x_1 : A_1, \dots, x_n : A_n \vdash W : B \\
			\forall i = 1, \dots, n.\ \Gamma \vdash V_i = V'_i : A_i
		}{
			\Gamma \vdash W[V_1/x_1, \dots, V_n/x_n] = W[V'_1/x_1, \dots, V'_n/x_n] : B
		}
	\end{mathpar}
	The same holds for computation terms and handlers.
\end{lemma}
\begin{proof}
	We prove the statement for value terms, computation terms and handlers simultaneously.
	The proof is by induction on $W$.
	We apply congruence rules for each term constructor.
	Most cases are straightforward.
	If variable binding is involved, we extend $[V_1/x_1, \dots, V_n/x_n]$ to $[V_1/x_1, \dots, V_n/x_n, z/z]$ where $z$ is a variable bound in $W$.
\end{proof}

\begin{lemma}\label{lem:subst-value-eq-2}
	For any value terms $W$ and $W'$, the following rule is admissible:
	\begin{mathpar}
		\inferrule{
			x_1 : A_1, \dots, x_n : A_n \vdash W = W' : B \\
			\forall i = 1, \dots, n.\ \Gamma \vdash V_i : A_i
		}{
			\Gamma \vdash W[V_1/x_1, \dots, V_n/x_n] = W'[V_1/x_1, \dots, V_n/x_n] : B
		}
	\end{mathpar}
	The same holds for computation terms and handlers.
\end{lemma}
\begin{proof}
	We prove the statement for value terms, computation terms and handlers simultaneously.
	The proof is by induction on derivation of $x_1 : A_1, \dots, x_n : A_n \vdash W = W' : B$.
	\begin{itemize}
		\item If $W = W'$ is derived by one of reflexivity, symmetry, transitivity, congruence rules, then the proof is straightforward.
		\item Axioms ($\beta$-/$\eta$-laws, monad laws, and equations for effect handlers) are closed under substitution: if $W = W'$ is an instance of one of those axioms, then $W[V_1/x_1, \dots, V_n/x_n] = W'[V_1/x_1, \dots, V_n/x_n]$ is also an instance of the same axiom.
		\qedhere
	\end{itemize}
\end{proof}

\subsection{Finite Products}

\begin{lemma}\label{lem:term-model-products}
	Finite products in $\TermModel{\EffectHandlerCalculus}$ are given as follows.
	\begin{itemize}
		\item product: $A_1 \times \dots \times A_n$
		\item projection: $x : A_1 \times \dots \times A_n \vdash \pi_i\ x : A_i$
		\item tupling: Given $x : A \vdash V_i : B_i$ for each $i = 1, \dots, n$, we have
		\[ x : A \vdash \langle V_1, \dots, V_n \rangle : B_1 \times \dots \times B_n. \]
	\end{itemize}
\end{lemma}
\begin{proof}
	We show the universal property.
	The tupling satisfies the following equation by the $\beta$-law for products:
	\[ \pi_i \comp \langle V_1, \dots, V_n \rangle = V_i \]
	The uniqueness of such morphism also follows from the $\eta$-law for products.
	If $\pi_i W = V_i$ for each $i$, then
	\[ W = \langle \pi_1\ W, \dots, \pi_n\ W \rangle = \langle V_1, \dots, V_n \rangle. \]
\end{proof}

\subsection{Strong Monads}

\begin{lemma}\label{lem:term-model-monad}
	Definition~\ref{def:term-model-monad} defines a strong monad on $\TermModel{\EffectHandlerCalculus}$.
\end{lemma}
\begin{proof}
	We show that the axioms for strong Kleisli triples hold.
	\allowdisplaybreaks
	\begin{itemize}
		\item $(\eta^T_A \comp \leftunitor)^{\sharp} = \leftunitor$
		\begin{align}
			&(\eta^T_A \comp \leftunitor)^{\sharp} \\
			&=\lambda y. \letin{x_2}{\pi_2\ x'\ \langle\rangle}{(\lambda y. \return{x})[\pi_2\ x/x][\langle \pi_1\ x', x_2 \rangle/x]\ \langle\rangle} \\
			&= \lambda y. \letin{x_2}{\pi_2\ x'\ \langle\rangle}{\return{x_2}} \\
			&= \lambda y. \pi_2\ x'\ \langle\rangle \\
			&= \lambda y. \pi_2\ x'\ y \\
			&= \pi_2\ x' \\
			&= \leftunitor
		\end{align}
		\item Given
		\[ x : A_1 \times A_2 \vdash V : \UnitType \to B ! \Sigma \]
		we show
		\[ x : A_1 \times A_2 \vdash\quad V^{\sharp} \comp (A_1 \times \eta^T_{A_2}) = V \quad: \UnitType \to B ! \Sigma \]
		\begin{align}
			&V^{\sharp} \comp (A_1 \times \eta^T_{A_2}) \\
			&=(\lambda y. \letin{x_2}{\pi_2\ x'\ \langle\rangle}{V[\langle \pi_1\ x', x_2 \rangle/x]\ \langle\rangle})[\langle \pi_1\ x, \lambda y. \return{\pi_2\ x} \rangle/x'] \\
			&= \lambda y. \letin{x_2}{\return{\pi_2\ x}}{V[\langle \pi_1\ x, x_2 \rangle/x]\ \langle\rangle} \\
			&= \lambda y. V[\langle \pi_1\ x, \pi_2\ x \rangle/x]\ \langle\rangle \\
			&= \lambda y. V\ \langle\rangle \\
			&= V
		\end{align}
		\item Given
		\[ x : A_2 \times A_3 \vdash V : \UnitType \to B_1 ! \Sigma \qquad\qquad y : A_1 \times A_2 \vdash W : \UnitType \to B_2 ! \Sigma \]
		we show
		\[ x : (A_1 \times A_2) \times (\UnitType \to A_3 ! \Sigma) \vdash\quad W^{\sharp} \comp (A_1 \times V^{\sharp}) \comp \associator = (W^{\sharp} \comp (A_1 \times V) \comp \associator)^{\sharp} \quad: \UnitType \to B_2 ! \Sigma \]
		The left-hand side:
		\begin{align}
			&W^{\sharp} \comp (A_1 \times V^{\sharp}) \comp \associator \\
			&= (\lambda z. \letin{y_2}{\pi_2\ y'\ \langle\rangle}{W[\langle \pi_1\ y', y_2 \rangle/y]\ \langle\rangle}) \\
			&\qquad [\langle \pi_1\ x, (\lambda z. \letin{x_2}{\pi_2\ x'\ \langle\rangle}{V[\langle \pi_1\ x', x_2 \rangle/x]\ \langle\rangle})[\pi_2\ x/x'] \rangle/y'] \\
			&\qquad [\langle \pi_1\ \pi_1\ x, \langle \pi_2\ \pi_1\ x, \pi_2\ x \rangle \rangle/x] \\
			&= (\lambda z. \letin{y_2}{(\letin{x_2}{\pi_2\ \pi_2\ x\ \langle\rangle}{V[\langle \pi_1\ \pi_2\ x, x_2 \rangle/x]\ \langle\rangle})}{W[\langle \pi_1\ x, y_2 \rangle/y]\ \langle\rangle}) \\
			&\qquad [\langle \pi_1\ \pi_1\ x, \langle \pi_2\ \pi_1\ x, \pi_2\ x \rangle \rangle/x] \\
			&= \lambda z. \letin{y_2}{(\letin{x_2}{\pi_2\ x\ \langle\rangle}{V[\langle \pi_2\ \pi_1\ x, x_2 \rangle/x]\ \langle\rangle})}{W[\langle \pi_1\ \pi_1\ x, y_2 \rangle/y]\ \langle\rangle} \\
			&= \lambda z. \letin{x_2}{\pi_2\ x\ \langle\rangle}{\letin{y_2}{V[\langle \pi_2\ \pi_1\ x, x_2 \rangle/x]\ \langle\rangle}{W[\langle \pi_1\ \pi_1\ x, y_2 \rangle/y]\ \langle\rangle}}
		\end{align}
		The right-hand side:
		\begin{align}
			&W^{\sharp} \comp (A_1 \times V) \comp \associator \\
			&= (\lambda z. \letin{y_2}{\pi_2\ y'\ \langle\rangle}{W[\langle \pi_1\ y', y_2 \rangle/y]\ \langle\rangle})[\langle \pi_1\ x, V[\pi_2\ x/x] \rangle/y'] \\
			&\qquad [\langle \pi_1\ \pi_1\ x, \langle \pi_2\ \pi_1\ x, \pi_2\ x \rangle \rangle/x] \\
			&= (\lambda z. \letin{y_2}{V[\pi_2\ x/x]\ \langle\rangle}{W[\langle \pi_1\ x, y_2 \rangle/y]\ \langle\rangle}) \\
			&\qquad [\langle \pi_1\ \pi_1\ x, \langle \pi_2\ \pi_1\ x, \pi_2\ x \rangle \rangle/x] \\
			&= (\lambda z. \letin{y_2}{V[\langle \pi_2\ \pi_1\ x, \pi_2\ x \rangle/x]\ \langle\rangle}{W[\langle \pi_1\ \pi_1\ x, y_2 \rangle/y]\ \langle\rangle}) \\
			&(W^{\sharp} \comp (A_1 \times V) \comp \associator)^{\sharp} \\
			&= \lambda z. \letin{x_2}{\pi_2\ x\ \langle\rangle}{(W^{\sharp} \comp (A_1 \times V) \comp \associator)[\langle \pi_1\ x, x_2 \rangle/x]\ \langle\rangle} \\
			&= \lambda z. \letin{x_2}{\pi_2\ x\ \langle\rangle}{\letin{y_2}{V[\langle \pi_2\ \pi_1\ x, x_2 \rangle/x]\ \langle\rangle}{W[\langle \pi_1\ \pi_1\ x, y_2 \rangle/y]\ \langle\rangle}}
			\qedhere
		\end{align}
	\end{itemize}
\end{proof}

\subsection{Kleisli Exponentials}

\begin{definition}[Kleisli category]
	For each $\Sigma \in \SemanticEffectType{\TermModel{\EffectHandlerCalculus}}$, we define a category $\mathbf{Kl}_\Sigma$ as follows.
	\begin{itemize}
		\item object: types $A$
		\item morphism: computations $x : A \vdash M : B ! \Sigma$ modulo equations
		\item identity: $x : A \vdash \return{x} : A ! \Sigma$
		\item composition: $M \comp N = \letin{x}{N}{M}$
	\end{itemize}
\end{definition}
\begin{proof}
	We show that this is indeed a category.
	\begin{itemize}
		\item Well-defined-ness of the composition follows from the congruence rule for let.
		\item The left/right-identity law and associativity follows from the monad laws.
		\qedhere
	\end{itemize}
\end{proof}

\begin{lemma}\label{lem:term-model-kleisli-category}
	The Kleisli category $\TermModel{\EffectHandlerCalculus}_{T_\Sigma}$ is isomorphic to $\mathbf{Kl}_\Sigma$.
\end{lemma}
\begin{proof}
	The isomorphism $\TermModel{\EffectHandlerCalculus}_{T_\Sigma}(A, B) \cong \mathbf{Kl}_\Sigma(A, B)$ is given as follows.
	\begin{mathpar}
		\inferrule{
			x : A \vdash M : B ! \Sigma
		}{
			x : A \vdash \lambda y. M : \UnitType \to B ! \Sigma
		}
		\and
		\inferrule{
			x : A \vdash V : \UnitType \to B ! \Sigma
		}{
			x : A \vdash V\ \langle\rangle : B ! \Sigma
		}
	\end{mathpar}
	It is straightforward to check that they are mutually inverse.
	It remains to show that the above map $\TermModel{\EffectHandlerCalculus}_{T_\Sigma}(A, B) \to \mathbf{Kl}_\Sigma(A, B)$ is functorial.
	\begin{itemize}
		\item Identities $x : A \vdash \lambda y. \return{x} : \UnitType \to A ! \Sigma$ in $\TermModel{\EffectHandlerCalculus}_{T_\Sigma}$ are mapped to $x : A \vdash (\lambda y. \return{x})\ \langle\rangle = \return{x} : \UnitType \to A ! \Sigma$.
		\item Composition: Let $V$ and $W$ be morphisms in $\TermModel{\EffectHandlerCalculus}_{T_\Sigma}$.
		Their composite in $\TermModel{\EffectHandlerCalculus}_{T_\Sigma}$ is given as follows.
		\begin{align}
			\mu \comp T_\Sigma V \comp W &= \lambda y. \letin{z}{(\lambda y. \letin{x}{W\ \langle\rangle}{\return{V}})\ \langle\rangle}{z\ \langle\rangle} \\
			&= \lambda y. \letin{z}{(\letin{x}{W\ \langle\rangle}{\return{V}})}{z\ \langle\rangle} \\
			&= \lambda y. \letin{x}{W\ \langle\rangle}{\letin{z}{\return{V}}{z\ \langle\rangle}} \\
			&= \lambda y. \letin{x}{W\ \langle\rangle}{V\ \langle\rangle}
			\qedhere
		\end{align}
	\end{itemize}
\end{proof}

By Lemma~\ref{lem:term-model-kleisli-category}, we explain Kleisli exponentials using $\mathbf{Kl}_\Sigma$ in place of $\TermModel{\EffectHandlerCalculus}_{T_\Sigma}$.
Note that the adjunction $J \dashv K : \mathbf{Kl}_\Sigma \to \TermModel{\EffectHandlerCalculus}$ is given as follows.
\begin{itemize}
	\item $J : \TermModel{\EffectHandlerCalculus} \to \mathbf{Kl}_\Sigma$ is given by $J A = A$ for each object $A$ and $J V = \return{V}$ for each morphism $V$.
	\item $K : \mathbf{Kl}_\Sigma \to \TermModel{\EffectHandlerCalculus}$ is given by $K A = T_\Sigma A = \UnitType \to A ! \Sigma$ for each object $A$ and $K M = \lambda y. \letin{x}{x\ \langle\rangle}{M}$ for each morphism $M$.
\end{itemize}

\begin{lemma}[Kleisli exponential]\label{lem:term-model-kleisli-exponential}
	Kleisli exponentials for the strong monad in Definition~\ref{def:term-model-monad} are given as follows.
	\begin{itemize}
		\item $\KleisliExp{T_\Sigma}{A}{B} = A \to B ! \Sigma$
		\item The counit $\epsilon\in \mathbf{Kl}_\Sigma(J(\KleisliExp{T_\Sigma}{A}{B}) \times A, B)$ is defined as follows.
		\begin{equation}
			x : (A \to B ! \Sigma) \times A \vdash (\pi_1\ x)\ (\pi_2\ x) : B ! \Sigma
			\label{eq:term-model-eval}
		\end{equation}
		\item Universal property: for $x : A' \times A \vdash M : B ! \Sigma$, there exists a unique morphism $\overline{M} : A' \to \KleisliExp{T_\Sigma}{A}{B}$ in $\TermModel{\EffectHandlerCalculus}$ such that the following diagram commutes.
		\begin{center}
			\begin{tikzcd}
				J(A' \times A) \ar[rd, "M"] \ar[d, dashed, swap, "J(\overline{M} \times A)"] \\
				J((\KleisliExp{T_\Sigma}{A}{B}) \times A) \ar[r, "\epsilon"] & B
			\end{tikzcd}
			in $\mathbf{Kl}_\Sigma$
		\end{center}
		The unique morphism $\overline{M}$ is given by
		\begin{equation}
			x : A' \vdash \lambda y. M[\langle x, y \rangle/x] : A \to B ! \Sigma.
			\label{eq:term-model-currying}
		\end{equation}
	\end{itemize}
\end{lemma}
\begin{proof}
	We show the universal property.
	Note that we have the following equation.
	\begin{align}
		\epsilon \comp J(\overline{M} \times A) &= \letin{x}{\return{\langle \overline{M}[\pi_1\ x/x], \pi_2\ x \rangle}}{(\pi_1\ x)\ (\pi_2\ x)} \\
		&= \overline{M}[\pi_1\ x/x]\ (\pi_2\ x)
	\end{align}
	\begin{itemize}
		\item Existence: We show that $\overline{M}$ defined in \eqref{eq:term-model-currying} makes the diagram commute.
		\begin{align}
			\epsilon \comp J(\overline{M} \times A) &= \overline{M}[\pi_1\ x/x]\ (\pi_2\ x) \\
			&= (\lambda y. M[\langle \pi_1\ x, y \rangle/x])\ (\pi_2\ x) \\
			&= M[\langle \pi_1\ x, \pi_2\ x \rangle/x] \\
			&= M
		\end{align}
		\item Uniqueness: If $\epsilon \comp J(\overline{M} \times A) = M$, then we have $\overline{M}[\pi_1\ x/x]\ (\pi_2\ x) = M$, which implies
		\[ \lambda y. M[\langle x, y \rangle/x] = \lambda y. \overline{M}\ y = \overline{M}. \qedhere \]
	\end{itemize}
\end{proof}

Note that~\eqref{eq:term-model-eval} and~\eqref{eq:term-model-currying} correspond to the evaluation $\eval$ and currying $\Lambda$ modulo the isomorphism in Lemma~\ref{lem:term-model-kleisli-category}.

Now, it is straightforward to see that the interpretation of types and signatures in the term model $\TermModel{\EffectHandlerCalculus}$ are themselves.
\begin{lemma}\label{lem:term-model-type-interpretation}
	\begin{itemize}
		\item For any value type $A$, its interpretation $\interpret{A}$ in the term model $\TermModel{\EffectHandlerCalculus}$ is $A$ itself.
		\item For any signature $\Sigma$, its interpretation $\interpret{\Sigma}$ in the term model $\TermModel{\EffectHandlerCalculus}$ is $\Sigma$ itself.
	\end{itemize}
\end{lemma}
\begin{proof}
	By simultaneous induction on types and signatures.
	The proof is straightforward.
\end{proof}

\subsection{Handlers}

\begin{lemma}\label{lem:term-model-effect-handling-monad}
	Definition~\ref{def:term-model-monad} satisfies \eqref{eq:handle-unit}, \eqref{eq:handle-multiplication}, and \eqref{eq:handle-operation}.
\end{lemma}
\begin{proof}
	Firstly, note that the effect handler $H$ contains variable $x$ only in the form of $\pi_1\ x$, and hence we have $H[\langle \pi_1\ x, V \rangle/x] = H$ for any value term $V$.
	\allowdisplaybreaks
	\begin{itemize}
		\item We prove \eqref{eq:handle-unit}.
		\begin{align}
			&\mathbf{handle} \comp (\identity{} \times \eta) \\
			&= (\lambda z. \handlewithto{\pi_2\ x\ \langle\rangle}{H}{r}{r\ \langle\rangle})[\langle \pi_1\ x, \lambda y. \return{\pi_2\ x}\rangle/x] \\
			&= \lambda z. \handlewithto{\return{\pi_2\ x}}{H}{r}{r\ \langle\rangle} \\
			&= \lambda z. \pi_2\ x\ \langle\rangle \\
			&= \pi_2\ x
		\end{align}
		\item We prove \eqref{eq:handle-multiplication}.
		The left-hand side:
		\begin{align}
			&\mathbf{handle} \comp (\identity{} \times \mu) \\
			&= (\lambda z. \handlewithto{\pi_2\ x\ \langle\rangle}{H}{r}{r\ \langle\rangle})[\langle \pi_1\ x, \lambda z. \letin{y}{\pi_2\ x\ \langle\rangle}{y\ \langle\rangle}\rangle/x] \\
			&= \lambda z. \handlewithto{(\letin{y}{\pi_2\ x\ \langle\rangle}{y\ \langle\rangle})}{H}{r}{r\ \langle\rangle} \\
			&= \lambda z. \handlewithto{\pi_2\ x\ \langle\rangle}{H}{y}{\handlewithto{y\ \langle\rangle}{H}{r}{r\ \langle\rangle}}
		\end{align}
		The right-hand side:
		\begin{align}
			&\mathbf{handle} \comp (\identity{} \times T_\Sigma \mathbf{handle}) \\
			&= (\lambda z. \handlewithto{\pi_2\ x\ \langle\rangle}{H}{r}{r\ \langle\rangle})\\
			&\qquad [\langle \pi_1\ x, \lambda z. \letin{y}{\pi_2\ x\ \langle\rangle}{\return{(\lambda z. \handlewithto{\pi_2\ y\ \langle\rangle}{H}{r}{r\ \langle\rangle})}}\rangle/x] \\
			&= \lambda z. \handlewithto{\letin{y}{\pi_2\ x\ \langle\rangle}{\return{(\lambda z. \handlewithto{\pi_2\ y\ \langle\rangle}{H}{r}{r\ \langle\rangle})}}}{H}{r}{r\ \langle\rangle} \\
			&= \lambda z. \handlewithto{\pi_2\ x\ \langle\rangle}{H}{y}{\\ &\qquad \handlewithto{\return{(\lambda z. \handlewithto{\pi_2\ y\ \langle\rangle}{H}{r}{r\ \langle\rangle})}}{H}{r}{r\ \langle\rangle}} \\
			&= \lambda z. \handlewithto{\pi_2\ x\ \langle\rangle}{H}{y}{\handlewithto{\pi_2\ y\ \langle\rangle}{H}{r}{r\ \langle\rangle}} \\
			&\mathbf{handle} \comp (\identity{} \times T_\Sigma \mathbf{handle}) \comp \tupling{\pi_1}{\strength} \\
			&= (\lambda z. \handlewithto{\pi_2\ x\ \langle\rangle}{H}{y}{\handlewithto{\pi_2\ y\ \langle\rangle}{H}{r}{r\ \langle\rangle}})\\
			&\qquad [\langle \pi_1\ x, \lambda z. \letin{x_2}{\pi_2\ x\ \langle\rangle}{\return{(\pi_1\ x, x_2)}} \rangle/x] \\
			&= \lambda z. \handlewithto{\letin{x_2}{\pi_2\ x\ ()}{\return{(\pi_1\ x, x_2)}}}{H}{y}{\handlewithto{\pi_2\ y\ \langle\rangle}{H}{r}{r\ \langle\rangle}} \\
			&= \lambda z. \handlewithto{\pi_2\ x\ \langle\rangle}{H}{x_2}{\\&\qquad \handlewithto{\return{(\pi_1\ x, x_2)}}{H}{x}{\handlewithto{\pi_2\ x\ \langle\rangle}{H}{r}{r\ \langle\rangle}}} \\
			&= \lambda z. \handlewithto{\pi_2\ x\ \langle\rangle}{H}{x_2}{\handlewithto{x_2\ \langle\rangle}{H}{r}{r\ \langle\rangle}} \\
			&= \lambda z. \handlewithto{\pi_2\ x\ \langle\rangle}{H}{y}{\handlewithto{y\ \langle\rangle}{H}{r}{r\ \langle\rangle}}
		\end{align}
		\item We prove \eqref{eq:handle-operation}.
		The left-hand side:
		\begin{align}
			&\mathbf{handle} \comp (\identity{} \times T_\Sigma \eval) \\
			&= (\lambda z. \handlewithto{\pi_2\ x\ \langle\rangle}{H}{r}{r\ \langle\rangle})\\
			&\qquad [\langle \pi_1\ x, \lambda z. \letin{y}{\pi_2\ x\ \langle\rangle}{\return{(\lambda z. (\pi_1\ y)\ (\pi_2\ y))}} \rangle/x] \\
			&= \lambda z. \handlewithto{\letin{y}{\pi_2\ x\ \langle\rangle}{\return{(\lambda z. (\pi_1\ y)\ (\pi_2\ y))}}}{H}{r}{r\ \langle\rangle} \\
			&= \lambda z. \handlewithto{\pi_2\ x\ \langle\rangle}{H}{y}{(\pi_1\ y)\ (\pi_2\ y)} \\
			&\mathbf{handle} \comp (\identity{} \times T_\Sigma \eval) \comp (\identity{} \times \strength) \\
			&= (\lambda z. \handlewithto{\pi_2\ x\ \langle\rangle}{H}{y}{(\pi_1\ y)\ (\pi_2\ y)}) \\
			&\qquad [\langle \pi_1\ x, \lambda z. \letin{x_2}{\pi_2\ (\pi_2\ x)\ \langle\rangle}{\return{(\pi_1\ (\pi_2\ x), x_2)}} \rangle/x] \\
			&= \lambda z. \handlewithto{\letin{x_2}{\pi_2\ (\pi_2\ x)\ \langle\rangle}{\return{(\pi_1\ (\pi_2\ x), x_2)}}}{H}{x}{(\pi_1\ x)\ (\pi_2\ x)} \\
			&= \lambda z. \handlewithto{\pi_2\ (\pi_2\ x)\ \langle\rangle}{H}{x_2}{(\pi_1\ (\pi_2\ x))\ x_2} \\
			&\mathbf{handle} \comp (\identity{} \times T_\Sigma \eval) \comp (\identity{} \times \strength) \comp (\identity{} \times (\identity{} \times \interpret{\mathtt{op}})) \\
			&= (\lambda z. \handlewithto{\pi_2\ (\pi_2\ x)\ \langle\rangle}{H}{x_2}{(\pi_1\ (\pi_2\ x))\ x_2}) \\
			&\qquad [\langle \pi_1\ x, \langle \pi_1\ (\pi_2\ x), \lambda z. \mathtt{op}(\pi_2\ (\pi_2\ x)) \rangle \rangle/x] \\
			&= \lambda z. \handlewithto{\mathtt{op}(\pi_2\ (\pi_2\ x))}{H}{x_2}{(\pi_1\ (\pi_2\ x))\ x_2} \\
			&= \lambda z. \pi_{\mathtt{op}}\ (\pi_1\ x)\ \langle \pi_2\ (\pi_2\ x), \lambda x_2. (\pi_1\ (\pi_2\ x))\ x_2 \rangle \\
			&= \lambda z. \pi_{\mathtt{op}}\ (\pi_1\ x)\ \langle \pi_2\ (\pi_2\ x), \pi_1\ (\pi_2\ x) \rangle \\
			&\mathbf{handle} \comp (\identity{} \times T_\Sigma \eval) \comp (\identity{} \times \strength) \comp (\identity{} \times \braiding) \comp (\identity{} \times (\interpret{\mathtt{op}} \times \identity{})) \\
			&= \mathbf{handle} \comp (\identity{} \times T_\Sigma \eval) \comp (\identity{} \times \strength) \comp (\identity{} \times (\identity{} \times \interpret{\mathtt{op}})) \comp (\identity{} \times \braiding) \\
			&= \mathbf{handle} \comp (\identity{} \times T_\Sigma \eval) \comp (\identity{} \times \strength) \comp (\identity{} \times \braiding) \comp (\identity{} \times (\interpret{\mathtt{op}} \times \identity{})) \\
			&\qquad [\langle \pi_1\ x, \langle \pi_2\ (\pi_2\ x), \pi_1\ (\pi_2\ x) \rangle \rangle/x] \\
			&= \lambda z. \pi_{\mathtt{op}}\ (\pi_1\ x)\ \langle \pi_1\ (\pi_2\ x), \pi_2\ (\pi_2\ x) \rangle \\
			&= \lambda z. \pi_{\mathtt{op}}\ (\pi_1\ x)\ (\pi_2\ x)
		\end{align}
		The right-hand side:
		\begin{align}
			&\eval \comp (\pi_{\mathtt{op}} \times \identity{}) \\
			&= (\lambda z. (\pi_1\ x)\ (\pi_2\ x))[\langle \pi_{\mathtt{op}}\ (\pi_1\ x), \pi_2\ x \rangle/x] \\
			&= \lambda z. \pi_{\mathtt{op}}\ (\pi_1\ x)\ (\pi_2\ x)
			\qedhere
		\end{align}
	\end{itemize}
\end{proof}

\subsection{Interpretation in the Term Model}

We show that the interpretation of a term in $\TermModel{\EffectHandlerCalculus}$ is essentially equal to the term itself.
Here, ``essentially'' means that the equality holds up to two isomorphisms.
\begin{itemize}
	\item The interpretation of a computation term $\Gamma \vdash M : A ! \Sigma$ is given as a value term of the form $\Gamma \vdash \interpret{M} : \UnitType \to A ! \Sigma$.
	Thus, we need to identify value terms of type $\UnitType \to A ! \Sigma$ with computation terms of type $A ! \Sigma$ using the isomorphism in Lemma~\ref{lem:term-model-kleisli-category}.
	\item A morphism in $\TermModel{\EffectHandlerCalculus}$ is defined as a value term $x : A \vdash V : B$ whose context contains only one variable $x$.
	When we interpret a term whose context contains multiple variables, we need to modify the context so that it contains only one variable.
	This is done by using the following substitution $s_{\Gamma}$.
\end{itemize}

\begin{definition}
	Given a context $\Gamma$, we define $\Gamma^{\times}$ as $x^{\times} : A_{\Gamma}$ where $A_{\Gamma}$ is the product type of the types in $\Gamma$.
	\[ A_{\emptyctx} \coloneqq \UnitType \qquad A_{\Gamma, y : B} \coloneqq A_{\Gamma} \times B \]
	We also define a substitution $s_{\Gamma}$ inductively as follows.
	\[ s_{\Gamma, y : B}(z) \quad\coloneqq\quad \begin{cases}
		s_{\Gamma}(z)[\pi_1\ x^{\times}/x^{\times}] & \text{$z$ is a variable in $\Gamma$} \\
		\pi_2\ x^{\times} & \text{$z$ is $y$}
	\end{cases} \]
	We have $\Gamma^{\times} \vdash s_{\Gamma}(x) : A$ for each $(x : A) \in \Gamma$.
\end{definition}

We can define the ``inverse'' of $s_{\Gamma}$ as follows.
Given a context $\Gamma$, we define a term $\Gamma \vdash V_{\Gamma} : A_{\Gamma}$ as follows.
\[ V_{\emptyctx} \quad\coloneqq\quad \langle\rangle \qquad V_{\Gamma, y : B} \quad\coloneqq\quad \langle V_{\Gamma}, y \rangle \]
Then, $s_{\Gamma}[V_{\Gamma} / x^{\times}]$ is the identity substitution on $\Gamma$.

\begin{proposition}\label{prop:term-model-interpretation}
	The interpretation of terms in the term model $\TermModel{\EffectHandlerCalculus}$ is given as follows.
	\begin{itemize}
		\item For each value term $\Gamma \vdash V : B$, its interpretation $\interpret{V} : \interpret{\Gamma} \to \interpret{B}$ in the term model $\TermModel{\EffectHandlerCalculus}$ is $\Gamma^{\times} \vdash V[s_{\Gamma}] : B$.
		\item For each computation term $\Gamma \vdash M : B ! \Sigma$, its interpretation $\interpret{M} : \interpret{\Gamma} \to T_\Sigma \interpret{B}$ in the term model $\TermModel{\EffectHandlerCalculus}$ is $\Gamma^{\times} \vdash M[s_{\Gamma}] : B ! \Sigma$ up to the isomorphism $\TermModel{\EffectHandlerCalculus}_{T_\Sigma}(\interpret{\Gamma}, \interpret{B}) \cong \mathbf{Kl}_\Sigma(\interpret{\Gamma}, \interpret{B})$ in Lemma~\ref{lem:term-model-kleisli-category}.
		\item For each handler $\Gamma \vdash H : \Sigma \Rightarrow C$, its interpretation $\interpret{H} : \interpret{\Gamma} \to \HBundle{\interpret{\Sigma}}{\interpret{C}}$ in the term model $\TermModel{\EffectHandlerCalculus}$ is given as follows.
		\[ \Gamma^{\times} \vdash \langle \lambda \langle x, k \rangle. M_{\mathtt{op}}[s_{\Gamma}] \mid (\mathtt{op}(x, k) \mapsto M_{\mathtt{op}}) \in H \rangle : \prod_{(\mathtt{op} : A_{\mathtt{op}} \rightarrowtriangle B_{\mathtt{op}}) \in \Sigma} A_{\mathtt{op}} \times (B_{\mathtt{op}} \to C) \to C \]
	\end{itemize}
\end{proposition}
\begin{proof}
	Recall that by Lemma~\ref{lem:term-model-type-interpretation}, the interpretation of types in the term model $\TermModel{\EffectHandlerCalculus}$ are themselves.
	It is straightforward to see that the interpretations of context $\Gamma$ is $\interpret{\Gamma} = A_1 \times \dots \times A_n$.
	Thus, it suffices to show that the interpretations of terms are as stated.
	We prove the statement by induction.
	Most cases are straightforward.
	We only show some interesting cases.
	\begin{itemize}
		\item It is straightforward to show that the interpretation of variables $\Gamma \vdash x : A$ is given by $s_{\Gamma}(x)$.
		\item We show how variable bindings are treated in the proof.
		As an example, consider lambda abstraction $\Gamma \vdash \lambda y : A. M : C$.
		By the induction hypothesis, the interpretation of $M$ is given by $M[s_{\Gamma, y : A}]$.
		Thus, the interpretation of $\lambda y : A. M$ is given by
		\begin{align}
			\interpret{\lambda y : A. M} &= \Lambda(\interpret{M}) \\
			&= \Gamma^{\times} \vdash \lambda y. (M[s_{\Gamma, y : A}][\langle x^{\times}, y \rangle/x^{\times}]) : A \to C \\
			&= \Gamma^{\times} \vdash \lambda y. M[s_{\Gamma}] : A \to C \\
			&= \Gamma^{\times} \vdash (\lambda y. M)[s_{\Gamma}] : A \to C
		\end{align}
		\item $\Gamma \vdash \handlewithto{M}{H}{x}{N} : C$: We have
		\begin{align}
			&\interpret{\handlewithto{M}{H}{x}{N}} \\
			&= \mathbf{handle}_{\interpret{\Sigma}, \interpret{\Sigma'}, J \interpret{B}} \comp \tupling{\interpret{H}}{T_{\interpret{\Sigma}} \interpret{N} \comp \strength^{T_{\interpret{\Sigma}}} \comp \tupling{\identity{}}{\interpret{M}}} \\
			&= \lambda z. \handlewithto{V\ \langle\rangle}{H'}{r}{r\ \langle\rangle}
		\end{align}
		where
		\begin{align}
			V &= T_{\interpret{\Sigma}} \interpret{N} \comp \strength^{T_{\interpret{\Sigma}}} \comp \tupling{\identity{}}{\interpret{M}} \\
			H' &= \{ \mathtt{op}(y, k) \mapsto \pi_{\mathtt{op}}\ \interpret{H}\ \langle y, k \rangle \mid (\mathtt{op} : A_{\mathtt{op}} \rightarrowtriangle B_{\mathtt{op}}) \in \Sigma \}.
		\end{align}
		It is straightforward to see that $H' = H[s_{\Gamma}]$ by IH.
		On the other hand, $V$ is calculated as follows.
		\begin{align}
			&T_{\interpret{\Sigma}} \interpret{N} \comp \strength^{T_{\interpret{\Sigma}}} \comp \tupling{\identity{}}{\interpret{M}} \\
			&= (\eta \comp \interpret{N})^{\sharp} \comp \tupling{\identity{}}{\interpret{M}} \\
			&= \lambda z. \letin{x}{M[s_{\Gamma}]}{\return{\lambda z. N[s_{\Gamma}]}}
		\end{align}
		Therefore, we have the following.
		\begin{align}
			&\interpret{\handlewithto{M}{H}{x}{N}} \\
			&= \lambda z. \handlewithto{(\letin{x}{M[s_{\Gamma}]}{\return{\lambda z. N[s_{\Gamma}]}})}{H[s_{\Gamma}]}{r}{r\ \langle\rangle} \\
			&= \lambda z. \handlewithto{M[s_{\Gamma}]}{H[s_{\Gamma}]}{x}{\handlewithto{\return{\lambda z. N[s_{\Gamma}]}}{H[s_{\Gamma}]}{r}{r\ \langle\rangle}} \\
			&= \lambda z. \handlewithto{M[s_{\Gamma}]}{H[s_{\Gamma}]}{x}{N[s_{\Gamma}]}
		\end{align}
		\item $\Gamma \vdash H : \Sigma \Rightarrow C$:
		For each $\Gamma, x : A_{\mathtt{op}}, k : B_{\mathtt{op}} \to C \vdash M_{\mathtt{op}} : C$, we have the following.
		\begin{align}
			\interpret{M_{\mathtt{op}}} &= \lambda z. M_{\mathtt{op}}[s_{\Gamma, x : A_{\mathtt{op}}, k : B_{\mathtt{op}} \to C}] \\
			&= \lambda z. M_{\mathtt{op}}[s_{\Gamma}[\pi_1\ (\pi_1\ x^{\times})], \pi_2\ (\pi_1\ x^{\times}) / x, \pi_2\ x^{\times} / k] \\
			\interpret{M_{\mathtt{op}}} \comp \associator^{-1} &= \lambda z. M_{\mathtt{op}}[s_{\Gamma}[\pi_1\ x^{\times}], \pi_1\ (\pi_2\ x^{\times}) / x, \pi_2\ (\pi_2\ x^{\times}) / k] \\
			&= \lambda z. M_{\mathtt{op}}[s_{\Gamma}[\pi_1\ x^{\times}], \pi_2\ x^{\times} / (x, k)] \\
			\Lambda (\interpret{M_{\mathtt{op}}} \comp \associator^{-1}) &= \lambda \langle x, k \rangle. M_{\mathtt{op}}[s_{\Gamma}]
			\qedhere
		\end{align}
	\end{itemize}
\end{proof}

\begin{proof}[Proof of Theorem~\ref{thm:completeness}]
	Suppose that we have $\interpret{V} = \interpret{W}$ for any $\EffectHandlerCalculus$-model where $V$ and $W$ are value terms.
	Then, by Proposition~\ref{prop:term-model-interpretation}, we have $V[s_{\Gamma}] = W[s_{\Gamma}]$ as the interpretation in the term model $\TermModel{\EffectHandlerCalculus}$.
	Since $s_{\Gamma}[V_{\Gamma} / x^{\times}]$ is the identity substitution, we have $V = V[s_{\Gamma}][V_{\Gamma} / x^{\times}] = W[s_{\Gamma}][V_{\Gamma} / x^{\times}] = W$.
	The case for computation terms and handlers can be shown similarly.
\end{proof}

\end{toappendix}

\section{Examples of Models}
\label{sec:examples}

In this section, we present examples of $\EffectHandlerCalculus$-models.
The first one is the traditional model based on free monads, and the second one is a novel model based on continuation-passing style (CPS) semantics for effect handlers \cite{HillerstromFSCD2017}.

\subsection{Free Monads}

Below, we show that free monads on $\Set$ give rise to a $\EffectHandlerCalculus$-model.

\begin{definition}
	A \emph{free monad} on a category $\category{C}$ is a free object with respect to a forgetful functor from the category of monads on $\category{C}$ to the category of endofunctors on $\category{C}$.
	Concretely, given an endofunctor $F : \category{C} \to \category{C}$, a free monad for $F$ is a monad $T$ together with a natural transformation $\zeta : F \to T$ such that for any monad $T'$ and natural transformation $\phi : F \to T'$, there exists a unique monad morphism $\hat{\phi} : T \to T'$ satisfying $\phi = \hat{\phi} \comp \zeta$.
\end{definition}

Every elementary topos with a natural number object admits free monads for polynomial endofunctors \cite{MoerdijkAnnalsofPureandAppliedLogic2000}.
Specifically, $\Set$ admits free monads for all polynomial endofunctors.

\begin{definition}
	An \emph{algebraically-free monad} for an endofunctor $F : \category{C} \to \category{C}$ is a monad $T$ such that there exists an isomorphism between the category of $F$-algebras $\mathbf{Alg}(F)$ and the category of Eilenberg--Moore $T$-algebras $\mathbf{EMAlg}(T)$ that commutes with the forgetful functors to $\category{C}$.
\end{definition}

Any algebraically-free monad is a free monad.
Conversely, if $T$ is a free monad on a locally small and complete category $\category{C}$, then $T$ is an algebraically-free monad \cite[Section~22]{KellyBullAustralMathSoc1980}.
Specifically, free monads on $\Set$ are algebraically-free.
Below, we restrict our attention to the category $\Set$.

\begin{proposition}\label{prop:free-monad-effect-handling}
	For each semantic signature $S \in \SemanticEffectType{\Set}$, we define a polynomial functor $F_S : \Set \to \Set$ by
	\[ F_S X \quad\coloneqq\quad \sum_{(\mathtt{op} : A_{\mathtt{op}} \rightarrowtriangle B_{\mathtt{op}}) \in S} A_{\mathtt{op}} \times (\exponential{B_{\mathtt{op}}}{X}). \]
	Let $T_S$ be the free monad for $F_S$.
	Then, $\{ T_S \}_{S \in \SemanticEffectType{\Set}}$ together with the following structures is a $\EffectHandlerCalculus$-model on $\Set$.
	\setlength{\leftmargini}{20pt}
	\begin{description}
		\item[Operation] Since $T_S$ is free, we have a natural transformation $\zeta : F_S \to T_S$. Thus, we define $\interpret{op}$ as the following composite where $\iota_{\mathtt{op}}$ is the coproduct injection.
		\[ A_{\mathtt{op}} \to A_{\mathtt{op}} \times (\exponential{B_{\mathtt{op}}}{B_{\mathtt{op}}}) \xrightarrow{\iota_{\mathtt{op}}} F_S B_{\mathtt{op}} \xrightarrow{\zeta} T_S B_{\mathtt{op}} \]
		\item[Handle] For each $X \in \Set_{T_{S'}}$, $h \in \HBundle{S}{X}$ is (equivalent to) an $F_S$-algebra structure on $K_{T_{S'}} X$.
		Thus, we define a function $\mathbf{handle}_{S, S', X} : \HBundle{S}{X} \times T_S K_{T_{S'}} X \to K_{T_{S'}} X$ so that for each $h$, $\mathbf{handle}(h, {-})$ is given as the Eilenberg--Moore $T_S$-algebra such that $\mathbf{handle}(h, {-}) \comp \zeta = h'$ where $h'$ is the $F_S$-algebra structure corresponding to $h$.
		\qed
	\end{description}
\end{proposition}
\begin{appendixproof}[Proof of Proposition~\ref{prop:free-monad-effect-handling}]
	The equation~\eqref{eq:handle-unit} and~\eqref{eq:handle-multiplication} follow from the fact that $\mathbf{handle}_{S, S', X}(h, {-})$ is an Eilenberg--Moore algebra.
	Moreover, we have $\MixedAlgOp{S'}{X}{\interpret{\mathtt{op}}^T_S} = \zeta \comp \iota_{\mathtt{op}}$, which implies~\eqref{eq:handle-operation}.
	Note that in $\Set$, any functor has a canonical strength, and any natural transformation between functors is automatically strong.
	\begin{align}
		&\MixedAlgOp{S'}{X}{\interpret{\mathtt{op}}^T_S} \\
		&= T_S \eval \comp \strength^{T_S} \comp \braiding \comp ((\zeta \comp \iota_{\mathtt{op}} \comp \tupling{\identity{}}{\Lambda \pi_2}) \times \identity{}) \\
		&= T_S \eval \comp \zeta \comp \strength^{F_S} \comp \braiding \comp ((\iota_{\mathtt{op}} \comp \tupling{\identity{}}{\Lambda \pi_2}) \times \identity{}) \\
		&= \zeta \comp F_S \eval \comp \strength^{F_S} \comp \braiding \comp ((\iota_{\mathtt{op}} \comp \tupling{\identity{}}{\Lambda \pi_2}) \times \identity{}) \\
		&= \zeta \comp \iota_{\mathtt{op}} \comp (A_{\mathtt{op}} \times (\exponential{B_{\mathtt{op}}}{\eval})) \comp \strength^{(A_{\mathtt{op}} \times (\exponential{B_{\mathtt{op}}}{{-}}))} \comp \braiding \comp ((\tupling{\identity{}}{\Lambda \pi_2}) \times \identity{}) \\
		&= \zeta \comp \iota_{\mathtt{op}}
	\end{align}
\end{appendixproof}

\subsection{CPS Semantics for Effect Handlers}

It is known that the continuation-passing style (CPS) transformation for computational lambda calculus ($\lambda_c$-calculus) corresponds to the interpretation by a continuation monad \cite{FuhrmannInformationandComputation2004}.
We show that a similar result holds for the CPS transformation for $\EffectHandlerCalculus$.

\begin{definition}
	The \emph{typed CPS transformation} $\CPS{-}$ for $\EffectHandlerCalculus$~\cite{HillerstromFSCD2017,KammarICFP2013}\footnote{We slightly modify the definition of the CPS transformation in \cite{HillerstromFSCD2017,KammarICFP2013}, as their definition does not satisfy the equations for handlers.} is a syntactic transformation of the following form.
	\begin{align}
		&\Gamma \vdash V : A &&\mapsto  &&\CPS{\Gamma} \vdash \CPS{V} : \CPS{A} \\
		&\Gamma \vdash M : A ! \Sigma &&\mapsto &&\CPS{\Gamma} \vdash \CPS{M} : \CPS{A ! \Sigma} \\
		&\Gamma \vdash H : \Sigma \Rightarrow A ! \Sigma' &&\mapsto &&\CPS{\Gamma} \vdash \CPS{H} : \mathcal{H}_\Sigma[\CPS{A ! \Sigma'}]
	\end{align}
	Here, the source language is $\EffectHandlerCalculus$, and the target language is System F, i.e., a polymorphic lambda calculus with universal quantification.
	The type $\mathcal{H}_\Sigma[\beta]$ is defined as follows.
	\[ \mathcal{H}_\Sigma[\beta] \quad\coloneqq\quad \prod_{(\mathtt{op} : A_{\mathtt{op}} \rightarrowtriangle B_{\mathtt{op}}) \in \Sigma} \Big(\CPS{A_{\mathtt{op}}} \times (\CPS{B_{\mathtt{op}}} \to \beta) \to \beta\Big) \]
	Concretely, the transformation $\CPS{-}$ is defined as follows.
	\begin{itemize}
		\item Value types and terms are transformed homomorphically.
		\item Computation types are transformed as follows.
		\begin{align}
			\CPS{A ! \Sigma} \quad&\coloneqq\quad \forall \beta. (X \to \beta) \to \mathcal{H}_\Sigma[\beta] \to \beta
		\end{align}
		\item Computation terms are transformed as follows.
		\begin{align}
			&\CPS{V\ W} \quad\coloneqq\quad \CPS{V}\ \CPS{W} \\
			&\CPS{\return{V}} \quad\coloneqq\quad \forall \beta. \lambda k. \lambda h. k\ \CPS{V} \\
			&\CPS{\letin{x}{M}{N}} \quad\coloneqq\quad \forall \beta. \lambda k. \lambda h. \CPS{M}\ \beta\ (\lambda x. \CPS{N}\ \beta\ k\ h)\ h \\
			&\CPS{\mathtt{op}(V)} \quad\coloneqq\quad \forall \beta. \lambda k. \lambda h. \pi_{\mathtt{op}}\ h\ \langle \CPS{V}, k \rangle \\
			&\CPS{\handlewithto{M}{H}{x}{N}} \quad\coloneqq\quad \CPS{M}\ \CPS{A ! \Sigma}\ (\lambda x. \CPS{N})\ \CPS{H} \\
			&\qquad\qquad \text{where $N : A ! \Sigma$}
		\end{align}
		\item Handlers are transformed as follows.
		\[ \CPS{\{ \mathtt{op}(x, k) \mapsto M_{\mathtt{op}} \mid \mathtt{op} \in \Sigma \}} \quad\coloneqq\quad \langle \lambda \langle x, k \rangle. \CPS{M_{\mathrm{op}}} \mid \mathtt{op} \in \Sigma \rangle \]
	\end{itemize}
\end{definition}

Below, we explain how the CPS transformation defined above can be understood at the semantic level.
We briefly recall categorical models of System F, which is given by $\lambda{\to}{\forall}$-fibrations.
More details on this notion can be found in the literature~\cite{Jacobs2001}.

\begin{definition}[$\lambda{\to}{\forall}$-fibration]
	A \emph{polymorphic fibration} is a fibration with a generic object $\Omega$, fibred finite products, and finite products in the base category.
	A \emph{$\lambda{\to}$-fibration} is a polymorphic fibration with fibred exponentials.
	A \emph{$\lambda{\to}{\forall}$-fibration} is a $\lambda{\to}$-fibration with simple $\Omega$-products $\prod_{\Omega} : \category{E}_{I \times \Omega} \to \category{E}_{I}$.
\end{definition}

To capture the CPS transformation for effect handlers, we use the $\lambda{\to}{\forall}$-fibration defined by the term model for System F.

\begin{definition}
	The term model $\TermModel{\lambda{\to}{\forall}}$ for System F is defined as follows.
	We consider the following types for System F.
	\[ A, B \quad\coloneqq\quad X \mid \UnitType \mid A \times B \mid A \to B \mid \forall X. A \]
	Let $\Delta = X_1, \dots, X_m$ be a type variable context and $\Gamma = x_1 : A_1, \dots, x_n : A_n$ be a term variable context.
	Well-formed types are denoted by $\Delta \vdash A$, and well-typed terms are denoted by $\Delta \mid \Gamma \vdash V : A$.
	Then, the following construction gives a $\lambda{\to}{\forall}$-fibration.
	\setlength{\leftmargini}{20pt}
	\begin{description}
		\item[Base category] Objects are type variable contexts $\Delta$.
		A morphism $\Delta \to \Delta'$ is a substitution, i.e., an assignment of $\Delta' \vdash A_i$ for each $X_i \in \Delta'$.
		This gives a cartesian category where products are given by concatenation of type variable contexts.
		\item[Total category] An object over $\Delta$ is a well-formed type $\Delta \vdash A$.
		A vertical morphism $\Delta \vdash A \to \Delta \vdash B$ is a term $\Delta \mid x : A \vdash V : B$ modulo $\beta$/$\eta$-equivalence.
		For convenience, we often identify $\Delta \mid x_1 : A_1, \dots, x_n : A_n \vdash V : B$ with $\Delta \mid x : A_1 \times \dots \times A_n \vdash V[\pi_1\ x/ x_1, \dots, \pi_n\ x / x_n] : B$.
	\end{description}
\end{definition}

By definition of $\lambda{\to}{\forall}$-fibrations, the fibre category $\TermModel{\lambda{\to}{\forall}}_1$ over the terminal object $1$ is a cartesian closed category.
Since the terminal object $1$ corresponds to the empty type variable context, the fibre category $\TermModel{\lambda{\to}{\forall}}_1$ can be seen as the term model restricted to types without free type variables.
We can construct a $\EffectHandlerCalculus$-model on $\TermModel{\lambda{\to}{\forall}}_1$ that gives a category-theoretic interpretation of the CPS transformation for effect handlers.

\begin{definition}\label{def:multi-barrelled-monad}
	We define a strong monad $\{ D_S \}_{S \in \SemanticEffectType{\TermModel{\lambda{\to}{\forall}}_1}}$ on the fibre category $\TermModel{\lambda{\to}{\forall}}_1$ as follows.
	\setlength{\leftmargini}{20pt}
	\begin{description}
		\item[Object map] For each $S$ and $X$, we define $D_S$ as follows.
		\[ X : \Type \quad\mapsto\quad \forall \beta. (X \to \beta) \to \mathcal{H}_S[\beta] \to \beta : \Type \]
		Here, $\mathcal{H}_S[\beta]$ is a type defined as follows.
		\[ \mathcal{H}_S[\beta] \quad\coloneqq\quad \prod_{(\mathtt{op} : A_{\mathtt{op}} \rightarrowtriangle B_{\mathtt{op}}) \in S} (A_{\mathtt{op}} \times (B_{\mathtt{op}} \to \beta) \to \beta) \]
		\item[Unit] For each $A$, the unit is defined as follows.
		\[ x : A \vdash\quad \forall \beta.\ \lambda k : A \to \beta.\ \lambda h : \mathcal{H}_S[\beta].\ k\ x \quad: D_S A \]
		\item[Strong Kleisli extension] For each $x_1 : A_1, x_2 : A_2 \vdash M : D_S B$, its strong Kleisli extension $x_1 : A_1, x_2 : D_S A_2 \vdash M^{\sharp} : D_S B$ is defined as follows.
		\[ M^{\sharp}\ \coloneqq\ \forall \beta.\ \lambda k : B \to \beta.\ \lambda h : \mathcal{H}_S[\beta].\ x_2\ \beta\ (\lambda x_2 : A_2.\ M\ \beta\ k\ h)\ h \]
	\end{description}
\end{definition}

\begin{proposition}\label{prop:polymorphic-multi-barrelled-monad}
	The strong monad in Definition~\ref{def:multi-barrelled-monad} together with the following structures is a $\EffectHandlerCalculus$-model.
	\setlength{\leftmargini}{20pt}
	\begin{description}
		\item[Operation] For each $S$ and $\mathtt{op} : A \rightarrowtriangle B$ in $S$, we define $\interpret{\mathtt{op}} : A \to D_S B$ as follows.
		\[ x : A_{\mathtt{op}} \vdash \forall \beta.\ \lambda k : B_{\mathtt{op}} \to \beta.\ \lambda h : \mathcal{H}_S[\beta].\ \pi_{\mathtt{op}}\ h\ \langle x, k \rangle : D_S B_{\mathtt{op}} \]
		\item[Handle] We define $\mathbf{handle}_{S, S', A}$ as follows.
		\begin{align}
			&h : \mathcal{H}_S[D_{S'} A],\quad x : D_S D_{S'} A \\
			&\quad\vdash\quad x\ (D_{S'} A)\ (\lambda y : D_{S'} A.\ y)\ h \quad: D_{S'} A \tag*{\qed}
		\end{align}
	\end{description}
\end{proposition}
\begin{appendixproof}[Proof of Proposition~\ref{prop:polymorphic-multi-barrelled-monad}]
	First, we show that $D_S$ is a strong monad.
	\begin{itemize}
		\item $(\eta \comp \leftunitor)^{\sharp} = \leftunitor$: The left-hand side is the strong Kleisli extension of
		\[ u : 1, x : A \vdash \forall \beta. \lambda k. \lambda h. k\ x : D_S A \]
		\begin{align}
			&(\eta \comp \leftunitor)^{\sharp} \\
			&= \forall \beta. \lambda k. \lambda h. x\ \beta\ (\lambda x. (\forall \beta. \lambda k. \lambda h. k\ x)\ \beta\ k\ h)\ h \\
			&= \forall \beta. \lambda k. \lambda h. x\ \beta\ (\lambda x. k\ x)\ h \\
			&= \forall \beta. \lambda k. \lambda h. x\ \beta\ k\ h \\
			&= x
		\end{align}
		\item $V^{\sharp} \comp (A_1 \times \eta_{A_2}) = V$: Let $V$ be a term of type
		$x_1 : A_1, x_2 : A_2 \vdash V : B$.
		\begin{align}
			&V^{\sharp} \comp (A_1 \times \eta_{A_2}) \\
			&= \forall \beta. \lambda k. \lambda h. (\forall \beta. \lambda k. \lambda h. k\ x_2)\ \beta\ (\lambda x_2. V\ \beta\ k\ h)\ h \\
			&= \forall \beta. \lambda k. \lambda h. (\lambda x_2. V\ \beta\ k\ h)\ x_2 \\
			&= \forall \beta. \lambda k. \lambda h. V\ \beta\ k\ h \\
			&= V
		\end{align}
		\item $W^{\sharp} \comp (A_1 \times V^{\sharp}) \comp \associator = (W^{\sharp} \comp (A_1 \times V) \comp \associator)^{\sharp}$: Let $V$ and $W$ be terms of type
		\[ x_2 : A_2, x_3 : A_3 \vdash V : D_S B_1 \qquad x_1 : A_1, y : B_1 \vdash W : D_S B_2 \]
		We show 
		\[ x_1 : A_1, x_2 : A_2, x_3 : D_S A_3 \vdash W^{\sharp} \comp (A_1 \times V^{\sharp}) \comp \associator = (W^{\sharp} \comp (A_1 \times V) \comp \associator)^{\sharp} : D_S B_2 \]
		\begin{align}
			W^{\sharp} &= \forall \beta. \lambda k. \lambda h. y\ \beta\ (\lambda y. W\ \beta\ k\ h)\ h \\
			V^{\sharp} &= \forall \beta. \lambda k. \lambda h. x_3\ \beta\ (\lambda x_3. V\ \beta\ k\ h)\ h
		\end{align}
		\begin{align}
			&W^{\sharp} \comp (A_1 \times V^{\sharp}) \comp \associator \\
			&= \forall \beta. \lambda k. \lambda h. (\forall \beta. \lambda k. \lambda h. x_3\ \beta\ (\lambda x_3. V\ \beta\ k\ h)\ h)\ \beta\ (\lambda y. W\ \beta\ k\ h)\ h \\
			&= \forall \beta. \lambda k. \lambda h. x_3\ \beta\ (\lambda x_3. V\ \beta\ (\lambda y. W\ \beta\ k\ h)\ h)\ h \\
			&(W^{\sharp} \comp (A_1 \times V) \comp \associator)^{\sharp} \\
			&= \forall \beta. \lambda k. \lambda h. x_3\ \beta\ (\lambda x_3. (\forall \beta. \lambda k. \lambda h. V\ \beta\ (\lambda y. W\ \beta\ k\ h)\ h)\ \beta\ k\ h)\ h \\
			&= \forall \beta. \lambda k. \lambda h. x_3\ \beta\ (\lambda x_3. V\ \beta\ (\lambda y. W\ \beta\ k\ h)\ h)\ h
		\end{align}
	\end{itemize}
	We show axioms for effect-handling strong monads.
	\begin{itemize}
		\item $\mathbf{handle}_{S, S', X} \comp (\identity{} \times \eta) = \pi_2$
		The right-hand side is
		\[ h : \mathcal{H}_S[D_{S'} A], x : D_{S'} A \vdash x : D_{S'} A \]
		and the left-hand side is
		\begin{align}
			&\mathbf{handle}_{S, S', X} \comp (\identity{} \times \eta) \\
			&= (\forall \beta. \lambda k. \lambda h. k\ x)\ (D_{S'} A)\ (\lambda y. y)\ h \\
			&= x
		\end{align}
		\item $\mathbf{handle}_{S, S', X} \comp (\identity{} \times \mu) = \mathbf{handle}_{S, S', X} \comp (\identity{} \times T_S \mathbf{handle}_{S, S', X}) \comp \langle \pi_1, \strength \rangle$
		This is equivalent to
		\[ \mathbf{handle}_{S, S', X} \comp (\identity{} \times \mu) = \mathbf{handle}_{S, S', X} \comp \langle \pi_1, (\eta \comp \mathbf{handle}_{S, S', X})^{\sharp} \rangle \]
		Since $\mu = \leftunitor^{\sharp} \comp \leftunitor^{-1}$, $\mu$ is given as the following term.
		\[ x : (D_S)^2 D_{S'} A \vdash \forall \beta.\ \lambda k.\ \lambda h.\ x\ \beta\ (\lambda x.\ x\ \beta\ k\ h)\ h : D_S D_{S'} A \]
		The left-hand side $h : \mathcal{H}_S[D_{S'} A], x : (D_S)^2 D_{S'} A \vdash \mathbf{handle}_{S, S', X} \comp (\identity{} \times \mu) : D_{S'} A$ is given as follows.
		\begin{align}
			&\mathbf{handle}_{S, S', X} \comp (\identity{} \times \mu) \\
			&= (\forall \beta.\ \lambda k.\ \lambda h.\ x\ \beta\ (\lambda x.\ x\ \beta\ k\ h)\ h)\ (D_{S'} A)\ (\lambda y.\ y)\ h \\
			&= x\ (D_{S'} A)\ (\lambda x.\ x\ (D_{S'} A)\ (\lambda y.\ y)\ h)\ h
		\end{align}
		The right-hand side is given as follows.
		\begin{itemize}
			\item $h : \mathcal{H}_S[D_{S'} A], x : D_S D_{S'} A \vdash \eta \comp \mathbf{handle}_{S, S', X} : D_S D_{S'} A$ is
			\[ \forall \beta.\ \lambda k.\ \lambda h'.\ k\ (x\ (D_{S'} A)\ (\lambda y.\ y)\ h) \]
			\item $h : \mathcal{H}_S[D_{S'} A], x : (D_S)^2 D_{S'} A \vdash (\eta \comp \mathbf{handle}_{S, S', X})^{\sharp} : D_S D_{S'} A$ is
			\begin{align}
				&\forall \beta.\ \lambda k.\ \lambda h'.\ x\ \beta\ (\lambda x.\ (\forall \beta.\ \lambda k.\ \lambda h'.\ k\ (x\ (D_{S'} A)\ (\lambda y.\ y)\ h))\ \beta\ k\ h')\ h' \\
				&= \forall \beta.\ \lambda k.\ \lambda h'.\ x\ \beta\ (\lambda x.\ k\ (x\ (D_{S'} A)\ (\lambda y.\ y)\ h))\ h'
			\end{align}
		\end{itemize}
		\begin{align}
			&\mathbf{handle}_{S, S', X} \comp \langle \pi_1, (\eta \comp \mathbf{handle}_{S, S', X})^{\sharp} \rangle \\
			&= (\forall \beta.\ \lambda k.\ \lambda h'.\ x\ \beta\ (\lambda x.\ k\ (x\ (D_{S'} A)\ (\lambda y.\ y)\ h))\ h')\ (D_{S'} A)\ (\lambda y.\ y)\ h \\
			&= x\ (D_{S'} A)\ (\lambda x.\ (\lambda y.\ y)\ (x\ (D_{S'} A)\ (\lambda y.\ y)\ h))\ h \\
			&= x\ (D_{S'} A)\ (\lambda x.\ x\ (D_{S'} A)\ (\lambda y.\ y)\ h)\ h
		\end{align}
		\item $\mathbf{handle}_{S, S', X} \comp (\identity{} \times \MixedAlgOp{S'}{X}{\interpret{\mathtt{op}}_S}) = \eval \comp (\pi_{\mathtt{op}} \times \identity{})$
		Using strong Kleisli extension, we can rewrite $\MixedAlgOp{S'}{X}{\interpret{\mathtt{op}}_S}$ as follows.
		\begin{align}
			&\MixedAlgOp{S'}{X}{\interpret{\mathtt{op}}_S} \\
			&= D_S \eval \comp \strength^{T_S} \comp \braiding \comp (\interpret{\mathtt{op}}_S \times \identity{}) \\
			&= (\eta \comp \eval)^{\sharp} \comp \braiding \comp (\interpret{\mathtt{op}}_S \times \identity{})
		\end{align}
		Then, we unfold
		\[ x : A_{\mathtt{op}}, c : B_{\mathtt{op}} \to D_{S'} C \vdash \MixedAlgOp{S'}{X}{\interpret{\mathtt{op}}_S} : D_{S} D_{S'} C \]
		as follows.
		\begin{itemize}
			\item $x : A_{\mathtt{op}} \vdash \interpret{\mathtt{op}}_S : D_S B_{\mathtt{op}}$ is $\forall \beta. \lambda k. \lambda h. \pi_{\mathtt{op}}\ h\ \langle x, k \rangle$.
			\item $c : B_{\mathtt{op}} \to D_{S'} C, b  : B_{\mathtt{op}} \vdash \eta \comp \eval : D_S D_{S'} C$ is $\forall \beta. \lambda k. \lambda h. k\ (c\ b)$.
			\item $\MixedAlgOp{S'}{X}{\interpret{\mathtt{op}}_S}$ is given by
			\begin{align}
				&\MixedAlgOp{S'}{X}{\interpret{\mathtt{op}}_S} \\
				&= (\eta \comp \eval)^{\sharp} \comp \braiding \comp (\interpret{\mathtt{op}}_S \times \identity{}) \\
				&= \forall \beta. \lambda k. \lambda h. (\forall \beta. \lambda k. \lambda h. \pi_{\mathtt{op}}\ h\ \langle x, k \rangle)\ \beta\ (\lambda b. (\forall \beta. \lambda k. \lambda h. k\ (c\ b))\ \beta\ k\ h)\ h \\
				&= \forall \beta. \lambda k. \lambda h. (\forall \beta. \lambda k. \lambda h. \pi_{\mathtt{op}}\ h\ \langle x, k \rangle)\ \beta\ (\lambda b. k\ (c\ b))\ h \\
				&= \forall \beta. \lambda k. \lambda h. \pi_{\mathtt{op}}\ h\ \langle x, \lambda b. k\ (c\ b) \rangle
			\end{align}
		\end{itemize}
		Thus, the left-hand side
		\[ h : \mathcal{H}_S[D_{S'} C], x : A_{\mathtt{op}}, c : B_{\mathtt{op}} \to D_{S'} C \vdash \mathbf{handle}_{S, S', X} \comp (\identity{} \times \MixedAlgOp{S'}{X}{\interpret{\mathtt{op}}_S}) : D_{S'} C \]
		is given as follows.
		\begin{align}
			&\mathbf{handle}_{S, S', X} \comp (\identity{} \times \MixedAlgOp{S'}{X}{\interpret{\mathtt{op}}_S}) \\
			&= (\forall \beta. \lambda k. \lambda h. \pi_{\mathtt{op}}\ h\ \langle x, \lambda b. k\ (c\ b) \rangle)\ (D_{S'} A)\ (\lambda y. y)\ h \\
			&= \pi_{\mathtt{op}}\ h\ \langle x, \lambda b. (\lambda y. y)\ (c\ b) \rangle \\
			&= \pi_{\mathtt{op}}\ h\ \langle x, \lambda b. c\ b \rangle \\
			&= \pi_{\mathtt{op}}\ h\ \langle x, c \rangle
		\end{align}
		On the other hand, the right-hand side is
		\begin{align}
			&\eval \comp (\pi_{\mathtt{op}} \times \identity{}) \\
			&= \pi_{\mathtt{op}}\ h\ \langle x, c \rangle
		\end{align}
	\end{itemize}
\end{appendixproof}

\begin{theorem}
	The interpretation of a term $M$ in the $\EffectHandlerCalculus$-model in Proposition~\ref{prop:polymorphic-multi-barrelled-monad} coincides with the interpretation of the CPS-transformed term $\CPS{M}$ in the $\lambda{\to}{\forall}$-fibration $\TermModel{\lambda{\to}{\forall}}$, i.e., the equivalence class of $\CPS{M}$ module $\beta$/$\eta$-equivalence.
	\qed
\end{theorem}

By the soundness of the denotational semantics (Theorem~\ref{thm:soundness}), it follows immediately that the CPS transformation is correct in the following sense.
\begin{corollary}
	If $\Gamma \vdash M = N : A ! \Sigma$ in $\EffectHandlerCalculus$, then $\CPS{M}$ and $\CPS{N}$ are equal modulo $\beta$-/$\eta$-equivalence in System F.
	In particular, if $M \leadsto N$ in $\EffectHandlerCalculus$, then $\CPS{M}$ and $\CPS{N}$ are equal modulo $\beta$/$\eta$-equivalence.
	\qed
\end{corollary}

\section{Incorporating Effect Theories}
\label{sec:with-effect-theories}
The idea of effect handlers originates from the study of algebraic theories \cite{PlotkinESOP2009,PlotkinLMCS2013}, which are specified by operations and equational axioms between operations.
However, the effect handler calculus $\EffectHandlerCalculus^{+}$ that we presented in Section~\ref{sec:syntax} and~\ref{sec:semantics} does not take equational axioms between operations into consideration.
In this section, we consider an extension of $\EffectHandlerCalculus^{+}$ that incorporates equational axioms between operations.
We provide categorical semantics of this extended calculus and show its soundness and completeness.

Before going into details, we give a brief technical overview of this section.
The syntax of the extended calculus, which we call $\EffectHandlerCalculusWithEquations$, is largely inspired by \cite{LuksicJFunctProg2020}.
The main difference from $\EffectHandlerCalculus^{+}$ is that we use computation types of the form $A ! \Sigma / \mathcal{E}$ where $\mathcal{E}$ is a set of equational axioms between operations in $\Sigma$.
This type is intended to represent computations that may perform operations in $\Sigma$ and are guaranteed to satisfy the equations in $\mathcal{E}$.
To type $\handlewithto{M}{H}{x}{N}$ where $M$ has type $A ! \Sigma / \mathcal{E}$, we require that the effect handler $H$ respects the equations in $\mathcal{E}$, as we will later explain in Figure~\ref{fig:typing-rules-effect-handlers-with-equations}.
We call a pair $(\Sigma, \mathcal{E})$ an \emph{effect theory}, following the terminology of \cite{LuksicJFunctProg2020,PlotkinLMCS2013}.

The semantics of $\EffectHandlerCalculusWithEquations$ is defined by adapting the semantics of $\EffectHandlerCalculus^{+}$ presented in Section~\ref{sec:semantics}.
Although the structure of the models is similar to that of $\EffectHandlerCalculus^{+}$, this arises a couple of technical difficulties.
One difficulty is caused by the modified typing rule for $\handlewithto{M}{H}{x}{N}$, which introduces a mutual dependence between typing rules and the derivation rules for equations.
To solve this, we define a partial interpretation of terms that does not depend on derivation rules for equations, and later show that this interpretation is defined for all well-typed terms in $\EffectHandlerCalculusWithEquations$ by mutual induction with the soundness theorem for $\EffectHandlerCalculusWithEquations$.
The other difficulty is to modify the definition of $\mathbf{handle}_{S, S', X}$ in Definition~\ref{def:effect-handler-model}.
Since $\handlewithto{M}{H}{x}{N}$ is only well-typed when $H$ respects the equations in the effect theory of $M$, the morphism $\mathbf{handle}_{S, S', X} : \HBundle{S}{X} \times T_S K^{T_{S'}} X \to K^{T_{S'}} X$ needs not be defined for any (generalized) element of $\HBundle{S}{X}$.
Thus, we define $\mathbf{handle}_{\mathcal{T}, \mathcal{T}', X} : H_{\mathcal{T}} \times \Yoneda T_{\mathcal{T}} K^{T_{\mathcal{T}'}} X \to \Yoneda K^{T_{\mathcal{T}'}} X$ as a morphism in the presheaf category $[\category{C}^{\op}, \Set]$ where $\Yoneda X = \category{C}({-}, X)$ is the Yoneda embedding and $H_{\mathcal{T}}$ is a subpresheaf of $\Yoneda \HBundle{S}{X}$ consisting of generalized elements that respect the equations in the effect theory $\mathcal{T} = (\Sigma, \mathcal{E})$.
Then, we show that this definition of models admits soundness (Theorem~\ref{thm:soundness-with-effect-theories}) and completeness (Theorem~\ref{thm:completeness-with-effect-theories}).

\subsection{Effect Theories}
\label{sec:effect-theories}
To define the syntax for effect theories $(\Sigma, \mathcal{E})$~\cite{LuksicJFunctProg2020,PlotkinLMCS2013}, we first consider the first-order fragment of value types and terms in $\EffectHandlerCalculus^{+}$, which we call \emph{ground types} and \emph{ground terms}, respectively.
\begin{align}
	\hat{A}, \hat{B} \quad&\coloneqq\quad b \mid \hat{A}_1 \times \dots \times \hat{A}_n \mid \hat{A}_1 + \dots + \hat{A}_n \\
	\hat{V}, \hat{W} \quad&\coloneqq\quad x \mid \langle \hat{V}_1, \dots, \hat{V}_n \rangle \mid \pi_i\ \hat{V} \mid \iota_i\ \hat{V}
\end{align}
Here, $b$ ranges over base types.
We write ground types and terms with hats to distinguish them from value types and terms in the effect handler calculus $\EffectHandlerCalculusWithEquations$, which we will define later.
Typing judgements for ground terms are of the form $\hat{\Delta} \vdash \hat{V} : \hat{A}$ where $\hat{\Delta}$ is a context containing only variables of ground types.
Typing rules are defined in the same way as in $\EffectHandlerCalculus^{+}$.

Equational axioms between operations are specified using \emph{effect theory terms} defined as follows.
\begin{align}
	\hat{M}, \hat{N} \quad&\coloneqq\quad \return{\hat{V}} \mid \letin{x}{\mathtt{op}(\hat{V})}{\hat{M}} \\
	&\qquad\quad\mid \caseof{\hat{V}}{\casepattern{\iota_i\ x_i}{\hat{M}_i}}_{i = 1}^n
\end{align}
The intuitive meaning of these terms is the same as computation terms in $\EffectHandlerCalculus^{+}$.
Note that application of operations is often written as $\mathtt{op}(\hat{V}; x. \hat{M})$ in the literature \cite{LuksicJFunctProg2020}, which is equivalent to our syntax $\letin{x}{\mathtt{op}(\hat{V})}{\hat{M}}$.
Typing judgements for these terms are written as $\hat{\Delta} \vdash \hat{M} : \hat{A} ! \Sigma$, and their derivation rules are defined in the similar way to $\EffectHandlerCalculus^{+}$.
Here, types appearing in \emph{operation signatures} (or simply \emph{signatures}) $\Sigma$ are restricted to ground types.
\[ \Sigma \quad\coloneqq\quad \emptyset \mid \{ \mathtt{op} : \hat{A} \rightarrowtriangle \hat{B} \} \cup \Sigma \]

\begin{toappendix}
Formally, the typing rules for effect theory terms are given as follows.
\begin{mathpar}
	\inferrule{
		\hat{\Delta} \vdash \hat{V} : \hat{A}
	}{
		\hat{\Delta} \vdash \return{\hat{V}} : \hat{A} ! E
	}
	\and
	\inferrule{
		\hat{\Delta} \vdash \hat{V} : \hat{A}_{\mathtt{op}} \\
		(\mathtt{op} : \hat{A}_{\mathtt{op}} \rightarrowtriangle \hat{B}_{\mathtt{op}}) \in \Sigma \\
		\hat{\Delta}, x : \hat{B}_{\mathtt{op}} \vdash \hat{M} : \hat{A} ! \Sigma
	}{
		\hat{\Delta} \vdash \letin{x}{\mathtt{op}(\hat{V})}{\hat{M}} : \hat{A} ! \Sigma
	}
	\and
	\inferrule{
		\hat{\Delta} \vdash \hat{V} : \hat{A}_1 + \dots + \hat{A}_n \\
		\forall i\quad \hat{\Delta}, x_i : \hat{A}_i \vdash \hat{M}_i : \hat{B} ! \Sigma
	}{
		\hat{\Delta} \vdash \caseof{\hat{V}}{\casepattern{\iota_i\ x_i}{\hat{M}_i}}_{i = 1}^n : \hat{B} ! \Sigma
	}
\end{mathpar}
\end{toappendix}

\begin{definition}
	An \emph{equational axiom} for a signature $\Sigma$ is a pair of well-typed terms $\hat{\Delta} \vdash \hat{M} : \hat{A} ! \Sigma$ and $\hat{\Delta} \vdash \hat{N} : \hat{A} ! \Sigma$ that have the same context and type.
	Equational axioms are often written as $\hat{\Delta} \vdash \hat{M} \sim \hat{N} : \hat{A} ! \Sigma$ for brevity.
	An \emph{effect theory} is a pair $(\Sigma, \mathcal{E})$ where $\Sigma$ is a signature and $\mathcal{E}$ is a set of equational axioms.
\end{definition}

\begin{example}
	The equational theory for nondeterminism is given as follows.
	Let $\Sigma = \{ \mathtt{choice} : \mathtt{unit} \rightarrowtriangle \mathtt{unit} + \mathtt{unit} \}$.
	For notational simplicity, we define the following syntactic sugar.
	\[ \hat{M} \mathrel{\mathtt{or}} \hat{N}\ \coloneqq\ \letin{x}{\mathtt{choice}\ \langle\rangle}{\caseof{x}{\casepattern{\iota_1\ \_}{\hat{M}}; \casepattern{\iota_2\ \_}{\hat{N}}}} \]
	Then the idempotence, commutativity, and associativity of $\mathtt{choice}$ are expressed as the following equational axioms.
	\begin{align}
		&\vdash \return{\langle\rangle} \mathrel{\mathtt{or}} \return{\langle\rangle} \sim \return{\langle\rangle} : \mathtt{unit} ! \Sigma \\
		&\vdash \return{\iota_1\ \langle\rangle} \mathrel{\mathtt{or}} \return{\iota_2\ \langle\rangle} \\
		&\qquad \sim \return{\iota_2\ \langle\rangle} \mathrel{\mathtt{or}} \return{\iota_1\ \langle\rangle} : \mathtt{unit} + \mathtt{unit} ! \Sigma \\
		&\vdash (\return{\iota_1\ \langle\rangle} \mathrel{\mathtt{or}} \return{\iota_2\ \langle\rangle}) \mathrel{\mathtt{or}} \return{\iota_3\ \langle\rangle} \\
		&\qquad \sim \return{\iota_1\ \langle\rangle} \mathrel{\mathtt{or}} (\return{\iota_2\ \langle\rangle} \mathrel{\mathtt{or}} \return{\iota_3\ \langle\rangle}) \\
		&\qquad\qquad : \mathtt{unit} + \mathtt{unit} + \mathtt{unit} ! \Sigma
	\end{align}
\end{example}

\subsection{Two Interpretations of Effect Theory Terms}

Let $T$ be a strong monad on a cartesian category $\category{C}$ such that $\category{C}$ has distributive Kleisli coproducts for $T$.
Given an effect theory $(\Sigma, \mathcal{E})$, suppose that we have an interpretation $\interpret{\mathtt{op}} : \interpret{\hat{A}_{\mathtt{op}}} \to T \interpret{\hat{B}_{\mathtt{op}}}$ for each operation $(\mathtt{op} : \hat{A}_{\mathtt{op}} \rightarrowtriangle \hat{B}_{\mathtt{op}}) \in \Sigma$.
Here, for each ground type $\hat{A}$, its interpretation $\interpret{\hat{A}} \in \category{C}$ is defined in the standard way as in $\EffectHandlerCalculus^{+}$.
We consider when $T$ satisfies the equations in $\mathcal{E}$.

\begin{definition}\label{def:operation-satisfies-equational-axioms}
	We say $\{ \interpret{\mathtt{op}} : \interpret{\hat{A}_{\mathtt{op}}} \to T \interpret{\hat{B}_{\mathtt{op}}} \}_{(\mathtt{op} : \hat{A}_{\mathtt{op}} \rightarrowtriangle \hat{B}_{\mathtt{op}}) \in \Sigma}$ \emph{satisfies equational axioms} in $\mathcal{E}$ if for each equational axiom $\hat{\Delta} \vdash \hat{M} \sim \hat{N} : \hat{A} ! \Sigma$ in $\mathcal{E}$, we have $\interpret{\hat{M}} = \interpret{\hat{N}}$ where $\interpret{\hat{M}} : \interpret{\hat{\Delta}} \to T \interpret{\hat{A}}$ is defined as follows.
	\begin{align}
		&\interpret{\return{\hat{V}}}\ \coloneqq\quad \eta \comp \interpret{\hat{V}} \\
		&\interpret{\letin{x}{\mathtt{op}(\hat{V})}{\hat{M}}}\ \coloneqq\quad \mu \comp T \interpret{\hat{M}} \comp \strength \comp \tupling{\identity{}}{\interpret{\mathtt{op}} \comp \interpret{\hat{V}}} \\
		&\interpret{\caseof{\hat{V}}{\casepattern{\iota_i\ x_i}{\hat{M}_i}}_{i = 1}^n}\ \coloneqq\quad [\interpret{\hat{M}_1}, \dots, \interpret{\hat{M}_n}] \comp \tupling{\identity{}}{\interpret{\hat{V}}}
	\end{align}
	Note that this interpretation is essentially the same as the one for computation terms in $\EffectHandlerCalculus^{+}$.
\end{definition}

In $\EffectHandlerCalculus^{+}$, the interpretation of an effect handler $\Gamma \vdash H : \Sigma \Rightarrow A ! \Sigma'$ is given by $\interpret{H} : \interpret{\Gamma} \to \HBundle{\Sigma}{J^{T_{\Sigma'}} \interpret{A}}$.
To define when $\interpret{H}$ satisfies the equations in $\mathcal{E}$, we use another interpretation of effect theory terms defined as follows.
\begin{definition}\label{def:algebra-satisfies-equational-axioms}
	For each $\hat{\Delta} \vdash \hat{M} : \hat{A} ! \Sigma$ and $X \in \category{C}_T$, we define
	\[ \AlgInt{\hat{M}}^T_X : \HBundle{\Sigma}{X} \to \KleisliExp{T}{\interpret{\hat{\Delta}} \times (\KleisliExp{T}{\interpret{\hat{A}}}{X})}{X} \]
	by induction on $\hat{M}$ as follows.
	\begin{align}
		&\Lambda^{-1} \AlgInt{\return{\hat{V}}}^T_X \quad\coloneqq\quad \eval \comp \braiding \comp (\interpret{\hat{V}} \times \identity{}) \comp \pi_2 \\
		&\Lambda^{-1} \AlgInt{\letin{x}{\mathtt{op}(\hat{V})}{\hat{M}}}^T_X \\
		&\coloneqq\ \eval \comp \tupling{\pi_{\mathtt{op}} \comp \pi_1}{\tupling{\interpret{\hat{V}} \comp \pi_1 \comp \pi_2}{\Lambda (\Lambda^{-1} \AlgInt{\hat{M}}^T_X \comp w)}} \\
		&\Lambda^{-1} \AlgInt{\caseof{\hat{V}}{\casepattern{\iota_i\ x_i}{\hat{M}_i}}_{i = 1}^n}^T_X \\
		&\coloneqq\ [\Lambda^{-1} \AlgInt{\hat{M}_1}^T_X \comp w, \dots, \Lambda^{-1} \AlgInt{\hat{M}_n}^T_X \comp w] \comp \tupling{\identity{}}{\interpret{\hat{V}} \comp \pi_1 \comp \pi_2}
	\end{align}
	Here,
	$w : (\HBundle{\Sigma}{X} \times (\interpret{\hat{\Delta}} \times (\KleisliExp{T}{\interpret{\hat{B}}}{X}))) \times \interpret{\hat{A}} \to \HBundle{\Sigma}{X} \times (\interpret{\hat{\Delta}, x : \hat{A}} \times (\KleisliExp{T}{\interpret{\hat{B}}}{X}))$ is the canonical isomorphism rearranging the components.
	We say that $f : Z \to \HBundle{\Sigma}{X}$ \emph{satisfies equational axioms} in $\mathcal{E}$ if for each equational axiom $\hat{\Delta} \vdash \hat{M} \sim \hat{N} : \hat{A} ! \Sigma$ in $\mathcal{E}$, we have $\AlgInt{\hat{M}}^T_X \comp f = \AlgInt{\hat{N}}^T_X \comp f$.
\end{definition}

As expected, these two interpretations are closely related.
It is known~\cite{PlotkinApplCategStruct2003} that there is a one-to-one correspondence between generic effects of the form $A \to T B$ and algebraic operations of the form $\KleisliExp{T}{B}{X} \to \KleisliExp{T}{A}{X}$.
Using this correspondence, we can relate the two interpretations above.
Given a morphism $f : A \to T B$, we define $A \times (\KleisliExp{T}{B}{X}) \to K X$ for each $X \in \category{C}_T$ by $\lfloor f \rfloor_X \coloneqq \mu \comp T \eval \comp \strength \comp \braiding \comp (f \times \identity{})$.
\begin{lemma}\label{lem:two-interpretations-of-effect-theory-terms}
	For each effect theory term $\hat{\Delta} \vdash \hat{M} : \hat{A} ! \Sigma$ and $X \in \category{C}_T$, we have the following equation relating $\AlgInt{\hat{M}}^T_X$ and $\interpret{\hat{M}}$.
	\begin{align}
		&\AlgInt{\hat{M}}^T_X \comp \langle \Lambda (\lfloor \interpret{\mathtt{op}} \rfloor \comp \pi_2) \rangle_{\mathtt{op} \in \Sigma} \quad=\quad \Lambda (\lfloor \interpret{\hat{M}} \rfloor_X \comp \pi_2) \\
		&\quad\qquad: 1 \to \KleisliExp{T}{(\interpret{\hat{\Delta}} \times (\KleisliExp{T}{\interpret{\hat{A}}}{X}))}{X}
		\tag*{\qed}
	\end{align}
\end{lemma}
\begin{appendixproof}[Proof of Lemma~\ref{lem:two-interpretations-of-effect-theory-terms}]
	\allowdisplaybreaks
	By induction on $\hat{M}$.
	For brevity, we define $h \coloneqq \langle \Lambda (\lfloor \interpret{\mathtt{op}} \rfloor \comp \pi_2) \rangle_{\mathtt{op} \in \Sigma}$.
	\begin{itemize}
		\item If $\hat{M} = \return{\hat{V}}$, then
		\begin{align}
			&\AlgInt{\hat{M}}^T_X \comp h \\
			&= \Lambda (\eval \comp \braiding \comp (\interpret{\hat{V}} \times \identity{}) \comp \pi_2) \comp h \\
			&= \Lambda (\eval \comp \braiding \comp (\interpret{\hat{V}} \times \identity{}) \comp \pi_2) \\
			&\Lambda (\lfloor \interpret{\hat{M}} \rfloor_X \comp \pi_2) \\
			&= \Lambda (\mu \comp T \eval \comp \strength \comp \braiding \comp ((\eta \comp \interpret{\hat{V}}) \times \identity{}) \comp \pi_2) \\
			&= \Lambda (\eval \comp \braiding \comp (\interpret{\hat{V}} \times \identity{}) \comp \pi_2)
		\end{align}
		\item If $\hat{M} = \letin{x}{\mathtt{op}(\hat{V})}{\hat{N}}$, then
		\begin{align}
			&\AlgInt{\hat{M}}^T_X \comp h \\
			&= \Lambda (\eval \comp \tupling{\pi_{\mathtt{op}} \comp \pi_1}{\tupling{\interpret{\hat{V}} \comp \pi_1 \comp \pi_2}{\Lambda (\Lambda^{-1} \AlgInt{\hat{N}}^T_X \comp w)}}) \comp h \\
			&= \Lambda (\eval \comp \tupling{\Lambda (\lfloor \interpret{\mathtt{op}} \rfloor \comp \pi_2) \comp \pi_1}{\tupling{\interpret{\hat{V}} \comp \pi_1 \comp \pi_2}{\Lambda (\Lambda^{-1} \AlgInt{\hat{N}}^T_X \comp w) \comp (h \times \identity{})}}) \\
			&= \Lambda (\lfloor \interpret{\mathtt{op}} \rfloor \comp \pi_2 \comp \tupling{\pi_1}{\tupling{\interpret{\hat{V}} \comp \pi_1 \comp \pi_2}{\Lambda (\Lambda^{-1} \AlgInt{\hat{N}}^T_X \comp w) \comp (h \times \identity{})}}) \\
			&= \Lambda (\lfloor \interpret{\mathtt{op}} \rfloor \comp \tupling{\interpret{\hat{V}} \comp \pi_1 \comp \pi_2}{\Lambda (\Lambda^{-1} \AlgInt{\hat{N}}^T_X \comp w) \comp (h \times \identity{})}) \\
			&= \Lambda (\lfloor \interpret{\mathtt{op}} \rfloor \comp \tupling{\interpret{\hat{V}} \comp \pi_1 \comp \pi_2}{\Lambda (\Lambda^{-1} \AlgInt{\hat{N}}^T_X \comp (h \times \identity{}) \comp w)}) \\
			&= \Lambda (\lfloor \interpret{\mathtt{op}} \rfloor \comp \tupling{\interpret{\hat{V}} \comp \pi_1 \comp \pi_2}{\Lambda (\Lambda^{-1} (\AlgInt{\hat{N}}^T_X \comp h) \comp w)}) \\
			&= \Lambda (\lfloor \interpret{\mathtt{op}} \rfloor \comp \tupling{\interpret{\hat{V}} \comp \pi_1 \comp \pi_2}{\Lambda (\lfloor \interpret{\hat{N}} \rfloor \comp \pi_2 \comp w)}) \\
			&= \Lambda (\mu \comp T \eval \comp \strength \comp \braiding \comp (\interpret{\mathtt{op}} \times \identity{}) \comp \tupling{\interpret{\hat{V}} \comp \pi_1 \comp \pi_2}{\Lambda (\lfloor \interpret{\hat{N}} \rfloor \comp \pi_2 \comp w)}) \\
			&= \Lambda (\mu \comp T \eval \comp \strength \comp \tupling{\Lambda (\lfloor \interpret{\hat{N}} \rfloor \comp \pi_2 \comp w)}{\interpret{\mathtt{op}} \comp \interpret{\hat{V}} \comp \pi_1 \comp \pi_2}) \\
			&= \Lambda (\mu \comp T (\lfloor \interpret{\hat{N}} \rfloor \comp \pi_2 \comp w) \comp \strength \comp \tupling{\identity{}}{\interpret{\mathtt{op}} \comp \interpret{\hat{V}} \comp \pi_1 \comp \pi_2}) \\
			&= \Lambda (\mu \comp T (\mu \comp T \eval \comp \strength \comp \braiding \comp (\interpret{\hat{N}} \times \identity{}) \comp \pi_2 \comp w) \comp \strength \comp \tupling{\identity{}}{\interpret{\mathtt{op}} \comp \interpret{\hat{V}} \comp \pi_1 \comp \pi_2}) \\
			&= \Lambda (\mu \comp T \mu \comp T (T \eval \comp \strength \comp (\identity{} \times \interpret{\hat{N}}) \comp \braiding \comp \pi_2 \comp w) \comp \strength \comp \tupling{\identity{}}{\interpret{\mathtt{op}} \comp \interpret{\hat{V}} \comp \pi_1 \comp \pi_2}) \\
			&= \Lambda (\mu \comp T \mu \comp T (T \eval \comp \strength \comp (\identity{} \times \interpret{\hat{N}}) \comp \associator \comp ((\braiding \comp \pi_2) \times \identity{})) \comp \strength \comp \tupling{\identity{}}{\interpret{\mathtt{op}} \comp \interpret{\hat{V}} \comp \pi_1 \comp \pi_2}) \\
			&= \Lambda (\mu \comp T \mu \comp T (T \eval \comp \strength \comp (\identity{} \times \interpret{\hat{N}}) \comp \associator) \comp \strength \comp ((\braiding \comp \pi_2) \times \identity{}) \comp \tupling{\identity{}}{\interpret{\mathtt{op}} \comp \interpret{\hat{V}} \comp \pi_1 \comp \pi_2}) \\
			&= \Lambda (\mu \comp T \mu \comp T (T \eval \comp \strength \comp (\identity{} \times \interpret{\hat{N}}) \comp \associator) \comp \strength \comp \tupling{\braiding}{\interpret{\mathtt{op}} \comp \interpret{\hat{V}} \comp \pi_1} \comp \pi_2) \\
			&= \Lambda (\mu \comp T \mu \comp T (T \eval \comp \strength \comp (\identity{} \times \interpret{\hat{N}})) \comp \strength \comp (\identity{} \times \strength) \comp \associator \comp \tupling{\braiding}{\interpret{\mathtt{op}} \comp \interpret{\hat{V}} \comp \pi_1} \comp \pi_2) \\
			&= \Lambda (\mu \comp T \mu \comp T (T \eval \comp \strength \comp (\identity{} \times \interpret{\hat{N}})) \comp \strength \comp (\identity{} \times \strength) \comp \braiding \comp (\tupling{\identity{}}{\interpret{\mathtt{op}} \comp \interpret{\hat{V}}} \times \identity{}) \comp \pi_2) \\
			&= \Lambda (\mu \comp \mu \comp T (T \eval \comp \strength \comp (\identity{} \times \interpret{\hat{N}})) \comp \strength \comp (\identity{} \times \strength) \comp \braiding \comp (\tupling{\identity{}}{\interpret{\mathtt{op}} \comp \interpret{\hat{V}}} \times \identity{}) \comp \pi_2)
		\end{align}
		\begin{align}
			&\Lambda (\lfloor \interpret{\hat{M}} \rfloor \comp \pi_2) \\
			&= \Lambda (\lfloor \mu \comp T \interpret{\hat{N}} \comp \strength \comp \tupling{\identity{}}{\interpret{\mathtt{op}} \comp \interpret{\hat{V}}} \rfloor \comp \pi_2) \\
			&= \Lambda (\mu \comp T \eval \comp \strength \comp \braiding \comp ((\mu \comp T \interpret{\hat{N}} \comp \strength \comp \tupling{\identity{}}{\interpret{\mathtt{op}} \comp \interpret{\hat{V}}}) \times \identity{}) \comp \pi_2) \\
			&= \Lambda (\mu \comp T \eval \comp \mu \comp T \strength \comp \strength \comp \braiding \comp ((T \interpret{\hat{N}} \comp \strength \comp \tupling{\identity{}}{\interpret{\mathtt{op}} \comp \interpret{\hat{V}}}) \times \identity{}) \comp \pi_2) \\
			&= \Lambda (\mu \comp T \eval \comp \mu \comp T \strength \comp T (\identity{} \times \interpret{\hat{N}}) \comp \strength \comp (\identity{} \times \strength) \comp \braiding \comp (\tupling{\identity{}}{\interpret{\mathtt{op}} \comp \interpret{\hat{V}}} \times \identity{}) \comp \pi_2)
		\end{align}
		\item If $\hat{M} = \caseof{\hat{V}}{\casepattern{\iota_i\ x_i}{\hat{M}_i}}_{i = 1}^n$, then
		\begin{align}
			&\AlgInt{\hat{M}}^T_X \comp h \\
			&= \Lambda ([\Lambda^{-1} \AlgInt{\hat{M}_1}^T_X \comp w, \dots, \Lambda^{-1} \AlgInt{\hat{M}_n}^T_X \comp w] \comp \tupling{\identity{}}{\interpret{\hat{V}} \comp \pi_1 \comp \pi_2}) \comp h \\
			&= \Lambda ([\Lambda^{-1} \AlgInt{\hat{M}_1}^T_X \comp w, \dots, \Lambda^{-1} \AlgInt{\hat{M}_n}^T_X \comp w] \comp \tupling{h \times \identity{}}{\interpret{\hat{V}} \comp \pi_1 \comp \pi_2}) \\
			&= \Lambda ([\Lambda^{-1} \AlgInt{\hat{M}_1}^T_X \comp (h \times \identity{}) \comp w, \dots, \Lambda^{-1} \AlgInt{\hat{M}_n}^T_X \comp (h \times \identity{}) \comp w] \comp \tupling{\identity{}}{\interpret{\hat{V}} \comp \pi_1 \comp \pi_2}) \\
			&= \Lambda ([\lfloor \interpret{\hat{M}_1} \rfloor_X \comp \pi_2 \comp w, \dots, \lfloor \interpret{\hat{M}_n} \rfloor_X \comp \pi_2 \comp w] \comp \tupling{\identity{}}{\interpret{\hat{V}} \comp \pi_1 \comp \pi_2}) \\
			&= \Lambda (\mu \comp T \eval \comp \strength \comp \braiding \comp [(\interpret{\hat{M}_1} \times \identity{}) \comp \pi_2 \comp w, \dots, (\interpret{\hat{M}_n} \times \identity{}) \comp \pi_2 \comp w] \comp \tupling{\identity{}}{\interpret{\hat{V}} \comp \pi_1 \comp \pi_2}) \\
			&= \Lambda (\mu \comp T \eval \comp \strength \comp \braiding \comp \tupling{[\interpret{\hat{M}_1} \comp \pi_1 \comp \pi_2 \comp w, \dots, \interpret{\hat{M}_n} \comp \pi_1 \comp \pi_2 \comp w]}{\pi_2 \comp \pi_2 \comp w} \comp \tupling{\identity{}}{\interpret{\hat{V}} \comp \pi_1 \comp \pi_2}) \\
			&= \Lambda (\mu \comp T \eval \comp \strength \comp \braiding \comp \tupling{[\interpret{\hat{M}_1}, \dots, \interpret{\hat{M}_n}] \comp ((\pi_1 \comp \pi_2) \times \identity{})}{\pi_2 \comp \pi_2 \comp \pi_1} \comp \tupling{\identity{}}{\interpret{\hat{V}} \comp \pi_1 \comp \pi_2}) \\
			&= \Lambda (\mu \comp T \eval \comp \strength \comp \braiding \comp \tupling{[\interpret{\hat{M}_1}, \dots, \interpret{\hat{M}_n}] \comp \tupling{\identity{}}{\interpret{\hat{V}}} \comp \pi_1 \comp \pi_2}{\pi_2 \comp \pi_2}) \\
			&= \Lambda (\mu \comp T \eval \comp \strength \comp \braiding \comp \tupling{[\interpret{\hat{M}_1}, \dots, \interpret{\hat{M}_n}] \comp \tupling{\identity{}}{\interpret{\hat{V}}} \comp \pi_1}{\pi_2} \comp \pi_2) \\
			&= \Lambda (\mu \comp T \eval \comp \strength \comp \braiding \comp (([\interpret{\hat{M}_1}, \dots, \interpret{\hat{M}_n}] \comp \tupling{\identity{}}{\interpret{\hat{V}}}) \times \identity{}) \comp \pi_2) \\
			&= \Lambda (\lfloor \interpret{\hat{M}} \rfloor \comp \pi_2)
			\qedhere
		\end{align}
	\end{itemize}
\end{appendixproof}

By Lemma~\ref{lem:two-interpretations-of-effect-theory-terms}, it follows that Definition~\ref{def:operation-satisfies-equational-axioms} can be reduced to Definition~\ref{def:algebra-satisfies-equational-axioms}, and thus $\AlgInt{\hat{M}}^T_X$ can be  regarded as a more basic interpretation of effect theory terms.

\begin{corollary}\label{cor:equational-axioms-two-interpretations}
	An interpretation $\{ \interpret{\mathtt{op}} \}_{\mathtt{op} \in \Sigma}$ satisfies equational axioms in the sense of Definition~\ref{def:operation-satisfies-equational-axioms} if and only if for each $X \in \category{C}_T$, $h \coloneqq \langle \Lambda (\lfloor \interpret{\mathtt{op}} \rfloor \comp \pi_2) \rangle_{\mathtt{op} \in \Sigma} : 1 \to \HBundle{\Sigma}{X}$ satisfies equational axioms in the sense of Definition~\ref{def:algebra-satisfies-equational-axioms}.
	\qed
\end{corollary}
\begin{appendixproof}[Proof of Corollary~\ref{cor:equational-axioms-two-interpretations}]
	The only-if part is immediate from Lemma~\ref{lem:two-interpretations-of-effect-theory-terms}.
	For the if-part, it suffices to prove the following equation.
	\begin{equation}
		\eval \comp \tupling{\AlgInt{\hat{M}}^T_{\interpret{\hat{A}}} \comp \langle \Lambda (\lfloor \interpret{\mathtt{op}} \rfloor \comp \pi_2) \rangle_{\mathtt{op} \in \Sigma} \comp {!}}{\tupling{\identity{}}{\Lambda(\eta \comp \pi_2) \comp {!}}} \quad=\quad \interpret{\hat{M}}
		: \interpret{\hat{\Delta}} \to T \interpret{\hat{A}}
	\end{equation}
	This is proved by Lemma~\ref{lem:two-interpretations-of-effect-theory-terms} as follows.
	\begin{align}
		&\eval \comp \tupling{\AlgInt{\hat{M}}^T_{\interpret{\hat{A}}} \comp \langle \Lambda (\lfloor \interpret{\mathtt{op}} \rfloor \comp \pi_2) \rangle_{\mathtt{op} \in \Sigma} \comp {!}}{\tupling{\identity{}}{\Lambda(\eta \comp \pi_2) \comp {!}}} \\
		&= \eval \comp \tupling{\Lambda (\lfloor \interpret{\hat{M}} \rfloor_X \comp \pi_2) \comp {!}}{\tupling{\identity{}}{\Lambda(\eta \comp \pi_2) \comp {!}}} \\
		&= \lfloor \interpret{\hat{M}} \rfloor_X \comp \pi_2 \comp \tupling{{!}}{\tupling{\identity{}}{\Lambda(\eta \comp \pi_2) \comp {!}}} \\
		&= \lfloor \interpret{\hat{M}} \rfloor_X \comp \tupling{\identity{}}{\Lambda(\eta \comp \pi_2) \comp {!}} \\
		&= \mu \comp T \eval \comp \strength \comp \braiding \comp (\interpret{\hat{M}} \times \identity{}) \comp \tupling{\identity{}}{\Lambda(\eta \comp \pi_2) \comp {!}} \\
		&= \mu \comp T (\eval \comp (\Lambda(\eta \comp \pi_2) \times \identity{})) \comp \strength \comp \braiding \comp \tupling{\interpret{\hat{M}}}{{!}} \\
		&= \mu \comp T (\eta \comp \pi_2) \comp \strength \comp \tupling{{!}}{\interpret{\hat{M}}} \\
		&= \pi_2 \comp \tupling{{!}}{\interpret{\hat{M}}} \\
		&= \interpret{\hat{M}}
		\qedhere
	\end{align}
\end{appendixproof}

\subsection{Effect Handler Calculus with Effect Theories}

We define the type system of $\EffectHandlerCalculusWithEquations$ by extending that of $\EffectHandlerCalculus^{+}$ with effect theories.
\emph{Value types} $A, B$, \emph{computation types} $C, D$, and \emph{handler types} $H$ are defined as follows.
\begin{align}
	A, B \quad&\coloneqq\quad b \mid A \to C \mid A_1 \times \dots \times A_n \mid A_1 + \dots + A_n \\
	C, D \quad&\coloneqq\quad A ! \Sigma / \mathcal{E} \qquad\qquad\qquad
	H \quad\coloneqq\quad \Sigma \Rightarrow C
\end{align}
Here, $b$ ranges over base types.
Note that, in contrast to $\EffectHandlerCalculus^{+}$, signatures $\Sigma$ are now restricted to ground types only, as we have explained in Section~\ref{sec:effect-theories}.
Terms are defined as in $\EffectHandlerCalculus^{+}$.

\begin{figure}
	\begin{mathpar}
		\inferrule{
			\Gamma \vdash H : \Sigma \Rightarrow C \\
			\Gamma, \hat{\Delta}, k : \hat{A} \to C \vdash \hat{M}[H, k] = \hat{N}[H, k] : C \quad\forall\ \hat{\Delta} \vdash \hat{M} \sim \hat{N} : \hat{A} ! \Sigma\ \in \mathcal{E}
		}{
			\Gamma \vdash H : \Sigma \Rightarrow C \ \mathbf{respects}\ \mathcal{E}
		}
		\and
		\inferrule{
				\Gamma \vdash M : A ! \Sigma / \mathcal{E} \\
				\Gamma \vdash H : \Sigma \Rightarrow C\ \mathbf{respects}\ \mathcal{E} \\
				\Gamma, x : A \vdash N : C
			}{
				\Gamma \vdash \handlewithto{M}{H}{x : A}{N} : C
			}
	\end{mathpar}
	\caption{Typing rules for $\EffectHandlerCalculusWithEquations$ related to effect handlers.
	We implicitly assume that variables in $\Gamma$ and $\hat{\Delta}$ are disjoint.}
	\label{fig:typing-rules-effect-handlers-with-equations}
\end{figure}

\paragraph{Typing rules}
Most of the typing rules of $\EffectHandlerCalculusWithEquations$ are the same as those of $\EffectHandlerCalculus^{+}$, except that computation types $A ! \Sigma$ are replaced with $A ! \Sigma / \mathcal{E}$.
One notable difference is the typing rule for handle-with (Figure~\ref{fig:typing-rules-effect-handlers-with-equations}), where we use the following additional judgement.
\[ \Gamma \vdash H : \Sigma \Rightarrow C\ \mathbf{respects}\ \mathcal{E}, \]
This judgement means that an effect handler $H$ has type $\Gamma \vdash H : \Sigma \Rightarrow C$ and respects the equational axioms in the effect theory $(\Sigma, \mathcal{E})$.
The derivation rule for this judgement is defined in Figure~\ref{fig:typing-rules-effect-handlers-with-equations}.
In the rule, we use the following auxiliary translation.
\begin{definition}
	For each effect theory term $\hat{M}$ and fresh variable $k$, we define $\hat{M}[H, k]$ inductively as follows.
	\begin{align}
		&(\return{\hat{V}})[H, k] \quad\coloneqq\quad k\ \hat{V} \\
		&(\letin{x}{\mathtt{op}(\hat{V})}{\hat{N}})[H, k] \quad\coloneqq\quad M_{\mathtt{op}}[\hat{V}/x, \lambda x. \hat{N}[H, k]/k'] \\
		&(\caseof{\hat{V}}{\casepattern{\iota_i\ x_i}{\hat{M}_i}}_{i = 1}^n)[H, k] \\
		&\quad\coloneqq\quad \caseof{\hat{V}}{\casepattern{\iota_i\ x_i}{\hat{M}_i[H, k]}}_{i = 1}^n
	\end{align}
	Here, $M_{\mathtt{op}}$ is the operation clause $(\mathtt{op}(x, k') \mapsto M_{\mathtt{op}}) \in H$ for $\mathtt{op}$.
\end{definition}

\begin{lemma}\label{lem:effect-handler-translation-typing}
	If $\hat{\Delta} \vdash \hat{M} : \hat{A} ! \Sigma$ and $\Gamma \vdash H : \Sigma \Rightarrow C$, then $\Gamma, \hat{\Delta}, k : \hat{A} \to C \vdash \hat{M}[H, k] : C$.
	\qed
\end{lemma}
\begin{appendixproof}[Proof of Lemma~\ref{lem:effect-handler-translation-typing}]
	By induction on $\hat{M}$.
	\begin{itemize}
		\item If $\hat{M} = \return{\hat{V}}$, then $\Gamma, \hat{\Delta}, k : \hat{A} \to C \vdash k\ \hat{V} : C$.
		\item If $\hat{M} = \letin{x}{\mathtt{op}(\hat{V})}{\hat{N}}$, then by the typing rule for effect handlers, we have $\Gamma, x : \hat{B}_{\mathtt{op}} \vdash \hat{N} : \hat{A} ! \Sigma$ and $\Gamma, x : \hat{A}_{\mathtt{op}}, k' : \hat{B}_{\mathtt{op}} \to C \vdash M_{\mathtt{op}} : C$ where $(\mathtt{op}(x, k') \mapsto M_{\mathtt{op}}) \in H$.
		By the induction hypothesis, we have $\Gamma, x : \hat{B}_{\mathtt{op}}, \hat{\Delta}, k : \hat{A} \to C \vdash \hat{N}[H, k] : C$.
		Thus, by substituting $\hat{V}$ for $x$ and $\lambda x. \hat{N}[H, k]$ for $k'$, we have $\Gamma, \hat{\Delta}, k : \hat{A} \to C \vdash M_{\mathtt{op}}[\hat{V}/x, \lambda x. \hat{N}[H, k]/k'] : C$.
		\item If $\hat{M} = \caseof{\hat{V}}{\casepattern{\iota_i\ x_i}{\hat{N}_i}}_{i = 1}^n$, then by the induction hypothesis, we have $\Gamma, x : \hat{A}_i, \hat{\Delta}, k : \hat{A} \to C \vdash \hat{N}_i[H, k] : C$ for each $i$.
		Thus, by the typing rule for case expressions, we have $\Gamma, \hat{\Delta}, k : \hat{A} \to C \vdash \caseof{\hat{V}}{\casepattern{\iota_i\ x_i}{\hat{N}_i[H, k]}}_{i = 1}^n : C$.
		\qedhere
	\end{itemize}
\end{appendixproof}

The intuition of the above translation is that $\hat{M}[H, k]$ represents the computation term obtained by replacing each operation call in $\hat{M}$ with the corresponding clause in the effect handler $H$, and finally passing the result to the continuation $k$.
A similar translation is also used in \cite{LuksicJFunctProg2020,PlotkinLMCS2013}.

\paragraph{Equational theory}
In $\EffectHandlerCalculusWithEquations$, we define the typing rules and the derivation rules for equations $\Gamma \vdash M = N : C$ by simultaneous induction, since deriving the judgement $\Gamma \vdash H : \Sigma \Rightarrow C\ \mathbf{respects}\ \mathcal{E}$ requires the judgement for equality of computation terms.

The derivation rules for equations $\Gamma \vdash M = N : C$ are mostly the same as in $\EffectHandlerCalculus^{+}$, but we add a new rule for instantiating equational axioms in effect theories.

\begin{mathpar}
	\inferrule[\hypertarget{rule:Inst-Ax}{Inst-Ax}]{
		\hat{\Delta} \vdash \hat{M} \sim \hat{N} : \hat{A} ! \Sigma\ \in\ \mathcal{E} \\
		\sigma : \Gamma \to \hat{\Delta}
	}{
		\Gamma \vdash \hat{M}[\sigma] = \hat{N}[\sigma] : \hat{A} ! \Sigma / \mathcal{E}
	}
\end{mathpar}
Here, $\sigma : \Gamma \to \hat{\Delta}$ is a substitution such that for each $x : \hat{B} \in \hat{\Delta}$, $\sigma(x)$ is a value term of $\EffectHandlerCalculusWithEquations$ such that $\Gamma \vdash \sigma(x) : \hat{B}$, and $\hat{M}[\sigma]$ is the computation term of $\EffectHandlerCalculusWithEquations$ obtained by substituting each variable $x$ in $\hat{M}$ with $\sigma(x)$.

\subsection{Denotational Semantics}
Now, we extend the definition of $\EffectHandlerCalculus^{+}$-models to interpret $\EffectHandlerCalculusWithEquations$.

\begin{definition}
	A \emph{$\EffectHandlerCalculusWithEquations$-model} consists of the following.
	\begin{itemize}
		\item A cartesian category $\category{C}$ and a family of strong monads $T_{(\Sigma, \mathcal{E})}$ on $\category{C}$ indexed by effect theory $(\Sigma, \mathcal{E})$ such that $\category{C}$ has Kleisli exponentials and joint distributive Kleisli coproducts.
		\item For each effect theory $(\Sigma, \mathcal{E})$ and $(\mathtt{op} : \hat{A}_{\mathtt{op}} \rightarrowtriangle \hat{B}_{\mathtt{op}}) \in \Sigma$, a morphism $\interpret{\mathtt{op}}_{(\Sigma, \mathcal{E})} : \interpret{\hat{A}_{\mathtt{op}}} \to T_{(\Sigma, \mathcal{E})} \interpret{\hat{B}_{\mathtt{op}}}$ in $\category{C}$ that satisfies equational axioms in the sense of Definition~\ref{def:operation-satisfies-equational-axioms}.
		\item For each effect theory $\mathcal{T} = (\Sigma, \mathcal{E})$, $\mathcal{T}' = (\Sigma', \mathcal{E}')$ and object $X \in \category{C}_{T_{\mathcal{T}'}}$, a natural transformation
		\[ \mathbf{handle}_{\mathcal{T}, \mathcal{T}', X} : H_{\mathcal{T}} \times \Yoneda T_{\mathcal{T}} K^{T_{\mathcal{T}'}} X \to \Yoneda K^{T_{\mathcal{T}'}} X \]
		where $\Yoneda X = \category{C}({-}, X)$ is the Yoneda embedding and $H_{(\Sigma, \mathcal{E})} : \category{C}^{\op} \to \Set$ is a subobject of $\Yoneda \HBundle{\Sigma}{X}$ defined by
		\[ H_{(\Sigma, \mathcal{E})}(Z) \coloneqq \{ f \in \category{C}(Z, \HBundle{\Sigma}{X}) \mid \text{$f$ satisfies $\mathcal{E}$ (Definition~\ref{def:algebra-satisfies-equational-axioms})} \}. \]
		The naturality of $\mathbf{handle}_{\mathcal{T}, \mathcal{T}', X}$ means that we have the following equation.
		\[ \mathbf{handle}_{\mathcal{T}, \mathcal{T}', X}(h, f) \comp g \quad=\quad \mathbf{handle}_{\mathcal{T}, \mathcal{T}', X}(h \comp g, f \comp g) \]
		It is required that $\mathbf{handle}$ satisfies the following equations for any $h \in H_{\mathcal{T}}(Z)$.
		\begin{align}
			&\mathbf{handle}_{\mathcal{T}, \mathcal{T}', X} (h, \eta \comp f) = f \quad \text{for $f \in \category{C}(Z, K^{T_{\mathcal{T}'}} X)$} \label{eq:effect-theory-handle-unit} \\
			&\mathbf{handle}_{\mathcal{T}, \mathcal{T}', X} (h, \mu \comp f) \\
			&= \mathbf{handle}_{\mathcal{T}, \mathcal{T}', X} (h, T \mathbf{handle}_{\mathcal{T}, \mathcal{T}', X} (h \comp \pi_1, \pi_2) \comp \strength \comp \tupling{\identity{}}{f}) \\
			&\qquad\text{for $f \in \category{C}(Z, T_{\mathcal{T}}^2 T_{\mathcal{T}'} X)$} \label{eq:effect-theory-handle-multiplication} \\
			&\mathbf{handle}_{\mathcal{T}, \mathcal{T}', X} (h, \MixedAlgOp{\mathcal{T}'}{X}{\interpret{\mathtt{op}}^T_{\mathcal{T}}} \comp f) = \eval \comp \tupling{\pi_{\mathtt{op}} \comp h}{f} \\
			&\qquad\text{for $f \in \category{C}(Z, A_{\mathtt{op}} \times (\KleisliExp{T_{\mathcal{T}'}}{B_{\mathtt{op}}}{X}))$} \label{eq:effect-theory-handle-operation}
		\end{align}
	\end{itemize}
\end{definition}
The main differences from $\EffectHandlerCalculus^{+}$-models are that (a) $\interpret{\mathtt{op}}$ must satisfy equational axioms, and (b) $\mathbf{handle}$ is defined as a natural transformation between presheaves.
The latter requirement is a bit technical, but the intuition behind it is that $H_{(\Sigma, \mathcal{E})}$ in the above definition represents the interpretation of effect handlers $\Gamma \vdash H : \Sigma \Rightarrow C \ \mathbf{respects}\ \mathcal{E}$ that respect the equational axioms in $\mathcal{E}$, and thus $\mathbf{handle}$ corresponds to the typing rule for handle-with in Figure~\ref{fig:typing-rules-effect-handlers-with-equations}.

If $\category{C}$ is a sufficiently complete category, then similarly to the case of $\EffectHandlerCalculus^{+}$-models, $\mathbf{handle}$ can be represented as a morphism in $\category{C}$ rather than as a natural transformation between presheaves.

\begin{proposition}\label{prop:representing-handle-with}
	If $\category{C}$ has joint equalizers, then $H_{\mathcal{T}} : \category{C}^{\op} \to \Set$ is representable as $H_{\mathcal{T}} \cong \Yoneda Q$ for some object $Q \in \category{C}$, and $\mathbf{handle}_{\mathcal{T}, \mathcal{T}', X}$ can be represented as a morphism of the form
	$\mathbf{handle}'_{\mathcal{T}, \mathcal{T}', X} : Q \times T_{\mathcal{T}} K^{T_{\mathcal{T}'}} X \to K^{T_{\mathcal{T}'}} X$.
	\qed
\end{proposition}
\begin{appendixproof}[Proof of Proposition~\ref{prop:representing-handle-with}]
	By Definition~\ref{def:algebra-satisfies-equational-axioms}, if a joint equalizer $e : Q \to \HBundle{\Sigma}{X}$ of the following morphisms exists, then it represents $H_{\mathcal{E}}$ as $\Yoneda Q$.
	\begin{equation}
		\begin{tikzcd}
			\HBundle{\Sigma}{X} \arrow[r, shift left=1.5, "\AlgInt{\hat{M}}"] \arrow[r, shift right=1.5, "\AlgInt{\hat{N}}"'] & \KleisliExp{T_{(\Sigma', \mathcal{E}')}}{(\interpret{\hat{\Delta}} \times (\KleisliExp{T_{(\Sigma', \mathcal{E}')}}{\interpret{\hat{A}}}{X}))}{X}
		\end{tikzcd}
		\qquad \text{for each } \hat{\Delta} \vdash \hat{M} \sim \hat{N} : \hat{A} ! \Sigma\ \in\ \mathcal{E}
	\end{equation}
	Since $\Yoneda$ preserves limits, $\mathbf{handle}_{\mathcal{T}, \mathcal{T}', X}$ is equivalent to a natural transformation of the form $\Yoneda (Q \times T_{\mathcal{T}} K^{T_{\mathcal{T}'}} X) \to \Yoneda K^{T_{\mathcal{T}'}} X$.
	By the Yoneda lemma, this natural transformation corresponds to a morphism $\mathbf{handle}'_{\mathcal{T}, \mathcal{T}', X} : Q \times T_{\mathcal{T}} K^{T_{\mathcal{T}'}} X \to K^{T_{\mathcal{T}'}} X$.
	Concretely, $\mathbf{handle}'_{\mathcal{T}, \mathcal{T}', X}$ is such that $\mathbf{handle}'_{\mathcal{T}, \mathcal{T}', X} \comp \tupling{h'}{f} = \mathbf{handle}_{\mathcal{T}, \mathcal{T}', X}(e \comp h', f)$.
\end{appendixproof}

However, there is no guarantee that such limits exist in general, especially when considering term models used to prove completeness.
For this reason, we work with presheaves in the general case.

Recall that typing rules and equations for $\EffectHandlerCalculusWithEquations$ are defined by simultaneous induction.
As a consequence, the interpretation of terms cannot be defined simply by induction on type derivations.
Instead, we first define a \emph{partial} interpretation of well-typed terms and later show that the interpretation is in fact defined for all well-typed terms.
\begin{definition}\label{def:interpretation-effect-handler-calculus-with-equations}
	We partially define the interpretation of well-typed terms in $\EffectHandlerCalculusWithEquations$ by induction on type derivations.
	\[ \Gamma \vdash M : A ! \Sigma / \mathcal{E} \quad\mapsto\quad \interpret{M} : \interpret{\Gamma} \to T_{(\Sigma, \mathcal{E})} \interpret{A} \]
	Most cases are similar to $\EffectHandlerCalculus^{+}$, but we define the interpretation of handle-with as follows.
	\begin{align}
		&\interpret{\handlewithto{M}{H}{x}{N}} \ \coloneqq \\
		&\begin{cases}
			\mathbf{handle}_{\mathcal{T}, \mathcal{T}', X}(\interpret{H}, T_{\mathcal{T}} \interpret{N} \comp \strength \comp \tupling{\identity{}}{\interpret{M}}) & \interpret{H} \in H_{\mathcal{T}}(\interpret{\Gamma}) \\
			\text{undefined} & \text{otherwise}
		\end{cases}
	\end{align}
	where $\mathcal{T} = (\Sigma, \mathcal{E})$ is such that $\Gamma \vdash M : A ! \Sigma / \mathcal{E}$, and $\mathcal{T}' = (\Sigma', \mathcal{E}')$ is such that $\Gamma, x : A \vdash N : B ! \Sigma' / \mathcal{E}'$.
\end{definition}

\begin{toappendix}
Similarly to Lemma~\ref{lem:weakening-interpretation}, Definition~\ref{def:interpretation-effect-handler-calculus-with-equations} satisfies the following weakening property.
\begin{lemma}[weakening]
	Let $\Gamma, \Delta \vdash V : A$ be a well-typed term.
	If $\interpret{\Gamma, \Delta \vdash V : A}$ is defined, then the interpretation of $\Gamma, x : B, \Delta \vdash V : A$ is defined and given by
	\[ \interpret{\Gamma, x : B, \Delta \vdash V : A} = \interpret{\Gamma, \Delta \vdash V : A} \comp \mathrm{proj}_{\Gamma; x : B; \Delta} \]
	where $\mathrm{proj}_{\Gamma; x : B; \Delta} : \interpret{\Gamma, x : B, \Delta} \to \interpret{\Gamma, \Delta}$ is the canonical projection.
	We also have similar equations for computations and handlers.
\end{lemma}
\begin{proof}
	By induction on the type derivation.
	Most cases are similar to Lemma~\ref{lem:weakening-interpretation}.
	For the case of handle-with, we have the following by the naturality of $\mathbf{handle}$ and the induction hypothesis.
	\begin{align}
		&\interpret{\Gamma, x : B, \Delta \vdash \handlewithto{M}{H}{y}{N} : C} \\
		&= \mathbf{handle}(\interpret{\Gamma, x : B, \Delta \vdash H : \Sigma \Rightarrow C}, T_{\mathcal{T}} \interpret{\Gamma, x : B, \Delta, y : A \vdash N : C} \comp \strength \comp \tupling{\identity{}}{\interpret{\Gamma, x : B, \Delta \vdash M : A ! \Sigma}}) \\
		&= \mathbf{handle}(\interpret{H} \comp \mathrm{proj}_{\Gamma; x : B; \Delta}, T_{\mathcal{T}} (\interpret{N} \comp \mathrm{proj}_{\Gamma; x : B; \Delta, y : A}) \comp \strength \comp \tupling{\identity{}}{\interpret{M} \comp \mathrm{proj}_{\Gamma; x : B; \Delta}}) \\
		&= \mathbf{handle}(\interpret{H} \comp \mathrm{proj}_{\Gamma; x : B; \Delta}, T_{\mathcal{T}} \interpret{N} \comp \strength \comp \tupling{\identity{}}{\interpret{M}} \comp \mathrm{proj}_{\Gamma; x : B; \Delta}) \\
		&= \mathbf{handle}(\interpret{H}, T_{\mathcal{T}} \interpret{N} \comp \strength \comp \tupling{\identity{}}{\interpret{M}}) \comp \mathrm{proj}_{\Gamma; x : B; \Delta} \\
		&= \interpret{\Gamma, \Delta \vdash \handlewithto{M}{H}{y}{N} : C} \comp \mathrm{proj}_{\Gamma; x : B; \Delta}
	\end{align}
	Here, we used the fact that if $\interpret{H} \in H_{\mathcal{T}}(\interpret{\Gamma, \Delta})$, then $\interpret{H} \comp \mathrm{proj}_{\Gamma; x : B; \Delta} \in H_{\mathcal{T}}(\interpret{\Gamma, x : B, \Delta})$.
\end{proof}

We also have the following substitution property similar to Lemma~\ref{lem:subst-interpretation}.
\begin{lemma}[substitution]
	For $\sigma : \Gamma \to \Delta$ and $\Delta \vdash W : B$, if $\interpret{\sigma} : \interpret{\Gamma} \to \interpret{\Delta}$ and $\interpret{W} : \interpret{\Delta} \to \interpret{B}$ are defined, then the interpretation of $\Gamma \vdash W[\sigma] : B$ is defined and given by
	\[ \interpret{W[\sigma]} = \interpret{W} \comp \interpret{\sigma}. \]
	We also have similar equations for computation terms and handlers.
\end{lemma}
\begin{proof}
	By induction on the type derivation.
	Most cases are similar to Lemma~\ref{lem:subst-interpretation}.
	For the case of handle-with, we use the naturality of $\mathbf{handle}$.
\end{proof}
\end{toappendix}

To show that $\Gamma \vdash H : \Sigma \Rightarrow C\ \mathbf{respects}\ \mathcal{E}$ is interpreted as $\interpret{H} \in H_{\mathcal{T}}(\interpret{\Gamma})$, we need the following lemma.
\begin{lemma}\label{lem:interpretation-of-handle-with}
	Let $\hat{\Delta} \vdash \hat{M} : \hat{A} ! \Sigma$.
	If $\interpret{H} : \interpret{\Gamma} \to \HBundle{\interpret{\Sigma}}{\interpret{C}}$ is defined, then $\interpret{\hat{M}[H, k]} : \interpret{\Gamma, \hat{\Delta}, k : \hat{A} \to C} \to \interpret{C}$ is defined and equal to $\Lambda^{-1}(\AlgInt{\hat{M}} \comp \interpret{H}) \comp s$ where $s$ is the canonical isomorphism $s : \interpret{\Gamma, \hat{\Delta}, k : \hat{A} \to C} \cong \interpret{\Gamma} \times \interpret{\hat{\Delta}, k : \hat{A} \to C}$.
	Specifically, we have $\interpret{\hat{M}[H, k]} = \interpret{\hat{N}[H, k]}$ if and only if $\AlgInt{\hat{M}} \comp \interpret{H} = \AlgInt{\hat{N}} \comp \interpret{H}$.
	\qed
\end{lemma}
\begin{appendixproof}[Proof of Lemma~\ref{lem:interpretation-of-handle-with}]
	We prove the former statement by induction on $\hat{M}$.
	The latter statement immediately follows from the former.
	\begin{itemize}
		\item If $\hat{M} = \return{\hat{V}}$:
		\begin{align}
			\interpret{\hat{M}[H, k]} &= \interpret{k\ \hat{V}} \\
			&= \eval \comp \tupling{\pi_2 \comp \pi_2 \comp s}{\interpret{\hat{V}} \comp \pi_1 \comp \pi_2 \comp s} \\
			&= \eval \comp \tupling{\pi_2}{\interpret{\hat{V}} \comp \pi_1} \comp \pi_2 \comp s \\
			&= \eval \comp \braiding \comp (\interpret{\hat{V}} \times \identity{}) \comp \pi_2 \comp s \\
			\Lambda^{-1}(\AlgInt{\hat{M}} \comp \interpret{H}) &= \Lambda^{-1}\AlgInt{\hat{M}} \comp (\interpret{H} \times \identity{}) \\
			&= \eval \comp \braiding \comp (\interpret{\hat{V}} \times \identity{}) \comp \pi_2 \comp (\interpret{H} \times \identity{}) \\
			&= \eval \comp \braiding \comp (\interpret{\hat{V}} \times \identity{}) \comp \pi_2
		\end{align}
		\item If $\hat{M} = \letin{x}{\mathtt{op}(\hat{V})}{\hat{N}}$:
		\begin{align}
			\interpret{\hat{M}[H, k]} &= \interpret{M_{\mathtt{op}}[\hat{V}/x, \lambda x. \hat{N}[H, k]/k']} \\
			&= \interpret{M_{\mathtt{op}}} \comp \tupling{\tupling{\pi_1 \comp s}{\interpret{\hat{V}} \comp \pi_1 \comp \pi_2 \comp s}}{\Lambda (\interpret{\hat{N}[H, k]} \comp \associator^{-1} \comp (\identity{} \times \braiding) \comp \associator)} \\
			&= \interpret{M_{\mathtt{op}}} \comp \tupling{\tupling{\pi_1}{\interpret{\hat{V}} \comp \pi_1 \comp \pi_2}}{\Lambda (\Lambda^{-1} (\AlgInt{\hat{N}}^T_X \comp \interpret{H}) \comp w)} \comp s \\
			\Lambda^{-1}(\AlgInt{\hat{M}} \comp \interpret{H}) &= \eval \comp \tupling{\pi_{\mathtt{op}} \comp \pi_1}{\tupling{\interpret{\hat{V}} \comp \pi_1 \comp \pi_2}{\Lambda (\Lambda^{-1} \AlgInt{\hat{N}}^T_X \comp w)}} \comp (\interpret{H} \times \identity{}) \\
			&= \eval \comp \tupling{\pi_{\mathtt{op}} \comp \interpret{H} \comp \pi_1}{\tupling{\interpret{\hat{V}} \comp \pi_1 \comp \pi_2}{\Lambda (\Lambda^{-1} (\AlgInt{\hat{N}}^T_X \comp \interpret{H}) \comp w)}} \\
			&= \eval \comp \tupling{\Lambda (\interpret{M_{\mathtt{op}}} \comp \associator^{-1}) \comp \pi_1}{\tupling{\interpret{\hat{V}} \comp \pi_1 \comp \pi_2}{\Lambda (\Lambda^{-1} (\AlgInt{\hat{N}}^T_X \comp \interpret{H}) \comp w)}} \\
			&= \interpret{M_{\mathtt{op}}} \comp \associator^{-1} \comp \tupling{\pi_1}{\tupling{\interpret{\hat{V}} \comp \pi_1 \comp \pi_2}{\Lambda (\Lambda^{-1} (\AlgInt{\hat{N}}^T_X \comp \interpret{H}) \comp w)}} \\
			&= \interpret{M_{\mathtt{op}}} \comp \tupling{\tupling{\pi_1}{\interpret{\hat{V}} \comp \pi_1 \comp \pi_2}}{\Lambda (\Lambda^{-1} (\AlgInt{\hat{N}}^T_X \comp \interpret{H}) \comp w)}
		\end{align}
		\item If $\hat{M} = \caseof{\hat{V}}{\casepattern{\iota_i\ x_i}{\hat{M}_i}}_{i = 1}^n$:
		\begin{align}
			\interpret{\hat{M}[H, k]} &= \interpret{\caseof{\hat{V}}{\casepattern{\iota_i\ x_i}{\hat{M}_i[H, k]}}_{i = 1}^n} \\
			&= [\interpret{\hat{M}_1[H, k]} \comp \associator^{-1} \comp (\identity{} \times \braiding) \comp \associator, \dots, \interpret{\hat{M}_n[H, k]} \comp \associator^{-1} \comp (\identity{} \times \braiding) \comp \associator] \\
			&\qquad \comp \tupling{\identity{}}{\interpret{\hat{V}} \comp \pi_1 \comp \pi_2 \comp s} \\
			&= [\Lambda^{-1}(\AlgInt{\hat{M}_1} \comp \interpret{H}) \comp w \comp (s \times \identity{}), \dots, \Lambda^{-1}(\AlgInt{\hat{M}_n} \comp \interpret{H}) \comp w \comp (s \times \identity{})] \\
			&\qquad \comp \tupling{\identity{}}{\interpret{\hat{V}} \comp \pi_1 \comp \pi_2 \comp s} \\
			&= [\Lambda^{-1}(\AlgInt{\hat{M}_1} \comp \interpret{H}) \comp w, \dots, \Lambda^{-1}(\AlgInt{\hat{M}_n} \comp \interpret{H}) \comp w] \\
			&\qquad \comp (s \times \identity{}) \comp \tupling{\identity{}}{\interpret{\hat{V}} \comp \pi_1 \comp \pi_2 \comp s} \\
			&= [\Lambda^{-1}(\AlgInt{\hat{M}_1} \comp \interpret{H}) \comp w, \dots, \Lambda^{-1}(\AlgInt{\hat{M}_n} \comp \interpret{H}) \comp w] \comp \tupling{\identity{}}{\interpret{\hat{V}} \comp \pi_1 \comp \pi_2} \comp s \\
			&= [\Lambda^{-1}\AlgInt{\hat{M}_1} \comp w, \dots, \Lambda^{-1}\AlgInt{\hat{M}_n} \comp w] \comp \tupling{\identity{}}{\interpret{\hat{V}} \comp \pi_1 \comp \pi_2} \comp (\interpret{H} \times \identity{}) \comp s \\
			&= \Lambda^{-1}\AlgInt{\hat{M}} \comp (\interpret{H} \times \identity{}) \comp s \\
			&= \Lambda^{-1}(\AlgInt{\hat{M}} \comp \interpret{H}) \comp s
			\qedhere
		\end{align}
	\end{itemize}
\end{appendixproof}

\begin{toappendix}
\begin{lemma}\label{lem:subst-interpretation-effect-theories}
	Suppose that we have $\Delta \vdash \hat{M} : \hat{A} ! \Sigma$ and $\sigma : \Gamma \to \hat{\Delta}$.
	If $\interpret{\sigma}$ is defined, then $\interpret{\hat{M}[\sigma]} = \interpret{\hat{M}} \comp \interpret{\sigma}$.
\end{lemma}
\begin{proof}
	It is straightforward to show that $\interpret{\hat{V}[\sigma]} = \interpret{\hat{V}} \comp \interpret{\sigma}$ holds for any $\hat{V}$.
	Then, we prove the statement by induction on $\hat{M}$.
	\begin{itemize}
		\item If $\hat{M} = \return{\hat{V}}$:
		\begin{align}
			\interpret{\hat{M}[\sigma]} &= \interpret{\return{\hat{V}[\sigma]}} \\
			&= \eta \comp \interpret{\hat{V}[\sigma]} \\
			&= \eta \comp \interpret{\hat{V}} \comp \interpret{\sigma} \\
			&= \interpret{\hat{M}} \comp \interpret{\sigma}
		\end{align}
		\item If $\hat{M} = \letin{x}{\mathtt{op}(\hat{V})}{\hat{N}}$:
		\begin{align}
			\interpret{\hat{M}[\sigma]} &= \interpret{\letin{x}{\mathtt{op}(\hat{V}[\sigma])}{\hat{N}[\sigma]}} \\
			&= \mu \comp T (\interpret{\hat{N}} \comp (\interpret{\sigma} \times \identity{})) \comp \strength \comp \tupling{\identity{}}{\interpret{\mathtt{op}} \comp \interpret{\hat{V}} \comp \interpret{\sigma}} \\
			&= \mu \comp T \interpret{\hat{N}} \comp \strength \comp \tupling{\identity{}}{\interpret{\mathtt{op}} \comp \interpret{\hat{V}}} \comp \interpret{\sigma} \\
			&= \interpret{\hat{M}} \comp \interpret{\sigma}
		\end{align}
		\item If $\hat{M} = \caseof{\hat{V}}{\casepattern{\iota_i\ x_i}{\hat{M}_i}}_{i = 1}^n$:
		\begin{align}
			\interpret{\hat{M}[\sigma]} &= \interpret{\caseof{\hat{V}[\sigma]}{\casepattern{\iota_i\ x_i}{\hat{M}_i[\sigma]}}_{i = 1}^n} \\
			&= [\interpret{\hat{M}_1[\sigma]}, \dots, \interpret{\hat{M}_n[\sigma]}] \comp \tupling{\identity{}}{\interpret{\hat{V}[\sigma]}} \\
			&= [\interpret{\hat{M}_1} \comp (\interpret{\sigma} \times \identity{}), \dots, \interpret{\hat{M}_n} \comp (\interpret{\sigma} \times \identity{})] \comp \tupling{\identity{}}{\interpret{\hat{V}} \comp \interpret{\sigma}} \\
			&= [\interpret{\hat{M}_1}, \dots, \interpret{\hat{M}_n}] \comp \tupling{\identity{}}{\interpret{\hat{V}}} \comp \interpret{\sigma} \\
			&= \interpret{\hat{M}} \comp \interpret{\sigma}
			\qedhere
		\end{align}
	\end{itemize}
\end{proof}
\end{toappendix}

After establishing necessary lemmas using the partial interpretation, we can show that the interpretation is in fact defined for all well-typed terms.
This is proved by simultaneous induction with the soundness theorem for $\EffectHandlerCalculusWithEquations$.

\begin{proposition}\label{prop:interpretation-with-effect-theories}
	If $\Gamma \vdash M : A ! \Sigma / \mathcal{E}$, then $\interpret{M} : \interpret{\Gamma} \to T_{(\Sigma, \mathcal{E})} \interpret{A}$ is defined.
	Similarly for value terms and handlers.
\end{proposition}
\begin{appendixproof}[Proof of Proposition~\ref{prop:interpretation-with-effect-theories}]
	By induction on type derivations simultaneously with Theorem~\ref{thm:soundness-with-effect-theories}.
	Most cases are the same as in $\EffectHandlerCalculus^{+}$.
	The only non-trivial case is the case of handle-with.
	By Lemma~\ref{lem:interpretation-of-handle-with} and IH, if $\Gamma \vdash H : \Sigma \Rightarrow C\ \mathbf{respects}\ \mathcal{E}$, then we have $\interpret{H} \in H_{\mathcal{T}}(\interpret{\Gamma})$.
	Thus, $\interpret{\handlewithto{M}{H}{x}{N}}$ is defined.
\end{appendixproof}

\begin{theorem}\label{thm:soundness-with-effect-theories}
	If $\Gamma \vdash M = N : A ! \Sigma / \mathcal{E}$, then $\interpret{M}$ and $\interpret{N}$ are defined, and $\interpret{M} = \interpret{N}$ holds.
\end{theorem}
\begin{proofsketch}
	Proposition~\ref{prop:interpretation-with-effect-theories} and Theorem~\ref{thm:soundness-with-effect-theories} are proved by simultaneous induction.
\end{proofsketch}
\begin{appendixproof}[Proof of Theorem~\ref{thm:soundness-with-effect-theories}]
	By simultaneous induction on the derivation of $\Gamma \vdash M = N : A ! \Sigma / \mathcal{E}$ with Proposition~\ref{prop:interpretation-with-effect-theories}.
	Most cases are the same as in $\EffectHandlerCalculus^{+}$.
	The non-trivial cases are the case for handle-with equations (Figure~\ref{fig:effect-handler-equations}) and the case for instantiating equational axioms.
	The latter follows from Lemma~\ref{lem:subst-interpretation-effect-theories}.
	For the former, we can show the equations by using the equations for $\mathbf{handle}$ in the definition of $\EffectHandlerCalculusWithEquations$-model.
	\begin{itemize}
		\item \eqref{eq:handle-return}:
		\begin{align}
			&\interpret{\handlewithto{\return{V}}{H}{x}{M}} \\
			&= \mathbf{handle}(\interpret{H}, T_{\mathcal{T}} \interpret{M} \comp \strength \comp \tupling{\identity{}}{\eta \comp \interpret{V}}) \\
			&= \mathbf{handle}(\interpret{H}, \eta \comp \interpret{M} \comp \tupling{\identity{}}{\interpret{V}}) \\
			&= \interpret{M} \comp \tupling{\identity{}}{\interpret{V}} \\
			&= \interpret{M[V/x]}
		\end{align}
		\item \eqref{eq:handle-let}:
		\begin{align}
			&\interpret{\handlewithto{\letin{x}{L}{M}}{H}{y}{N}} \\
			&= \mathbf{handle}(\interpret{H}, T_{\mathcal{T}} \interpret{N} \comp \strength \comp \tupling{\identity{}}{\mu \comp T \interpret{M} \comp \strength \comp \tupling{\identity{}}{\interpret{L}}}) \\
			&= \mathbf{handle}(\interpret{H}, \mu \comp T_{\mathcal{T}}^2 \interpret{N} \comp T \strength \comp \strength \comp \tupling{\identity{}}{T \interpret{M} \comp \strength \comp \tupling{\identity{}}{\interpret{L}}}) \\
			&= \mathbf{handle} (\interpret{H}, T \mathbf{handle} (\interpret{H} \comp \pi_1, \pi_2) \comp \strength \comp \tupling{\identity{}}{T_{\mathcal{T}}^2 \interpret{N} \comp T \strength \comp \strength \comp \tupling{\identity{}}{T \interpret{M} \comp \strength \comp \tupling{\identity{}}{\interpret{L}}}}) \\
			&= \mathbf{handle} (\interpret{H}, T \mathbf{handle} (\interpret{H} \comp \pi_1, T_{\mathcal{T}} \interpret{N} \comp \strength \comp \pi_2) \comp \strength \comp \tupling{\identity{}}{\strength \comp \tupling{\identity{}}{T \interpret{M} \comp \strength \comp \tupling{\identity{}}{\interpret{L}}}}) \\
			&= \mathbf{handle} (\interpret{H}, T \mathbf{handle} (\interpret{H} \comp \pi_1, T_{\mathcal{T}} \interpret{N} \comp \strength \comp (\identity{} \times \interpret{M}) \comp \pi_2) \comp \strength \comp \tupling{\identity{}}{\strength \comp \tupling{\identity{}}{\strength \comp \tupling{\identity{}}{\interpret{L}}}})
		\end{align}
		Here, we use the following equations.
		\begin{align}
			&\strength \comp \tupling{\identity{}}{\strength \comp \tupling{\identity{}}{\strength \comp \tupling{\identity{}}{\interpret{L}}}} \\
			&= \strength \comp (\identity{} \times \strength) \comp \associator \comp \tupling{\tupling{\identity{}}{\identity{}}}{\strength \comp \tupling{\identity{}}{\interpret{L}}} \\
			&= T \associator \comp \strength \comp \tupling{\tupling{\identity{}}{\identity{}}}{\strength \comp \tupling{\identity{}}{\interpret{L}}} \\
			&= T \tupling{\pi_1}{\identity{}} \comp \strength \comp \tupling{\identity{}}{\strength \comp \tupling{\identity{}}{\interpret{L}}} \\
			&= T \tupling{\pi_1}{\identity{}} \comp T \tupling{\pi_1}{\identity{}} \comp \strength \comp \tupling{\identity{}}{\interpret{L}}
		\end{align}
		Continuing the calculation:
		\begin{align}
			&\interpret{\handlewithto{\letin{x}{L}{M}}{H}{y}{N}} \\
			&= \mathbf{handle} (\interpret{H}, T \mathbf{handle} (\interpret{H} \comp \pi_1, T_{\mathcal{T}} \interpret{N} \comp \strength \comp (\identity{} \times \interpret{M}) \comp \pi_2) \comp T \tupling{\pi_1}{\identity{}} \comp T \tupling{\pi_1}{\identity{}} \comp \strength \comp \tupling{\identity{}}{\interpret{L}}) \\
			&= \mathbf{handle} (\interpret{H}, T \mathbf{handle} (\interpret{H} \comp \pi_1, T_{\mathcal{T}} \interpret{N} \comp \strength \comp (\identity{} \times \interpret{M}) \comp \tupling{\pi_1}{\identity{}}) \comp \strength \comp \tupling{\identity{}}{\interpret{L}}) \\
			&= \mathbf{handle} (\interpret{H}, T \mathbf{handle} (\interpret{H} \comp \pi_1, T_{\mathcal{T}} (\interpret{N} \comp (\pi_1 \times \identity{})) \comp \strength \comp \tupling{\identity{}}{\interpret{M}})\comp \strength \comp \tupling{\identity{}}{\interpret{L}}) \\
			&= \interpret{\handlewithto{L}{H}{x}{\handlewithto{M}{H}{y}{N}}}
		\end{align}
		\item \eqref{eq:handle-op}: We use Lemma~\ref{lem:handle-op-aux}.
		\begin{align}
			&\interpret{\handlewithto{\mathtt{op}(V)}{H}{x}{M}} \\
			&= \mathbf{handle}(\interpret{H}, T_{\mathcal{T}} \interpret{M} \comp \strength \comp \tupling{\identity{}}{\interpret{\mathtt{op}} \comp \interpret{V}}) \\
			&= \mathbf{handle}(\interpret{H}, \MixedAlgOp{E'}{X}{\interpret{\mathtt{op}}} \comp \tupling{\interpret{V}}{\Lambda \interpret{M}}) \\
			&= \eval \comp \tupling{\pi_{\mathtt{op}} \comp \interpret{H}}{\tupling{\interpret{V}}{\Lambda \interpret{M}}} \\
			&= \eval \comp \tupling{\Lambda (\interpret{M_{\mathtt{op}}} \comp \associator^{-1})}{\tupling{\interpret{V}}{\Lambda \interpret{M}}} \\
			&= \interpret{M_{\mathtt{op}}} \comp \associator^{-1} \comp \tupling{\identity{}}{\tupling{\interpret{V}}{\Lambda \interpret{M}}} \\
			&= \interpret{M_{\mathtt{op}}[V/x, \lambda x. M/k]}
			\qedhere
		\end{align}
	\end{itemize}
\end{appendixproof}

By constructing the term model of $\EffectHandlerCalculusWithEquations$, we can also show the completeness theorem.
\begin{theorem}[completeness]\label{thm:completeness-with-effect-theories}
	If $\interpret{M} = \interpret{N}$, then $\Gamma \vdash M = N : A ! E / \mathcal{E}$.
	\qed
\end{theorem}
\begin{appendixproof}[Proof of Theorem~\ref{thm:completeness-with-effect-theories}]
	We prove the completeness by constructing the term model of $\EffectHandlerCalculusWithEquations$.
	We define the category $\TermModel{\EffectHandlerCalculusWithEquations}$ whose objects are value types and morphisms are value terms modulo the equational theory as in the case of $\EffectHandlerCalculus$ and $\EffectHandlerCalculus^{+}$.
	We also define the strong monad $T_{(\Sigma, \mathcal{E})}$ by $T_{(\Sigma, \mathcal{E})} A \coloneqq \UnitType \to A ! \Sigma / \mathcal{E}$.
	Then, this gives a model of $\EffectHandlerCalculusWithEquations$.
	Most parts of the proof are similar to the term model of $\EffectHandlerCalculus^{+}$.
	The non-trivial parts are (1) to show that we can define $\mathbf{handle}$ satisfying the required equations and (2) to show that the interpretation of operations satisfies equational axioms.
	For the former, given $x : A \vdash H : \prod_{\mathtt{op} : A_{\mathtt{op}} \rightarrowtriangle B_{\mathtt{op}} \in \Sigma} (A_{\mathtt{op}} \times (B_{\mathtt{op}} \to B ! \Sigma' / \mathcal{E'}) \to B ! \Sigma' / \mathcal{E'})$ such that $H \in H_{(\Sigma, \mathcal{E})}(A)$ and $x : A \vdash V : T_{(\Sigma, \mathcal{E})} T_{(\Sigma', \mathcal{E}')} B$, we define $x : A \vdash \mathbf{handle}(H, V) : T_{(\Sigma', \mathcal{E}')} B$ as follows.
	\[ \mathbf{handle}(H, V) \quad\coloneqq\quad \lambda z. \handlewithto{V\ \langle\rangle}{H}{y}{y\ \langle\rangle} \]
	Here, we identify $H$ with the effect handler $\{ \mathtt{op}(x, k) \mapsto \pi_{\mathtt{op}}\ H\ \langle x, k \rangle \mid \mathtt{op} \in \Sigma \}$ induced by $H$.
	By definition of $H \in H_{(\Sigma, \mathcal{E})}(A)$, we have $\AlgInt{\hat{M}} \comp H = \AlgInt{\hat{N}} \comp H$ for each equational axiom $\hat{\Delta} \vdash \hat{M} \sim \hat{N} : \hat{A} ! \Sigma$ in $\mathcal{E}$.
	In $\TermModel{\EffectHandlerCalculusWithEquations}$, $h : \HBundle{\Sigma}{C} \vdash \AlgInt{\hat{M}} : (\prod \hat{\Delta} \times (\hat{A} \to C)) \to C$ is given as follows.
	\begin{align}
		\AlgInt{\return{\hat{V}}} &= \lambda \langle \tilde{x}, k \rangle. k\ \hat{V} \\
		\AlgInt{\letin{y}{\mathtt{op}(\hat{V})}{\hat{M}}} &= \lambda \langle \tilde{x}, k \rangle. \pi_{\mathtt{op}}\ h\ \langle \hat{V}, \lambda y. \AlgInt{\hat{M}}\ \langle \tilde{x}, y, k \rangle \rangle \\
		\AlgInt{\caseof{\hat{V}}{\casepattern{\iota_i\ x_i}{\hat{M}_i}}_{i = 1}^m} &= \lambda \langle \tilde{x}, k \rangle. \caseof{\hat{V}}{\casepattern{\iota_i\ x_i}{\AlgInt{\hat{M}_i}\ \langle \tilde{x}, x_i, k \rangle}}_{i = 1}^m
	\end{align}
	Here, we let $\hat{\Delta} = x_1 : \hat{B}_1, \dots, x_n : \hat{B}_n$ and $\tilde{x} = \langle x_1, \dots, x_n \rangle$.
	This implies that $\AlgInt{\hat{M}} \comp H$ is given by $\lambda \langle \tilde{x}, k \rangle. \hat{M}[H, k]$, and hence the corresponding effect handler respects $\mathcal{E}$.
	It is straightforward to show that the interpretation $\interpret{\handlewithto{M}{H}{x}{N}}$ in $\TermModel{\EffectHandlerCalculusWithEquations}$ is the equivalence class of the term $\lambda z : \UnitType. \handlewithto{M}{H}{x}{N}$.
	\begin{align}
		&\interpret{\handlewithto{M}{H}{x}{N}} \\
		&= \mathbf{handle}(\interpret{H}, T_{(\Sigma, \mathcal{E})} \interpret{N} \comp \strength \comp \tupling{\identity{}}{\interpret{M}}) \\
		&= \lambda z. \handlewithto{\letin{x}{M}{\return{\lambda z. N}}}{H}{y}{y\ \langle\rangle} \\
		&= \lambda z. \handlewithto{M}{H}{y}{N}
	\end{align}
	Now, we show that $\mathbf{handle}$ satisfies the required equations.
	\begin{itemize}
		\item Naturality:
			\begin{align}
				&\mathbf{handle}(H, V) \comp W \\
				&= (\lambda z. \handlewithto{V\ \langle\rangle}{H}{y}{y\ \langle\rangle})[W/x] \\
				&= \lambda z. \handlewithto{V[W/x]\ \langle\rangle}{H[W/x]}{y}{y\ \langle\rangle} \\
				&= \mathbf{handle}(H \comp W, V \comp W)
			\end{align}
		\item Unit~\eqref{eq:effect-theory-handle-unit}:
		\begin{align}
			&\mathbf{handle}_{\mathcal{T}, \mathcal{T}', X} (H, \eta \comp V) \\
			&= \lambda z. \handlewithto{\return{V}}{H}{y}{y\ \langle\rangle} \\
			&= \lambda z. V\ \langle\rangle \\
			&= \lambda z. V\ z \\
			&= V
		\end{align}
		\item Multiplication~\eqref{eq:effect-theory-handle-multiplication}:
		\begin{align}
			&\mathbf{handle}_{\mathcal{T}, \mathcal{T}', X} (H, \mu \comp V) \\
			&= \lambda z. \handlewithto{\letin{y_1}{V\ \langle\rangle}{y_1\ \langle\rangle}}{H}{y_2}{y_2\ \langle\rangle} \\
			&= \lambda z. \handlewithto{V\ \langle\rangle}{H}{y_1}{\handlewithto{y_1\ \langle\rangle}{H}{y_2}{y_2\ \langle\rangle}} \\
			&\mathbf{handle}_{\mathcal{T}, \mathcal{T}', X} (H \comp \pi_1, \pi_2) \\
			&= \lambda z. \handlewithto{\pi_2\ x\ \langle\rangle}{H[\pi_1\ x/x]}{y}{y\ \langle\rangle} \\
			&T \mathbf{handle}_{\mathcal{T}, \mathcal{T}', X} (H \comp \pi_1, \pi_2) \comp \strength \comp \tupling{\identity{}}{V} \\
			&= \lambda z. \letin{y}{V\ \langle\rangle}{\return{(\lambda z. \mathbf{handle}_{\mathcal{T}, \mathcal{T}', X} (H \comp \pi_1, \pi_2))[\langle x, y \rangle / x]}} \\
			&= \lambda z. \letin{y}{V\ \langle\rangle}{\return{\lambda z. \handlewithto{y\ \langle\rangle}{H}{y}{y\ \langle\rangle}}} \\
			&\mathbf{handle}_{\mathcal{T}, \mathcal{T}', X} (H, T \mathbf{handle}_{\mathcal{T}, \mathcal{T}', X} (H \comp \pi_1, \pi_2) \comp \strength \comp \tupling{\identity{}}{V}) \\
			&= \lambda z. \handlewithto{(\letin{y}{V\ \langle\rangle}{\return{\lambda z. \handlewithto{y\ \langle\rangle}{H}{y}{y\ \langle\rangle}}})}{H}{y}{y\ \langle\rangle} \\
			&= \lambda z. \handlewithto{V\ \langle\rangle}{H}{y}{\handlewithto{(\return{\lambda z. \handlewithto{y\ \langle\rangle}{H}{y}{y\ \langle\rangle}})}{H}{y}{y\ \langle\rangle}} \\
			&= \lambda z. \handlewithto{V\ \langle\rangle}{H}{y}{\handlewithto{y\ \langle\rangle}{H}{y}{y\ \langle\rangle}} \\
		\end{align}
		\item Operation~\eqref{eq:effect-theory-handle-operation}:
		\begin{align}
			&\mathbf{handle}_{\mathcal{T}, \mathcal{T}', X} (H, \MixedAlgOp{\mathcal{T}'}{X}{\interpret{\mathtt{op}}^T_{\mathcal{T}}} \comp \langle V, W \rangle) \\
			&= \lambda z. \handlewithto{\letin{y}{\mathtt{op}\ V}{\return{\lambda z. W\ y}}}{H}{y}{y\ \langle\rangle} \\
			&= \lambda z. \handlewithto{\mathtt{op}\ V}{H}{y}{\handlewithto{(\return{\lambda z. W\ y})}{H}{y}{y\ \langle\rangle}} \\
			&= \lambda z. \handlewithto{\mathtt{op}\ V}{H}{y}{W\ y} \\
			&= \lambda z. M_{\mathtt{op}}[V/x, W/k] \\
			&= \eval \comp \tupling{\pi_{\mathtt{op}} \comp H}{\langle V, W \rangle}
		\end{align}
	\end{itemize}
	Lastly, we show that the interpretation $\interpret{\mathtt{op}}$ of operations satisfies equational axioms.
	For each effect term $\hat{\Delta} \vdash \hat{M} : \hat{A} ! \Sigma$, its interpretation $\interpret{\hat{M}}$ in $\TermModel{\EffectHandlerCalculusWithEquations}$ is given by the equivalence class of the computation term $\lambda z. \hat{M}$.
	Thus, for each equational axiom $\hat{\Delta} \vdash \hat{M} \sim \hat{N} : \hat{A} ! \Sigma$ in $\mathcal{E}$, we have $\interpret{\hat{M}} = \interpret{\hat{N}}$ by \hyperlink{rule:Inst-Ax}{\textsc{Inst-Ax}}.
\end{appendixproof}

\begin{example}
	We show that traditional free model monads gives a $\EffectHandlerCalculusWithEquations$-model.
	For simplicity, consider the category $\Set$.
	Assuming that all base types are interpreted as countable sets, any ground type is also interpreted as a countable set.
	Let $T_{(\Sigma, \mathcal{E})}$ be the free model monad of the effect theory $(\Sigma, \mathcal{E})$, i.e., $T_{(\Sigma, \mathcal{E})} X$ is the underlying set of a free $(\Sigma, \mathcal{E})$-algebra generated by $X \in \Set$.
	Then, this naturally gives a $\EffectHandlerCalculusWithEquations$-model.
	Here, by Proposition~\ref{prop:representing-handle-with}, $\mathbf{handle}$ is defined by a function $\mathbf{handle}'_{\mathcal{T}, \mathcal{T}', X} : Q \times T_{\mathcal{T}} K^{T_{\mathcal{T}'}} X \to K^{T_{\mathcal{T}'}} X$ where $Q$ is the set of $(\Sigma, \mathcal{E})$-algebras on $K^{T_{\mathcal{T}'}} X$.
	By the freeness of $T_{\mathcal{T}}$, we can define $\mathbf{handle}'_{\mathcal{T}, \mathcal{T}', X}(h, {-}) : T_{\mathcal{T}} K^{T_{\mathcal{T}'}} X \to K^{T_{\mathcal{T}'}} X$ as the Eilenberg-Moore algebra corresponding to $h \in Q$.
	It is straightforward to show that this $\mathbf{handle}$ satisfies the required equations.
	\qed
\end{example}

\section{Related Work}
\label{sec:related-work}
\paragraph{Categorical semantics of effect handlers}
An early categorical semantics of effect handlers is given by free model monads for algebraic theories \cite{PlotkinESOP2009,PlotkinLMCS2013,LuksicJFunctProg2020}, building on the theory of algebraic and Lawvere theories \cite{PlotkinFoSSaCS2002,PlotkinApplCategStruct2003,HylandENTCS2007}.
Subsequent work has often adopted free monads generated by algebraic operations without equational axioms as semantic models \cite{BauerLMCS2014,KiselyovJFunctProg2021}.
Since free monads $T$ of a functor $F$ is given as the least fixed point $T X = \mu Y. X + F Y$, such models are typically formulated in domain-theoretic settings.

A denotational model based on realizability is proposed in~\cite{YangPOPL2026}.
Their effect handler calculus differs from ours in that it is based on \emph{raw monads}, which are not required to satisfy the monad laws, resulting in an equational theory different from ours.
In particular, they explicitly exclude \eqref{eq:handle-op} from their equational theory.
As a result, their models and ours are incomparable.

While these works typically study equational theories and establish soundness results, no completeness results have been obtained.

\paragraph{Variants of effect handler calculi}
Several variants of effect handlers have been studied in the literature.
Here, we compare our calculus with some of them.
In this paper, we focus on \emph{deep effect handlers}, as opposed to \emph{shallow effect handlers}~\cite{HillerstromAPLAS2018}.
Deep and shallow effect handlers can simulate each other using recursion~\cite[Section~3]{HillerstromAPLAS2018}, which suggests that our categorical models could potentially be adapted to shallow effect handlers.
We leave a detailed investigation of shallow effect handlers to future work.

Our effect system is based on~\cite{BauerLMCS2014}; however, for simplicity, we omit \emph{subtyping} induced by the subset relation $\Sigma \subseteq \Sigma'$ between signatures.
We believe that subtyping can be incorporated into our framework without much difficulty by introducing monad morphisms $T_{\Sigma} \to T_{\Sigma'}$ for $\Sigma \subseteq \Sigma'$.
Some effect handler calculi supports \emph{answer type modification}~\cite{SekiyamaPOPL2023,KawamataPOPL2024}, which allows the type system to track changes of the answer type.
From a semantic perspective, this would correspond to parameterized monads~\cite[Section~3.2.1]{AtkeyJFunctProg2009}.

\emph{Higher-order effects} are known to exhibit ``non-algebraic'' behaviors~\cite{WuHaskell2014}, and calculi for such effects have been studied~\cite{BosmanLMCS2024,YangPOPL2026}.
Here, higher-order effects refer to operations whose input and output types may involve non-ground types.
In our effect handler calculus $\EffectHandlerCalculus$, signatures for operations are \emph{not} restricted to ground types, which potentially gives rise to a mutual dependence between semantics of computations and operation signatures.
We address this issue by defining the semantics using monads indexed by \emph{semantic} signatures, thereby collapsing this potentially circular dependence.
By contrast, in $\EffectHandlerCalculusWithEquations$, we restrict attention to signatures with ground types, so that effect theories do not involve higher-order effects.

\section{Conclusions and Future Work}
\label{sec:conclusion}

We have presented sound and complete categorical semantics for two variants of effect handler calculi, with and without equational axioms for operations.
We have also illustrated the expressiveness of our framework by showing that both classical free model monads and CPS semantics can be captured as models of our calculi.

As future work, we plan to investigate fibrational logical relations for effect handlers.
Specifically, we would like to extend categorical $\top\top$-lifting~\cite{KatsumataCSL2005,KatsumataInfComput2013} to our setting.
We expect that this will require solving mixed-variance fixed-point equations in order to lift the $\mathbf{handle}$ morphism from the base category to the total category of a fibration, potentially using classical domain-theoretic techniques or step-indexing.

\bibliographystyle{ACM-Reference-Format}
\bibliography{effect-handler}

\end{document}